\documentclass[10pt, twocolumn, twoside, english]{IEEEtran}

\usepackage{cite}

\ifCLASSINFOpdf
   \usepackage[pdftex]{graphicx}
\else
  \usepackage[dvips]{graphicx}
   \DeclareGraphicsExtensions{.eps}
\fi
\usepackage{subfigure}
\usepackage{stfloats}
\usepackage{xcolor}

\usepackage{amsmath}
\usepackage{amsfonts,helvet}
\usepackage{amssymb}
\usepackage{amsthm}
\usepackage{stfloats}
\usepackage{optidef}

\usepackage{algorithm}
\usepackage{algpseudocode}

\usepackage{array}
\usepackage{multirow}
\usepackage{threeparttable}

\usepackage{hyperref}
\usepackage[nolist , nohyperlinks]{acronym}
\usepackage{siunitx}

\newtheorem{theorem}{Theorem}
\newtheorem{lemma}{Lemma}
\newtheorem{remark}{Remark}

\newcommand{\ie}{i.e.}
\newcommand{\eg}{e.g.}
\newcommand{\fig}{Fig.\,}

\newcommand{\tbl}{Table\,}

\newcommand{\ceil}[1]{\lceil #1 \rceil}
\newcommand{\floor}[1]{\lfloor #1 \rfloor}

\begin{document}
\title{Phase-Noise Compensation for OFDM Systems Exploiting Coherence Bandwidth:\\ Modeling, Algorithms, and Analysis}
\author{MinKeun~Chung,~\IEEEmembership{Member,~IEEE,}
        Liang~Liu,~\IEEEmembership{Member,~IEEE,}
        and~Ove~Edfors,~\IEEEmembership{Senior~Member,~IEEE}
\thanks{The authors are with the Department of Electrical and Information Technology, Lund University, Lund, 221 00, Sweden (e-mail: \{minkeun.chung, liang.liu, ove.edfors\}@eit.lth.se).}
\thanks{A part of this paper has been presented at the 20th IEEE International Workshop on Signal Processing  Advances in Wireless Communications~(SPAWC), Cannes, France, Jul. 2019~\cite{Mchung2019_1}.}

}

%

\maketitle
\begin{abstract}
Phase-noise~(PN) estimation and compensation are crucial in millimeter-wave~(mmWave) communication systems to achieve high reliability. The PN estimation, however, suffers from high computational complexity due to its fundamental characteristics, such as spectral spreading and fast-varying fluctuations. In this paper, we propose a new framework for low-complexity PN compensation in orthogonal frequency-division multiplexing systems. The proposed framework also includes a pilot allocation strategy to minimize its overhead. The key ideas are to exploit the coherence bandwidth of mmWave systems and to approximate the actual PN spectrum with its dominant components, resulting in a non-iterative solution by using linear minimum mean squared-error estimation. The proposed method obtains a reduction of more than $2.5 \times$ in total complexity, as compared to the existing methods. Furthermore, we derive closed-form expressions for normalized mean squared-errors~(NMSEs) as a function of critical system parameters, which help in understanding the NMSE behavior in low and high signal-to-noise ratio regimes. Lastly, we study a trade-off between performance and pilot-overhead to provide insight into an appropriate approximation of the PN spectrum.
\end{abstract}
\begin{IEEEkeywords}
Coherence bandwidth, millimeter-wave (mmWave) systems, orthogonal frequency-division multiplexing (OFDM), phase noise, pilot.
\end{IEEEkeywords}

%
\IEEEpeerreviewmaketitle
\section{Introduction}
\label{sec:int}
The range of frequencies from \SI{30}{\GHz} to \SI{300}{\GHz} is usually referred to as the \ac{mmWave} band. A key feature is that there is an abundant spectrum available to support ultra-high data rate transmission. Owing to this, the \ac{mmWave} bands have attracted considerable attention~\cite{rap13, andrews14, minkeun20_11, minkeun21_09}. A critical issue, however, is that severe \ac{PN} arises from a \ac{LO} in practical \ac{mmWave} systems. The \ac{PN} increases with the carrier frequency~\cite{robins84}, resulting in a 20--40 times higher \ac{PN} than \acp{LO} for sub-\SI{6}{\GHz}~\cite{minkeun20_10}. The non-negligible amount of \ac{PN} inevitably leads to significant performance degradation in coherent systems~\cite{zaidi16}. New modulation techniques, such as orthogonal time frequency modulation (OTFS)~\cite{surab19} and frequency-domain multiplexing with a frequency-domain cyclic prefix (FDM-FDCP)~\cite{grim19}, have been recently introduced to tackle this problem in mmWave communications. In \ac{OFDM} systems, the performance drop by \ac{PN} has been demonstrated using various metrics, such as \ac{SINR}~\cite{pol95, stee98, arm01, wu04, pia02}, \ac{BER}~\cite{pol95, tom98}, and channel capacity~\cite{mathe12}. To perform coherent detection in mmWave \ac{OFDM} systems, it is imperative to estimate and compensate the combined effect of \ac{PN} and the wireless channel, which is a multiplicative process in the time domain, and a circular convolution process in the frequency domain~\cite{pia02}. Unfortunately, this is not a simple task due to the following characteristics of \ac{PN} in \ac{OFDM} systems:
\begin{itemize}
\item Spectral spreading: \ac{PN} brings about spectral spreading of the ideal Dirac-delta impulse at the \ac{LO}'s frequency. The spectral spreading of \ac{PN} has two detrimental effects on the performance of \ac{OFDM} systems. One is the common rotation on all subcarriers of an \ac{OFDM} symbol, called \ac{CPE}; the other is \ac{ICI}, which destroys orthogonality of subcarriers. As \ac{PN} increases, it results in higher \ac{ICI} from neighboring subcarriers.  
\item Fast-varying fluctuations: The \ac{PN} process is fast-varying so that there is a low correlation across consecutive \ac{OFDM} symbols, resulting in estimation and compensation for each \ac{OFDM} symbol. It requires a stringent latency requirement or high buffer cost for \ac{PN} estimation.
\end{itemize}

The problem of simultaneously dealing with both spectral spreading and fast-varying fluctuations is especially challenging in the presence of severe \ac{PN}. In the \ac{OFDM} system, the effective channel coefficient is entanglement of two unknown variables of \ac{PN} and wireless channel components. For this reason, the required estimation problem of effective channel coefficients is formulated as an underdetermined system, which generally has infinitely many solutions. Obtaining an accurate solution is, therefore, not guaranteed. One could argue that, it is possible to solve this problem by using a Bayesian approach~\cite{zhong13}. However, it might require high-computational complexity, which makes it more challenging to meet the requirement that the \ac{PN} estimate must be updated every \ac{OFDM} symbol. 

The problem of severe \ac{PN} continues to be a significant challenge in multi-antenna systems, i.e, \ac{MIMO}, based on coherent beamforming. Recently, \cite{pitaro15, krishnan16, emil15, zhang18, jacob18} have investigated the impact of such \ac{PN} at large-scale antenna systems, so-called massive \ac{MIMO}. \cite{pitaro15, krishnan16} have analytically shown the \ac{PN} impact on the performance of precoders\hspace{0.2em}/equalizers at massive \ac{MIMO} \ac{BS}. It has been generalized with hardware impairments including multiplicative phase-drifts and additive distortion noise \cite{emil15, zhang18, jacob18}.  A common observation in the previous research is that, increasing the number of antennas can be beneficial for PN mitigation since the phase-drifts average out. However, separate LOs for each antenna are required to obtain such benefits, resulting in high-cost hardware architecture.

Plenty of methods for \ac{PN} estimation and compensation have been investigated in~\cite{rob95, arm98, casas02, liu04, suy09, meh12, wang16, kre18, hua15, wang18, Qi18}. Early studies on \ac{PN} compensation have used quite strong assumptions such as small \ac{PN}~\cite{rob95, arm98} and perfect channel state information~\cite{casas02, liu04} at the receiver. In the case where both \ac{PN} and channel state information are unknown, joint channel and \ac{PN} estimation~\cite{suy09, meh12}, iterative joint \ac{PN} estimation and data detection~\cite{wang16, kre18} have been presented. However, such techniques may be too complicated to be implemented in practical wireless systems. 

Pilot-assisted transmission simplifies the challenging task of receiver design for coherent processing in general. The use of pilots may also be beneficial to achieve low-complexity estimation or to acquire the instantaneous channel coefficients. In this regard, a dedicated pilot symbol for phase tracking, called \emph{\ac{PTRS}}, has been introduced in the \ac{3GPP} \ac{NR}~\cite{38_211}. Motivated by this fact, \cite{hua15, wang18, Qi18} have designed the dedicated pilot pattern for \ac{PN} tracking so that it has a high density in the time domain to tackle the low correlation of \ac{PN} across \ac{OFDM} symbols. These solutions, however, have been focused on tracking only \ac{CPE} while there is no consideration to estimate the performance limiting \ac{ICI} components in \ac{mmWave} systems. 

\emph{Contributions:} We develop a novel framework for low-complexity \ac{PN} compensation for \ac{OFDM} systems. The key ideas are to exploit the coherence bandwidth of \ac{mmWave} systems and to approximate the actual \ac{PN} spectrum with its dominant components. Our main contributions are summarized as follows:
\begin{itemize}
\item We reformulate the joint estimation problem of \ac{PN} and channel from an underdetermined system into a system with the same number of observations and unknowns, which enables low-complexity \ac{PN} estimation by using \ac{LS} and \ac{LMMSE} estimators. The proposed algorithm obtains a reduction of more than $2.5\times$ in total complexity, as compared to the existing method.
\item We design a pilot pattern that has a carefully selected set of symbols to estimate the combined effect of dominant \ac{PN} components and channel frequency response. Furthermore, the minimum pilot-overhead ratio for our proposed method is quantified with a set of system parameters related to the channel coherence structure. 
\item We derive closed-form expressions for \acp{NMSE} of each estimator for joint \ac{PN} and channel estimation. These expressions are represented as a function of \ac{OFDM} parameters, \ac{LO} quality, \ac{SNR}, and approximation order of the \ac{PN} spectrum. Further, this helps in understanding the \ac{NMSE} behavior in low and high \ac{SNR} regimes, providing an informative guideline for pilot allocation in \ac{mmWave} \ac{OFDM} systems.  
\item We present a trade-off between performance and pilot-overhead. For the analysis, the \ac{BER} and throughput performance of the proposed method are evaluated. This trade-off provides insight into an appropriate approximation of the \ac{PN} spectrum, according to \ac{SNR} and \ac{PN} environments. 
\end{itemize}

\textit{Notation:} The set of complex numbers is denoted by $\mathbb{C}$. Lowercase boldface letters stand for column vectors and uppercase boldface letters designate matrices. For a vector or a matrix, we denote its transpose, conjugate, and conjugate transpose $\mathbf{(\cdot)}^{\rm T}$, $\mathbf{(\cdot)}^{*}$, and $\mathbf{(\cdot)}^{\rm H}$, respectively; the subscript notations $\mathbf{(\cdot)}_{\rm t}$ and $\mathbf{(\cdot)}_{\rm f}$ stand for the time- and frequency-domain representations of a vector or a matrix. The $N \times N$ identity matrix is denoted by $\mathbf{I}_{N}$, and the $N \times M$ all-zeros matrix by $\mathbf{0}_{N \times M}$. The expectation operator and Euclidean norm is denoted by ${\mathbb{E}}[\mathbf{\cdot}]$ and $\|\mathbf{\cdot} \|_{2}$, respectively. Sets are designated by upper-case calligraphic letters; the cardinality and complement of the set $\mathcal{T}$ is $|\mathcal{T}|$ and $\mathcal{T}^{c}$, respectively; the difference between two sets $\mathcal{T}$ and $\mathcal{F}$ is denoted by $\mathcal{T} \setminus \mathcal{F}$. The operators for circular convolution, deconvolution, and Hadamard product are written as $\circledast$, $\circledast^{-1}$, and $\circ$, respectively; $\floor{x}$ and $\ceil{x}$ denote the greatest/least integer less/greater than or equal to $x$.

\textit{Outline:} The remainder of this paper is organized as follows: In Section~\ref{sec:sys}, we describe the system model under consideration. Section~\ref{sec:algorithm} describes the proposed \ac{PN} and channel compensation algorithm. In Section~\ref{sec:NMSE_analysis}, we analyze the \ac{NMSE} performance of the proposed method by numerical evaluation. Section~\ref{sec:OH_Complex_Analysis} addresses the pilot-overhead and the computational complexity of our proposed algorithm. Section~\ref{sec:tradeoff} present the trade-off between performance and pilot-overhead. A summary and concluding remarks appear in Section~\ref{sec:conclusion}.

\section{System Description and Preliminaries}
\label{sec:sys}
In this section, we briefly overview our basic idea to tackle the joint estimation problem of \ac{PN} and channel frequency response and compare it with the approach of existing solutions. Before moving on to this, we first present the system and \ac{PN} models that will be used in this paper.

\subsection{System Model}
We consider an \ac{OFDM} system with $N$~subcarriers, a  sampling period~$T_{\rm s}$,  a subcarrier spacing~$\Delta f$, and a bandwidth~$B = 1/T_{\rm s} = N\Delta f$. Let $\{X_{k}\}_{k=0}^{N-1}$ be the transmitted symbol sequence across $N$ subcarriers of an \ac{OFDM} symbol, with an average per-symbol power constraint $\mathbb{E} [|X_{k}|^2]=E_{\rm s}$. An $N$-point unitary \ac{IDFT} of $\{X_{k}\}_{k=0}^{N-1}$ provides the time-domain representation of the \ac{OFDM} symbol as
\begin{align}
\label{eq:ifft}
x_{n} = {\frac{1}{\sqrt{N}}} {\sum_{k=0}^{N-1}} X_{k} e^{j 2 \pi k n/N},
\end{align}
where time index $n \in \{-N_{\rm cp}, -N_{\rm cp}+1,\cdots,0,1,\ldots,N-1\}$. Each \ac{OFDM} symbol is assumed to consist of a \ac{CP} of length-$N_{\rm cp}$ samples.

For our subsequent analysis, we adopt the coherence block model with a coherence time $T_{\rm c}$ and a coherence bandwidth $B_{\rm c}$. In this model, there are two parameters widely used in the literature~\cite{mar10, rusek13, emil16}. One is the number of \ac{OFDM} symbols within $T_{\rm c}$, and the other is the number of subcarriers within $B_{\rm c}$. These parameters are defined as
\begin{align}
\label{eq:nct}
N_{\rm ct} \triangleq  \floor{{T_{\rm c}}/{T_{\rm sym}}},
\end{align}
\begin{align}
\label{eq:ncb}
N_{\rm cb} \triangleq \floor{{B_{\rm c}}/{\Delta f}},
\end{align}
where $T_{\rm sym}$ the duration of one \ac{OFDM} symbol. We assume that the coherence block spans $N_{\rm ct}$ and $N_{\rm cb}$ successive \ac{OFDM} symbols and subcarriers, over which the channel impulse and frequency response, respectively, is constant.

\subsection{Phase Noise Model}
\label{subsec:pnm}
We consider the model introduced in~\cite{demir00} to illustrate the \ac{PN} of a free-running oscillator. The \ac{PN} is defined as
\begin{align}
\label{eq:pn_def}
\phi(t) = 2 \pi f_{\rm o} \eta(t)	
\end{align}
where $f_{\rm o}$ denotes an oscillator frequency. A random time shift $\eta(t)$ becomes, asymptotically with time, a Wiener process as
\begin{align}
\label{eq:wm}
\eta(t) = \sqrt{c} W(t),	
\end{align}
where $c$ denotes the parameter indicating an oscillator quality; $W(t)$ represents a Wiener process having an accumulated Gaussian random variable with i.i.d. $\mathcal{N}(0,1)$, \ie, $W(t_{2})-W(t_{1}) \propto \hspace{0.125em} \mathcal{N} (0,\sqrt{\Delta t})$ where $\Delta t = \mid t_{2} - t_{1} \mid$. The variance of the Wiener process~$\eta(t)$ increases linearly with the time difference $\Delta t$, \ie, $\sigma_{\eta}^{2}= c \Delta t$. According to (\ref{eq:pn_def}), $\phi(t)$ is also a Wiener process with zero mean and variance $2 \pi \beta \Delta t $, where $\beta$ denotes the two-sided 3-\si{\decibel} linewidth of the Lorentzian power spectral density\footnote{In this \ac{PN} model, the connection between $\beta$ in the frequency domain and $c$ in the time domain is described as $\beta = 2 \pi f_{o}^{2} c$.}~\cite{pol95}. 

\subsection{OFDM Signal Model with Phase Noise}
\label{subsec:osm}
The \ac{PN} at the receiver influences the channel output as an angular multiplicative distortion in the time domain. Then, the received signal in the time domain ${\mathbf{y}_{\rm {t}}} \in \mathbb{C}^{N \times 1}$ is
\begin{align}
\label{eq:ytime}
\begin{split}
{\mathbf{y}_{\rm {t}}} 
&= 
\mathbf{p}_{\rm t} \circ ({\mathbf{x}_{\rm {t}}} \circledast {\mathbf{h}_{\rm {t}}}) + {\mathbf{z}_{\rm {t}}}\\
&=
{\mathbf{\Phi}}_{\rm t}
(
{\mathbf{x}}_{\rm t} 
\circledast
{\mathbf{h}}_{\rm t}
)
+
{\mathbf{z}}_{\rm t},
\end{split}
\end{align}
where $\mathbf{p}_{\rm t} = \lbrack e^{j\phi_{0}}, e^{j\phi_{1}}, \cdots, e^{j\phi_{N-1}} \rbrack^{\rm T} \in \mathbb{C}^{N \times 1}$ is the \ac{PN} realization during one OFDM symbol, $\mathbf{x}_{\rm {t}} \in \mathbb{C}^{N \times 1}$ the transmitted signal, $\mathbf{h}_{\rm {t}} \in \mathbb{C}^{N \times 1}$ the channel impulse response, ${\mathbf{z}_{\rm {t}}} \in \mathbb{C}^{N \times 1}$ the additive white Gaussian noise~(AWGN) with i.i.d. ${\mathcal{CN}}\left(0, \sigma_{z}^2 \right)$ entries, and $\mathbf\Phi_{\rm {t}} = {\rm{diag}} \{e^{j\phi_{n}}\}_{n=0}^{N-1} \in \mathbb{C}^{N \times N}$ the diagonal matrix with the entries of $\mathbf{p}_{\rm t}$ on its main diagonal. In view of the duality, the discrete Fourier transform~(DFT) of a product of two finite-length sequences is the circular convolution of their respective DFTs~\cite{opp89}. Thus, the received signal in the frequency domain ${\mathbf{y}_{\rm {f}}} \in \mathbb{C}^{N \times 1}$ is
\begin{align}
\label{eq:yfreq1}
\begin{split}
{\mathbf{y}_{\rm {f}}} 
&= 
{\mathbf{p}_{\rm {f}}} \circledast ({\mathbf{x}_{\rm {f}}}\circ {\mathbf{h}_{\rm {f}}}) + {\mathbf{z}_{\rm {f}}}\\
&=
\mathbf\Phi_{\rm {f}}{\mathbf{H}_{\rm f}} {\mathbf{x}_{\rm {f}}} + {\mathbf{z}_{\rm {f}}},
\end{split}
\end{align}
where $\mathbf{p}_{\rm {f}} = [P_{0}, P_{1}, \cdots, P_{N-1}]^{\rm T} \in \mathbb{C}^{N \times 1}$ is the DFT coefficient vector of the time-domain \ac{PN} sequence $\{e^{j\phi_{n}}\}_{n=0}^{N-1}$, \ie, $P_{i} = \frac{1}{N}	\sum_{n=0}^{N-1} e^{j \phi_{n}}e^{-j2\pi n i / N}$; $\mathbf{x}_{\rm {f}} = [X_{0}, X_{1}, \cdots, X_{N-1}]^{\rm T}$, $\mathbf{h}_{\rm {f}} = [H_{0}, H_{1}, \cdots, H_{N-1}]^{\rm T}$, $\mathbf{z}_{\rm {f}} = [Z_{0}, Z_{1}, \cdots, Z_{N-1}]^{\rm T}\in \mathbb{C}^{N \times 1}$ the transmit symbol, channel frequency fresponse, and noise, respectively, in the frequency domain; $\mathbf\Phi_{\rm {f}}={{\rm circ} (\mathbf{p}_{\rm f})}$ is a circulant matrix formed by the spectral \ac{PN} components, 
\begin{align}
\label{eq:Phi_f}
\mathbf\Phi_{\rm {f}}
=
\begin{bmatrix}
P_{0} & P_{N-1} & P_{N-2} & \cdots & \cdots & P_{1}\\
P_{1} & P_{0} & P_{N-1} & \cdots & \cdots & P_{2}\\
P_{2} & P_{1} & P_{0} & P_{N-1} & \ddots & P_{3}\\
\vdots & \ddots & \ddots & \ddots & \ddots & \vdots\\
P_{N-2} & P_{N-3} & \cdots & P_{1} & P_{0} & P_{N-1}\\
P_{N-1} & P_{N-2} & \cdots &  P_{2} & P_{1} & P_{0}
\end{bmatrix},
\end{align}
${\mathbf{H}_{\rm f}}={\rm diag}\{H_{k}\}_{k=0}^{N-1}$ is the diagonal matrix with the entries of $\mathbf{h}_{\rm f}$ on its main diagonal. Given the coherence block model, we denote the number of coherence blocks $N_{\rm c}$. Thus, the channel frequency response consists of $N_{\rm c}$ different channel coefficients\footnote{The parameter $N_{\rm c} \triangleq N/N_{\rm cb}$, where we assume that $N$ is divisible by $N_{\rm cb}$, is also called \emph{the number of resource blocks} in 3GPP.}, with i.i.d. $\mathcal{CN}(0,1)$ entries, and its index set is denoted $\mathcal{C}$, \ie, $\lbrace H_{k}\rbrace_{k = 0}^{N_{\rm c}-1}, k \in \mathcal{C}$. To look into the CPE and the ICI effect on the received signal for each subcarrier $k \in \lbrace 0,1, \cdots, N-1 \rbrace$, let us rewrite (\ref{eq:yfreq1}) in the sample-wise form
\begin{align}
\label{eq:yfreq2}
Y_{k} = \underbrace{P_{0}}_{\rm CPE}H_{k}X_{k} 
+ \underbrace{{\sum_{\ell=0, \ell \neq k}^{N-1}} P_{(k-\ell)_{N}}H_{\ell}X_{\ell}}_{\rm ICI}
+ Z_{k},
\end{align}
where $(\cdot)_{N}$ denotes the modulo-$N$ operation. In the absence of \ac{PN}, by the fact that $P_{i}$ is a Kronecker delta function $\delta \lbrack i \rbrack$, the received signal~(\ref{eq:yfreq2}) becomes
\begin{align}
\label{eq:ideal_y}
Y_{k} = X_{k}H_{k}+Z_{k}.
\end{align} 
\begin{figure}[t!]
\centering
\includegraphics[width = 3.46in]{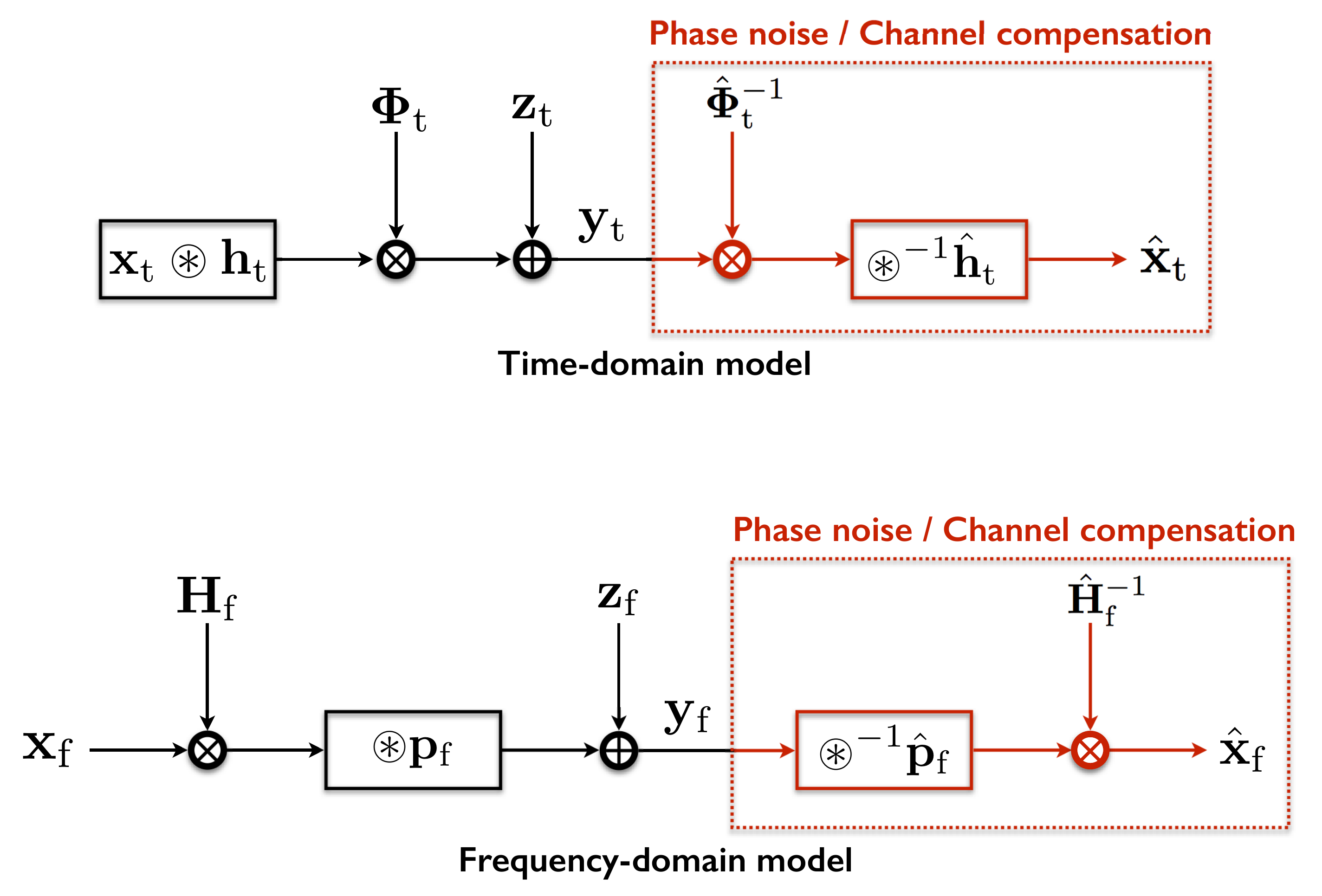}
\caption{Basic model of \ac{PN}/channel compensation and detection in time and frequency domains.}
\label{fig:basic_model_comp}
\end{figure}

\subsection{Phase Noise and Channel Compensation Model}
\label{subsec:pnccm}
In this subsection, we provide a brief comparison of the conventional and proposed approaches for \ac{PN} and channel compensation. Fig.~\ref{fig:basic_model_comp} displays the basic model of \ac{PN}/channel compensation and detection, where $\hat{(\cdot)}$ designates the corresponding estimated or decoded vector/matrix. Researchers have investigated how to efficiently reduce the unknowns to handle the underdetermined problem of joint \ac{PN} and channel estimation, which led to low-complexity estimation methods. A popular approach is to utilize the fact that the channel impulse response $\mathbf{h}_{\rm t}$ in the time domain has fewer parameters than in the frequency response, resulting in time-domain channel estimation with a smaller number of unknowns. Based on this fact, the joint estimation algorithms for frequency-domain \ac{PN}~\cite{rab10, mat16}, and time-domain \ac{PN}~\cite{zou07, wang17}, respectively, have been presented. The basic technique used in~\cite{rab10, mat16, zou07, wang17} is a joint least-squares estimation. Especially,~\cite{mat16} introduced a new constraint by the geometrical property of spectral \ac{PN} components to complement the weakness of the relaxed constraint used in~\cite{rab10}. Further,~\cite{wang17} showed to be able to reduce the computational complexity of least-squares estimation significantly by using the majorization-minimization technique. However, the above least-squares estimation methods require a full-pilot \ac{OFDM} symbol to perform joint \ac{PN} and channel estimation, translating into significant pilot overhead.
\begin{figure}[t!]
\centering
\includegraphics[width = 3.46in]{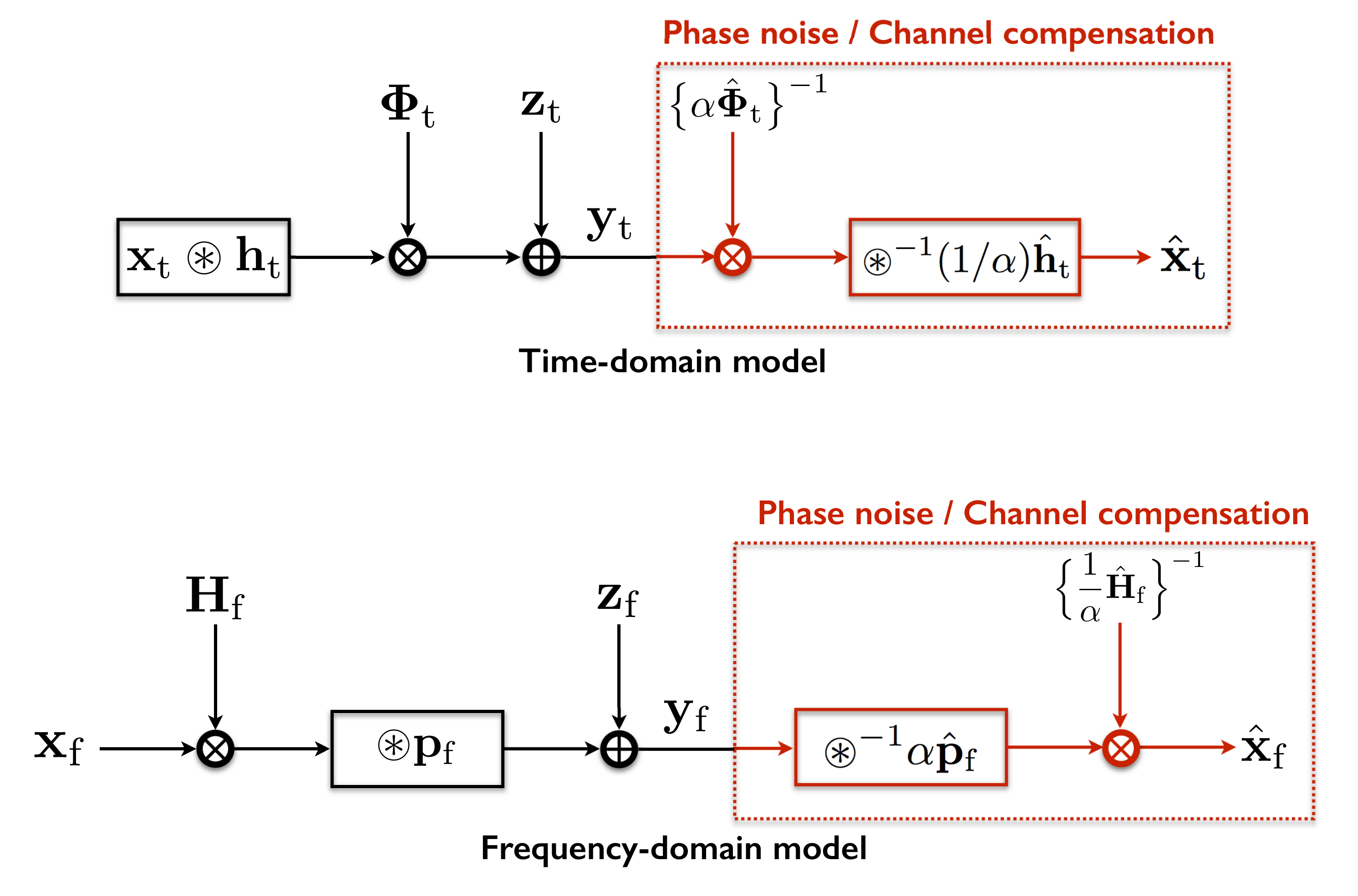}
\caption{Proposed model of \ac{PN}/channel compensation and detection in time and frequency domains.}
\label{fig:prop_model_comp}
\end{figure}

In contrast to the existing approach, we consider channel coherence in the frequency domain to manage the underdetermined problem. The coherence bandwidth of a \ac{mmWave} system is inherently much larger than those of conventional systems~\cite{hur13}. It is a promising basis for more suitable \ac{PN} compensation in \ac{mmWave} systems. Larger coherence bandwidth can facilitate the estimation of scaled \ac{PN} components in the frequency domain, \ie, $\alpha\hat{\mathbf{p}}_{\rm f}$, $\alpha \in \mathbb{C}$, as illustrated in Fig.~\ref{fig:prop_model_comp}. The deconvolution by the scaled \ac{PN} estimates suppresses the effect of \ac{ICI} by \ac{PN}, translating it into a simple estimation problem for $(1/\alpha)\mathbf{h}_{\rm f}$, which can be estimated by using as many pilots as there are channel coefficients in $\mathbf{h}_{\rm f}$.

\subsection{Effective Channel with Large Coherence Bandwidth}
\label{subsec:eclcb}
The effective channel coefficient can be recovered, provided that there are as many observations as unknowns. To see how coherence bandwidth could be utilized to meet this condition, let us go through two examples. Let $N_{\rm p}$ denote the number of dominant \ac{PN} components in the frequency domain\footnote{Since the output spectrum of \ac{PN} has a low-pass characteristic, a few numbers of significant PN components in the frequency domain provide a quite good approximation of the \ac{PN} realization. Essentially, severe spectral spreading increases $N_{\rm p}$ to be considered. In this paper, therefore, we will deal with the generalized $N_{\rm p}$ for \ac{PN} compensation in mmWave systems.}.

{\it{Example 1:}} Consider four received samples as shown in~(\ref{eq:yfreq2}) when $N_{\rm cb}=1$ and $N_{\rm p}=3$. 
\begin{align}
\label{eq:Np_3_N_cb_1}
\begin{split}
Y_{0} &= \underbrace{P_{0}}_{\rm CPE}H_{0}X_{0} 
+\underbrace{P_{1}H_{N-1}X_{N-1} + P_{N-1}H_{1}X_{1}}_{\rm dominant \ ICI}\\ 
&\;\;\;
+{\sum_{{\ell} \in \mathcal{L} \setminus \{0,1,N-1\}}} P_{(0-\ell)_{N}}H_{\ell}X_{\ell}
+ Z_{0},\\
Y_{1} &= \underbrace{P_{0}}_{\rm CPE}H_{1}X_{1} 
+\underbrace{P_{1}H_{0}X_{0} + P_{N-1}H_{2}X_{2}}_{\rm dominant \ ICI}\\
&\;\;\;  
+ {\sum_{{\ell} \in \mathcal{L} \setminus \{0,1,2\}}} P_{(1-\ell)_{N}}H_{\ell}X_{\ell}
+ Z_{1},\\
Y_{2} &= \underbrace{P_{0}}_{\rm CPE}H_{2}X_{2} 
+\underbrace{P_{1}H_{1}X_{1} + P_{N-1}H_{3}X_{3}}_{\rm dominant \ ICI}\\
&\;\;\;  
+ {\sum_{{\ell} \in \mathcal{L} \setminus \{1,2,3\}}} P_{(2-\ell)_{N}}H_{\ell}X_{\ell}
+ Z_{2},\\
Y_{3} &= \underbrace{P_{0}}_{\rm CPE}H_{3}X_{3} 
+\underbrace{P_{1}H_{2}X_{2} + P_{N-1}H_{4}X_{4}}_{\rm dominant \ ICI}\\
&\;\;\;  
+ {\sum_{{\ell} \in \mathcal{L} \setminus \{2,3,4\}}} P_{(3-\ell)_{N}}H_{\ell}X_{\ell}
+ Z_{3},\\
\end{split}
\end{align}
where subcarrier index $\ell \in \mathcal{L} \triangleq \{0,1,...,N-1\}$. Assume that ICI terms represented by the summation operator and noise components are negligible, and all transmitted symbols are used as pilots. In the four observations, there are twelve different unknowns, \ie, $\{P_{0}H_{k}\}_{k = 0}^{3}$, $\{P_{1}H_{k}\}_{k = N-1, 0}^{2}$, and $\{P_{N-1}H_{k}\}_{k = 1}^{4}$, being underdetermined. 

{\it{Example 2:}}
Consider the same number of received samples when $N_{\rm cb} = 6$ and $N_{\rm p} = 3$ as follows.
\begin{figure}[t!]
\centering
\includegraphics[width = 3.4in]{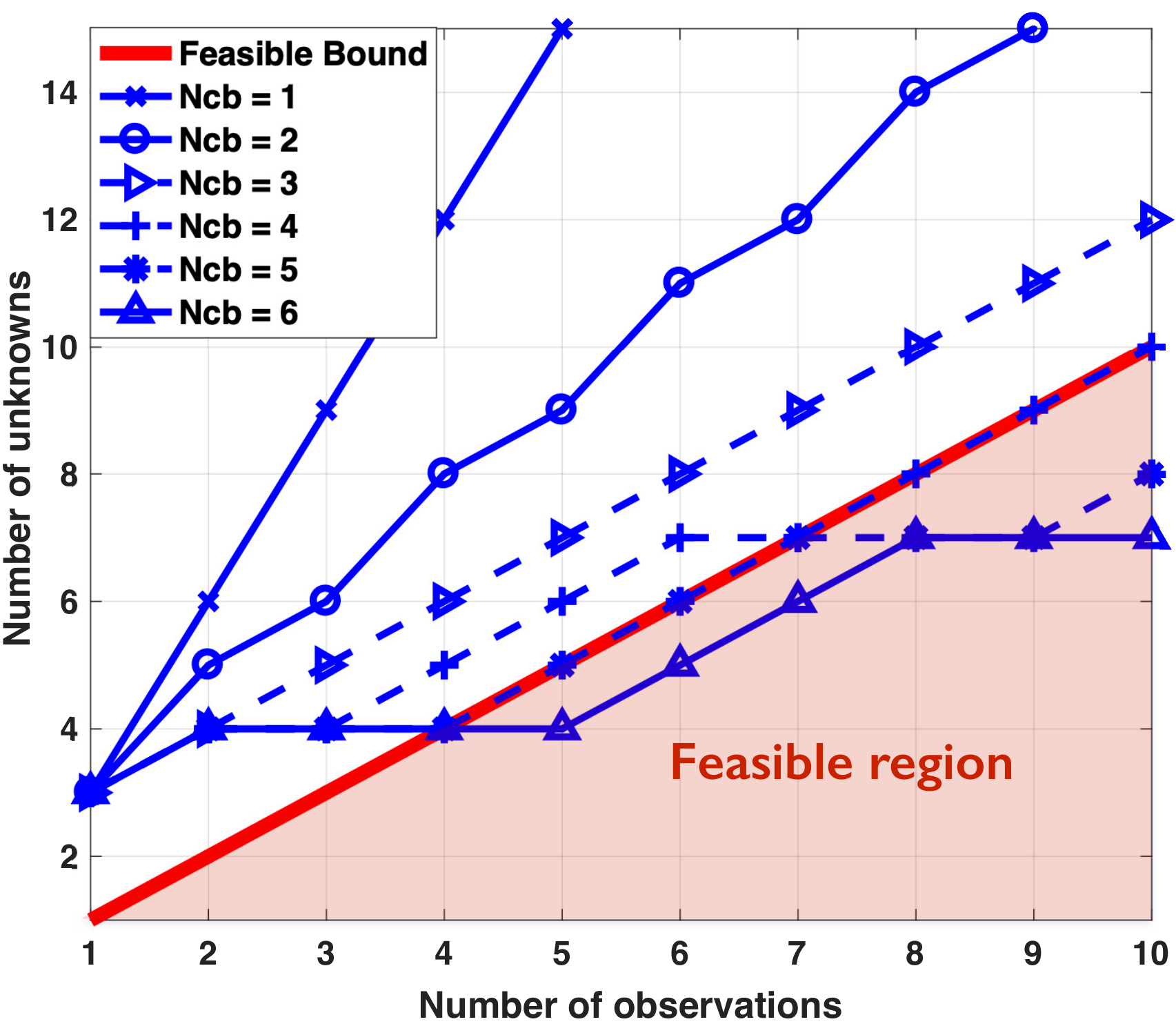}
\caption{The total number of effective channel unknowns involved in the corresponding numbers of observations according to $N_{\rm cb}$~($N_{\rm p}=3$).}
\label{fig:feasi_N_cb}
\vspace*{-0.25 cm}
\end{figure}
\begin{align}
\label{eq:Np_3_N_cb_6}
\begin{split}
Y_{0} &= \underbrace{P_{0}}_{\rm CPE}H_{0}X_{0} 
+\underbrace{P_{1}H_{\lceil{(N-1)/6}\rceil}X_{N-1} + P_{N-1}H_{0}X_{1}}_{\rm dominant \ ICI}\\
&\;\;\; 
+ {\sum_{{\ell} \in \mathcal{L} \setminus \{0,1,N-1\}}} P_{(0-\ell)_{N}}H_{\lfloor{\ell/6}\rfloor}X_{\ell}
+ Z_{0},\\
Y_{1} &= \underbrace{P_{0}}_{\rm CPE}H_{0}X_{1} 
+\underbrace{P_{1}H_{0}X_{0} + P_{N-1}H_{0}X_{2}}_{\rm dominant \ ICI}\\
&\;\;\;  
+ {\sum_{{\ell} \in \mathcal{L} \setminus \{0,1,2\}}} P_{(1-\ell)_{N}}H_{\lfloor{\ell/6}\rfloor}X_{\ell}
+ Z_{1},\\
Y_{2} &= \underbrace{P_{0}}_{\rm CPE}H_{0}X_{2} 
+\underbrace{P_{1}H_{0}X_{1} + P_{N-1}H_{0}X_{3}}_{\rm dominant \ ICI}\\
&\;\;\;  
+ {\sum_{{\ell} \in \mathcal{L} \setminus \{1,2,3\}}} P_{(2-\ell)_{N}}H_{\lfloor{\ell/6}\rfloor}X_{\ell}
+ Z_{2},\\
Y_{3} &= \underbrace{P_{0}}_{\rm CPE}H_{0}X_{3} 
+\underbrace{P_{1}H_{0}X_{2} + P_{N-1}H_{0}X_{4}}_{\rm dominant \ ICI}\\
&\;\;\;  
+ {\sum_{{\ell} \in \mathcal{L} \setminus \{2,3,4\}}} P_{(3-\ell)_{N}}H_{\lfloor{\ell/6}\rfloor}X_{\ell}
+ Z_{3}.\\
\end{split}
\end{align}
In this case, it is possible to recover all effective channel coefficients because there are only as many unknowns as observations. Fig.~\ref{fig:feasi_N_cb} shows the total number of effective channel unknowns involved in the corresponding numbers of observations, according to $N_{\rm cb}$, when $N_{\rm p} = 3$.  For channel frequency responses with $N_{\rm cb}$ larger than three, there are fewer unknowns than observations. With this insight, in the next section, we describe a low-complexity \ac{PN}/channel estimation followed by the \ac{NMSE} analysis.

\section{Proposed Algorithm}
\label{sec:algorithm}
Exploiting the approximation of the \ac{PN} spectrum and large coherence bandwidth, the joint estimation problem of \ac{PN} and channel can be reformulated from a heavily underdetermined system into a system with the same number of equations and unknowns, referred to as a fully determined linear system. This enables low-complexity \ac{PN}/channel estimation by using the \ac{LS} and \ac{LMMSE} estimators. In the proposed algorithm, two kinds of frequency-domain estimations are required. One is for the $N_{\rm p}$ dominant \ac{PN} components scaled by $\alpha$ and the other for the $N_{\rm c}$ scaled-channel coefficients, as illustrated in \fig\ref{fig:prop_model_comp}. 

To define dominant \ac{PN} components, we adopt $\gamma$ as the approximation order of the \ac{PN} spectrum, where $N_{\rm p} = 2\gamma+1$ for $\gamma \in \{0,1,\cdots,N/2\}$. The index set of dominant \ac{PN} is defined as $\mathcal{P} \triangleq \lbrace 0, 1, \cdots, N-1 \rbrace \setminus \lbrace \gamma+1, \gamma+2, \cdots, N-(\gamma+1) \rbrace$. Let $\mathbf{p}_{\rm {f}, \gamma} \in \mathbb{C}^{N \times 1}$ be the approximated \ac{PN} vector where $P_{i} = 0$, $i \in \mathcal{P}^{\rm c}$, and $\mathbf{e}_{\rm f, app} \triangleq \mathbf{p}_{\rm f} - \mathbf{p}_{{\rm f}, \gamma} \in \mathbb{C}^{N \times 1}$ be the approximation error vector, \eg, $\mathbf{p}_{\rm {f}, 2} = [P_{0}, P_{1}, P_{2}, 0, \cdots, 0, P_{N-2}, P_{N-1}]^{\rm T}$ and $\mathbf{e}_{\rm f, app} = [\mathbf{0}_{1 \times  3}, P_{3}, P_{4}, \cdots, P_{N-(\gamma+1)}, \mathbf{0}_{1 \times 2}]^{\rm T}$ for $\gamma = 2$. The frequency-domain effective channel component in (\ref{eq:yfreq1}) is defined as $F_{i,k} \triangleq  P_{i}H_{k}$, which is the element in a set of multiplications between $P_{i}$ and $H_{k}$  for $i \in \mathcal{P}$ and $k \in \mathcal{C}$. We call this \emph{PN-affected channel}. With the $\gamma$-order approximation, the PN-affected-channel matrix ${\mathbf{F}}_{\gamma}$ and the approximation error matrix ${\mathbf{E}}_{\gamma}$ are, respectively,
\begin{align}
\label{eq:F_gamma}
{\mathbf{F}}_{\gamma} = \mathbf\Phi_{\rm {f}, \gamma}{\mathbf{H}_{\rm f}},	 \text{and}
\end{align}
\begin{align}
\label{eq:E_gamma}
{\mathbf{E}}_{\gamma} = \mathbf{F} - {\mathbf{F}}_{\gamma}
=
\tilde{\mathbf{\Phi}}_{\rm f, \gamma}
\mathbf{H}_{\rm f},	
\end{align}
where $\mathbf\Phi_{\rm {f}, \gamma}={{\rm circ} (\mathbf{p}_{\rm {f},  \gamma})}$, $\tilde{\mathbf{\Phi}}_{\rm f, \gamma}={{\rm circ} (\mathbf{e}_{\rm f, app})}$, and ${\mathbf{F}} = \mathbf\Phi_{\rm {f}}{\mathbf{H}_{\rm f}}$. One of the columns in $\mathbf{F}_{\gamma}$ is estimated  for ICI suppression, which includes $N_{\rm p}$ dominant \ac{PN} components scaled by $\alpha$. As a result of the ICI suppression, the  Toeplitz convolution matrix (\ref{eq:F_gamma}) is converted into a diagonal matrix of which diagonal elements are called the \emph{ICI-free channel} in this paper. \fig\ref{fig:system_overview} illustrates the proposed architecture with PN-affected- and ICI-free-channel estimation. Before explaining the details of proposed algorithm, we first describe the transmission structure in the following subsection.

\begin{figure}[t!]
\centering
\includegraphics[width = 3.47in]{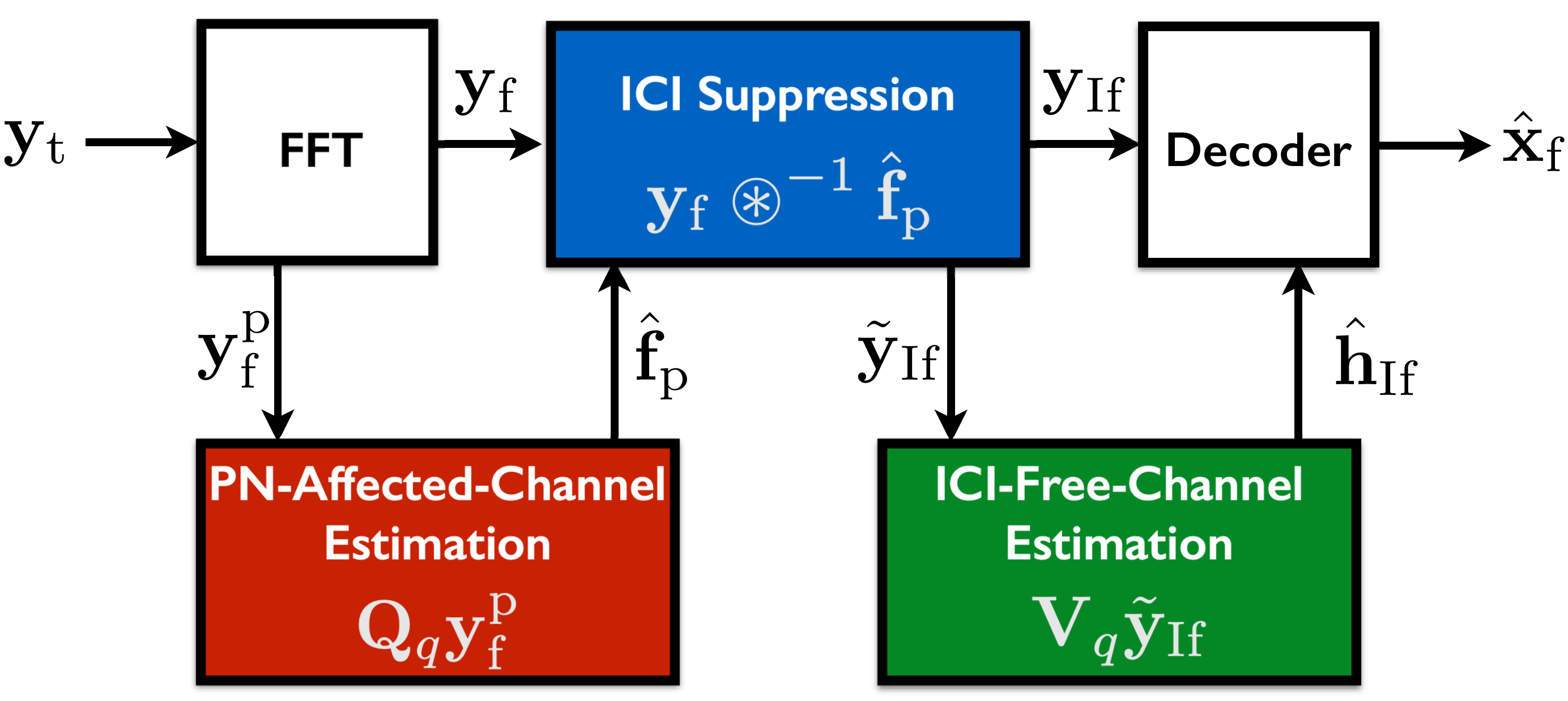}
\caption{System overview of the proposed \ac{PN}/channel compensation architecture. For the PN-affected-estimator $\mathbf{Q}_{q}$ and ICI-free-channel estimator $\mathbf{V}_{q}$, the \ac{LS} and \ac{LMMSE} estimators are applied, \ie, $q \in \lbrace \mathsf{ls}, \mathsf{lmmse} \rbrace$.}
\label{fig:system_overview}
\vspace*{-0.25 cm}
\end{figure}
\begin{figure}[t!]
\centering
\includegraphics[width = 3.47in]{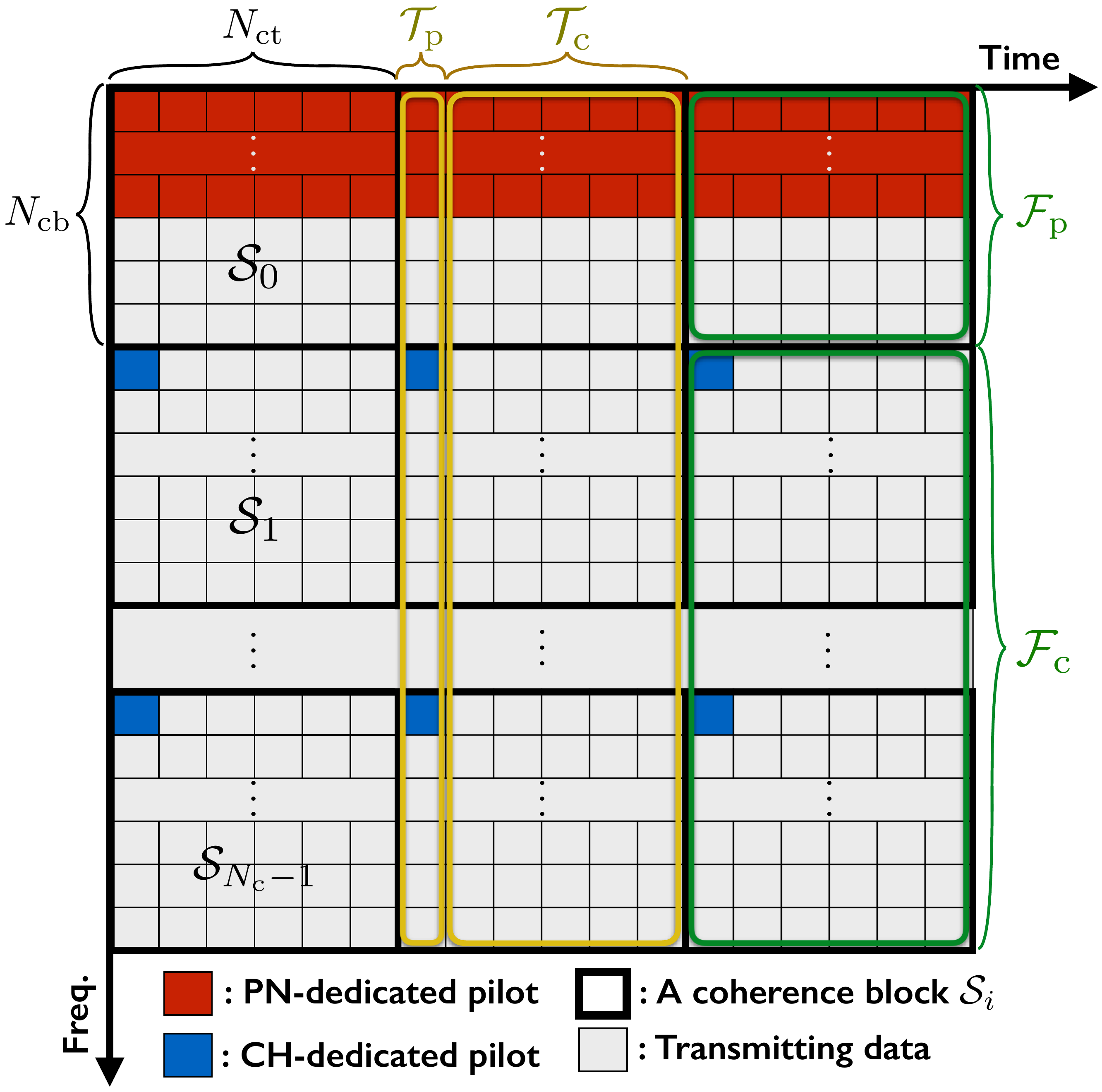}
\caption{An example of transmission structure for PN-affected- and ICI-free-channel estimation. Based on a set of coherence blocks across $N$ subcarriers, $\mathcal{S} = \mathcal{S}_{0} \cup \mathcal{S}_{1} \cup \cdots \mathcal{S}_{N_{\rm c}-1}$, the pattern for resource allocation is identical.}
\label{fig:tx_structure}
\vspace*{-0.25 cm}
\end{figure}

\subsection{Transmission Structure}
\label{subsec:ts}
Let us define a coherence block $\mathcal{S}_{k}$, $k \in \mathcal{C}$ with cardinality $|\mathcal{S}_{k}| = N_{\rm cb}N_{\rm ct}$, and $\mathcal{S} = \mathcal{S}_{0} \cup \mathcal{S}_{1} \cup \cdots \mathcal{S}_{N_{\rm c}-1}$ be a set of non-overlapping coherence blocks across $N$ subcarriers, \ie, $|\mathcal{S}| = N_{\rm c}N_{\rm cb}N_{\rm ct}$, as illustrated in \fig\ref{fig:tx_structure}. To describe resource allocation for pilots and transmitting data, we divide $\mathcal{S}$ into two subsets in the frequency and time domain, respectively; $\mathcal{F}_{\rm p}$ and $\mathcal{F}_{\rm c}$ in the frequency domain, and $\mathcal{T}_{\rm p}$ and $\mathcal{T}_{\rm c}$ in the time domain, where $\mathcal{S} = \mathcal{F}_{\rm p} \cup \mathcal{F}_{\rm c} = \mathcal{T}_{\rm p} \cup \mathcal{T}_{\rm c}$ and $\mathcal{F}_{\rm p} \cap \mathcal{F}_{\rm c} = \mathcal{T}_{\rm p} \cap \mathcal{T}_{\rm c} = \phi$. In the frequency domain, $\mathcal{F}_{\rm p}$ designate the coherence block set that includes \emph{PN-dedicated pilot} for PN-affected-channel estimation, and $\mathcal{F}_{\rm c}$ involves \emph{CH-dedicated pilot} for ICI-free-channel estimation. An example of transmission structure for PN-affected- and ICI-free-channel estimation is shown in \fig\ref{fig:tx_structure}. In the time domain, we consider the fact that the \ac{PN} process is fast-varying within channel coherence time while the wireless channel is invariant, resulting in the PN-affected-channel estimation of each OFDM symbol. Hence, only the PN-dedicated pilot is allocated in $\mathcal{T}_{\rm c}$ while both pilots in $\mathcal{T}_{\rm p}$. The remainder of the coherence block is used for transmitting data. 

\subsection{PN-Affected-Channel Estimation}
\label{subsec:pnac_est}
In this subsection, we elaborate on the PN-affected-channel estimation with the following example.

{\it{Example 3 (PN-Affected-Channel Estimation):}} Suppose $N_{\rm cb} = 6$ and $N_{\rm p} = 3$ in this example. Based on~(\ref{eq:yfreq1}), (\ref{eq:F_gamma}), and (\ref{eq:E_gamma}), the received signal is
\begin{align}
\label{eq:yfreq4}
{\mathbf{y}}_{\rm {f}} = {\mathbf{F}}_{1}{\mathbf{x}_{\rm {f}}} + \underbrace{\mathbf{E}_{1}{\mathbf{x}_{\rm {f}}} + {\mathbf{z}_{\rm {f}}}}_{{\mathbf{w}}_{1}}
= {\mathbf{F}}_{1}{\mathbf{x}_{\rm {f}}} + {\mathbf{w}}_{1},	
\end{align} 
where $\mathbf{F}_{1}$ is presented in~(\ref{eq:F_1}) at the top of next page and denotes the 1-order approximated channel matrix; $\mathbf{E}_{1}$ is its approximation error matrix; ${\mathbf{w}}_{1}$ is the ICI by the approximation error plus AWGN. 
\begin{figure*}[t]
\begin{align}
\label{eq:F_1}
{\mathbf{F}}_{1}=
\begin{bmatrix}
F_{0,0} & F_{N-1,0} & 0 & \cdots & 0 & F_{1,N_{\rm c}-1}\\
F_{1,0} & F_{0,0} & F_{N-1,0} & \ddots & \ddots & 0\\
0 & F_{1,0} & F_{0,0} & F_{N-1,0} & \ddots & \vdots\\
\vdots & \ddots & \ddots & \ddots & \ddots & 0\\
0 & \ddots & \ddots & F_{1,N_{\rm c}-1} & F_{0,N_{\rm c}-1} & F_{N-1,N_{\rm c}-1}\\
F_{N-1,0} & 0 & \cdots &  0 & F_{1,N_{\rm c}-1} & F_{0,N_{\rm c}-1}
\end{bmatrix} \in \mathbb{C}^{N \times N}.
\end{align}
\hrulefill\\
\vspace*{-.3cm}
\end{figure*}
We denote the PN-affected-channel vector by $\mathbf{f}_{\bar{\rm p}} \in \mathbb{C}^{N_{\rm p} \times 1}$, which consists of dominant \ac{PN} components scaled by a channel coefficient as 
\begin{align}
\label{eq:f_p_bar}
\begin{split}
{\mathbf{f}}_{\bar{\rm p}} 
&= [F_{N-\gamma,k}, F_{N-\gamma+1,k}, \cdots, F_{0,k}, \cdots, F_{\gamma-1, k}, F_{\gamma,k}]^{\rm T}\\ 
&= H_{k}\underbrace{\lbrack P_{N-\gamma}, P_{N-\gamma+1}, \cdots, P_{0}, \cdots, P_{\gamma-1}, P_{\gamma} \rbrack^{\rm T}}_{\bar{\mathbf{p}}_{\rm f,\gamma}},
\end{split}
\end{align} 
where $\bar{\mathbf{p}}_{\rm f, \gamma} \in \mathbb{C}^{N_{\rm p} \times 1}$ is the dominant \ac{PN} vector with the $\gamma$-order approximation. Based on~(\ref{eq:f_p_bar}), the coefficient $\alpha$ in Fig.~\ref{fig:prop_model_comp} indicates the channel coefficient $H_{k}$ in $\mathcal{S}_{k}$.{\footnote{As the element in $\mathbf{f}_{\bar{\rm p}}$ is a subset of element in $\mathbf{F_{\gamma}}$, the channel coefficient index $k$ in (\ref{eq:f_p_bar}) depends on the allocation of PN-dedicated pilot.}} In this example, $\mathbf{f}_{\bar{\rm p}} = \lbrack F_{N-1,k}, F_{0,k}, F_{1,k} \rbrack^{\rm T} \in \mathbb{C}^{3 \times 1}$. With the three unknowns in $\mathbf{f}_{\bar{\rm p}}$, a fully determined linear system can be constructed as follows:
\begin{align}
\begin{split}
\label{eq:y_pnac_np_3_1}
{\mathbf{y}_{\rm f}^{\rm p}}  &= {\mathbf{F}_{1}^{\rm p}}{\mathbf{x}_{\rm {f,1}}^{\rm p}} + {\mathbf{w}}_{1}^{\rm p}\\
&=
\hspace*{-0.05 cm}
\underbrace{
\begin{bmatrix}
F_{1,0} & F_{0,0} & F_{N-1,0} & 0 & 0\\
0 & F_{1,0} & F_{0,0} & F_{N-1,0} & 0\\
0 & 0 & F_{1,0} & F_{0,0} & F_{N-1,0}
\end{bmatrix}
}_{\mathbf{F}_{1}^{\rm p}}
\hspace*{-0.05 cm}
\underbrace{
\begin{bmatrix}
X_{0}^{\rm p} \\ X_{1}^{\rm p} \\ X_{2}^{\rm p} \\ X_{3}^{\rm p} \\ X_{4}^{\rm p} 
\end{bmatrix}
}_{\mathbf{x}_{\rm f,1}^{\rm p}}
+
{\mathbf{w}}_{1}^{\rm p},\\
\end{split}	
\end{align}
where $\mathbf{y}_{\rm f}^{\rm p} = \lbrack Y_{1}, Y_{2}, Y_{3} \rbrack^{\rm T}$ is the three observations in $\mathbf{y}_{\rm f}$, and $\mathbf{w}_{1}^{\rm p} \in \mathbb{C}^{3 \times 1}$ the corresponding vector in $\mathbf{w}_{1}$; the element in $\mathbf{x}_{\rm f, 1}^{\rm p}$ is denoted by $X_{k}^{\rm p}$ to distinguish the PN-dedicated pilot from transmitting data. Using the commutative property, (\ref{eq:y_pnac_np_3_1}) can be rewritten as 
\begin{align}
\label{eq:y_pnac_np_3_2}
{\mathbf{y}_{\rm f}^{\rm p}}
=
\underbrace{
\begin{bmatrix}
X_{2}^{\rm p} & X_{1}^{\rm p} & X_{0}^{\rm p}\\
X_{3}^{\rm p} & X_{2}^{\rm p} & X_{1}^{\rm p}\\
X_{4}^{\rm p} & X_{3}^{\rm p} & X_{2}^{\rm p}
\end{bmatrix}
}_{\mathbf{X}_{\rm {f,1}}^{\rm p}}
\underbrace{
\begin{bmatrix}
F_{N-1,0} \\ F_{0,0} \\ F_{1,0}
\end{bmatrix}
}_{\mathbf{f}_{\bar{\rm {p}}}}
+
{\mathbf{w}}_{1}^{\rm p}.
\end{align}
By the PN-dedicated pilot $\{X_{k}^{\rm p}\}_{k=0}^{4}$ such that ${\rm rank}(\mathbf{X}_{\rm f,1}^{\rm p})=3$, all the unknowns in $\mathbf{f}_{\bar{\rm p}}$ can be estimated. Based on (\ref{eq:y_pnac_np_3_2}), the optimization problem for the optimal PN-dedicated pilot matrix, with respect to the approximation order of $\gamma$, is
\begin{equation}
\label{eq:opt_PNAC}
  \begin{aligned}
  \hspace{-0.3em}
    & \underset{\mathbf{X}_{{\rm f}, \gamma}^{\rm p}}{\text{minimize}} \quad
      \|\mathbf{w}_{\gamma}^{\rm p}\|_{\rm F}^{2} \\
    & \text{subject to}
    \hspace{1.0em}
    {\rm rank}(\mathbf{X}_{{\rm f}, \gamma}^{\rm p}) = 2\gamma+1.
  \end{aligned}
\end{equation}
The following theorem provides the optimal solution of (\ref{eq:opt_PNAC}).
\begin{theorem}
\label{thm:optimal_pilot_design}
Assume that the PN-dedicated pilot is allocated in the $\mathcal{S}_{0}$. If a $\gamma$-order approximation of \ac{PN} spectrum is applied, the optimal PN-dedicated pilot matrix~$\mathbf{X}_{{\rm f},\gamma}^{\rm p} \in \mathbb{C}^{N_{\rm p} \times N_{\rm p}}$, for minimizing the ICI by the approximation error, is
\begin{align}
\label{eq:optimal_pilot_1}
\mathbf{X}_{{\rm f},\gamma}^{\rm p} = \mathbf{I}_{N_{\rm p}}	
\end{align}
where 
\begin{align}
\label{eq:optimal_pilot_2}
\mathbf{X}_{{\rm f},\gamma}^{\rm p}
=
\begin{bmatrix}
	X_{2\gamma}^{\rm p} & 	X_{2\gamma-1}^{\rm p} & \cdots  & \cdots & X_{0}^{\rm p}\\
	X_{2\gamma+1}^{\rm p} & X_{2\gamma}^{\rm p} & X_{2\gamma-1}^{\rm p} & \cdots & X_{1}^{\rm p}\\
	X_{2\gamma+2}^{\rm p} & X_{2\gamma+1}^{\rm p} & X_{2\gamma}^{\rm p} & \ddots & \vdots\\
	\vdots & \ddots & \ddots & \ddots & X_{2\gamma-1}^{\rm p}\\
	X_{2\gamma+2\gamma}^{\rm p} & \cdots & X_{2\gamma+2}^{\rm p} & X_{2\gamma+1}^{\rm p} & X_{2\gamma}^{\rm p}
\end{bmatrix} 	
\end{align}
\end{theorem}
\begin{proof}
See Appendix~\ref{append_sec:optimal_PNDP}
\end{proof}

We consider the LS and LMMSE estimators for the PN-affected channel. The optimal PN-dedicated pilot matrix~(\ref{eq:optimal_pilot_1}) leads to lower computational complexity as compared to the conventional \ac{LS} and \ac{LMMSE} estimators~\cite{kay93}. The \ac{LS} and \ac{LMMSE} PN-affected-channel estimators, respectively, is (see Appendix~\ref{append_sec:LMMSE_PNAC} for $\mathbf{Q}_{\mathsf{lmmse}}$)
\begin{align}
\begin{split}
\label{eq:Q_LS}
{\mathbf{Q}}	_{\mathsf{ls}} 
&=
({\mathbf{X}_{\rm {f, \gamma}}^{\rm p}})^{-1}
 = \mathbf{I}_{N_{\rm p}},
\end{split}\\
\begin{split}
\label{eq:Q_LMMSE}
{\mathbf{Q}}_{\mathsf{lmmse}} 
&=
\mathbf{R}_{\mathbf{p}\mathbf{p}}^{\gamma}
\lbrace\mathbf{R}_{\mathbf{p}\mathbf{p}}^{\gamma} + \mathbf{R}_{\mathbf{ici}}^{\gamma} + \frac{1}{\mathsf{SNR}}{\mathbf{I}_{N_{\rm p}}}\rbrace^{-1},
\end{split}
\end{align}
where $\mathbf{R}_{\mathbf{p}\mathbf{p}}^{\gamma}=\mathbb{E}\lbrace \bar{\mathbf{p}}_{{\rm f},\gamma}(\bar{\mathbf{p}}_{{\rm f},\gamma})^{\rm H}\rbrace$ is the autocorrelation matrix of $\bar{\mathbf{p}}_{{\rm f},\gamma}$ in (\ref{eq:f_p_bar}) and $\mathbf{R}_{\mathbf{ici}}^{\gamma}$ the autocorrelation matrix of ICI vector arising from the $\gamma$-order-approximation error{\footnote{In practice, the second-order statistics of spectral \ac{PN} components generated from a fixed \ac{LO} are stationary, therefore we assume that $\mathbf{R}_{\mathbf{pp}}^{\gamma}$ and $\mathbf{R}_{\mathbf{ici}}^{\gamma}$ can be estimated by using one-shot or long-term estimation.}}; ${\mathsf{SNR}} \triangleq E_{\rm s}/\sigma_{z}^{2}$ the average \ac{SNR}. The LS/LMMSE estimate of $\mathbf{f}_{\bar{\rm p}}$ is
\begin{align}
\begin{split}
\label{eq:f_p_hat_LS}
{\hat{\mathbf{f}}}_{\bar{\rm p},q} 
&=
{{\mathbf{Q}}_{q}}{\mathbf{y}_{\rm f}^{\rm p}} 
,
\;\;
q 
\in 
\lbrace 
\mathsf{ls, lmmse} 
\rbrace.
\end{split}
\end{align}

\subsection{ICI Suppression}
\label{subsec:deconv}
In general, the ICI brought on by \ac{PN} can be suppressed by the deconvolution between received signals and \ac{PN} components in the frequency domain~\cite{den07}. In this subsection, we start with a Lemma that provides our idea behind the ICI suppression.
\begin{lemma}
\label{lemma:deconv}
Let $\mathbf{z} \in \mathbb{C}^{N \times 1}$ be the output vector of circular convolution between $\mathbf{x} \in \mathbb{C}^{N \times 1}$ and vector $\mathbf{y} \in \mathbb{C}^{N \times 1}$. Then the deconvolution of $c \mathbf{x}$ from $\mathbf{z}$, where $c \in \mathbb{C}$ is a scalar, is given by\\
\begin{align}
\label{eq:lem_deconv_1}
\mathbf{z} \circledast^{-1} c \mathbf{x}= \frac{1}{c} \mathbf{y}	
\end{align}
\end{lemma}
\begin{proof}
	By the linear property of circular convolution~\cite{opp89}, 
\begin{align}
\label{eq:lem_deconv_2}
\mathbf{z} = \mathbf{x} \circledast \mathbf{y} 
= c \cdot \frac{1}{c} (\mathbf{x} \circledast \mathbf{y})
= c\mathbf{x} \circledast \frac{1}{c}\mathbf{y}.
\end{align}
\end{proof}

Let $\mathbf{f}_{\rm p} \triangleq \alpha \mathbf{p}_{\rm f, \gamma} = \lbrack F_{0,k}, F_{1,k}, \cdots, F_{N-1,k} \rbrack^{\rm T}$ be the length-$N$ PN-affected-channel vector, which has the corresponding coefficients in (\ref{eq:f_p_bar}) for $i \in \mathcal{P}$, and $F_{i,k} = 0$ for $i \in \mathcal{P}^{\rm c}$. From Lemma~\ref{lemma:deconv}, the deconvolution of $\mathbf{f}_{\rm p}$ from $\mathbf{y}_{\rm f}$  yields the effective channel $(1/\alpha)\mathbf{h}_{\rm f}$. In other words, the Toeplitz convolution matrix $\mathbf{F}_{\gamma}$ is converted into the diagonal matrix $\mathbf{H}_{\rm If} = (1/\alpha)\mathbf{H}_{\rm f}$ called the ICI-free channel, which means that the off-diagonal elements causing ICI in $\mathbf{F}_{\gamma}$ can be canceled. The ICI-free channel is represented as
\begin{align}
\label{eq:H_If}
\mathbf{H}_{\rm If} = 
\begin{bmatrix}
\mathbf{H}_{\rm If}^{0} & \mathbf{0}_{N_{\rm cb} \times N_{\rm cb}} & \cdots & \mathbf{0}_{N_{\rm cb} \times N_{\rm cb}}\\
\mathbf{0}_{N_{\rm cb} \times N_{\rm cb}} & \mathbf{H}_{\rm If}^{1} & \ddots & \vdots\\
\vdots & \ddots & \ddots & \mathbf{0}_{N_{\rm cb} \times N_{\rm cb}}\\
\mathbf{0}_{N_{\rm cb} \times N_{\rm cb}} & \cdots & \mathbf{0}_{N_{\rm cb} \times N_{\rm cb}} & \mathbf{H}_{\rm If}^{N_{\rm c}-1}
\end{bmatrix},
\end{align}
where $\mathbf{H}_{\rm If}^{m}=H_{{\rm If}, m}\mathbf{I}_{N_{\rm cb}} \in \mathbb{C}^{N_{\rm cb} \times N_{\rm cb}}$, for $m \in \mathcal{C}$, is the diagonal matrix with coefficient $H_{{\rm If}, m} \triangleq H_{m}/\alpha$. 

The complete ICI elimination shown in (\ref{eq:H_If}) can be achieved under the following assumptions: 1) \ac{PN} components beyond $\gamma$-order are negligible, and 2) perfect PN-affected-channel is estimated. From a practical perspective, we model the PN-affected channel estimate with the estimation error vector $\mathbf{e}_{\rm f, est} \triangleq (1/\alpha)\bar{\mathbf{e}}_{\rm f, est}$ as
\begin{align}
\label{eq:PN_aff_Ch_est_model}
\begin{split}	
\hat{\mathbf{f}}_{\rm p} 
&=
\mathbf{f}_{\rm p} 
+
\bar{\mathbf{e}}_{\rm f, est}
=
\alpha
\big(
\mathbf{p}_{\rm f, \gamma}
+
\mathbf{e}_{\rm f, est}
\big),
\end{split} 
\end{align}
where $\mathbf{e}_{\rm f, est} \triangleq [E_{{\rm f}, 0}^{\rm est}, E_{{\rm f}, 1}^{\rm est},\cdots, E_{{\rm f}, N-1}^{\rm est}]^{\rm T} \in {\mathbb{C}^{N \times 1}}$, $E_{{\rm f}, i}^{\rm est} \neq 0$ for $i \in \mathcal{P}$; otherwise $E_{{\rm f}, i}^{\rm est} = 0$. The PN-affected-channel estimate can be expressed as 
\begin{align}
\label{eq:PNAC_eff_err}
\hat{\mathbf{f}}_{\rm p}
=
\alpha
\big\{
\mathbf{p}_{\rm f}	
+
(
\mathbf{e}_{\rm f, est}
-
\mathbf{e}_{\rm f, app}
)
\big\},
\end{align}
where we define the effective error vector as $\mathbf{e}_{\rm f, eff} \triangleq \mathbf{e}_{\rm f,est} - \mathbf{e}_{\rm f, app} \in \mathbb{C}^{N \times 1}$.

To describe the output vector of deconvolution, we adopt the time-domain representation $\mathbf{g}_{\rm p}$ and $\mathbf{e}_{\rm t, eff}$ of $\hat{\mathbf{f}}_{\rm p}$ and $\mathbf{e}_{\rm f, eff}$, respectively, as
\begin{align}
\label{eq:g_p_vec}
\mathbf{g}_{\rm p} 
=
\sqrt{N}
\mathbf{D}_{N}^{\rm H} {\hat{\mathbf{f}}_{\rm p}}
=
[g_{{\rm p}, 0}, g_{{\rm p}, 1}, \cdots, g_{{\rm p}, N-1}]^{\rm T}
\in 
\mathbb{C}^{N \times 1}, {\text{and}}
\end{align}
\begin{align}
\label{eq:e_t_eff_vec}	
\mathbf{e}_{\rm t, eff} 
=
\sqrt{N}
\mathbf{D}_{N}^{\rm H} {\mathbf{e}}_{\rm f, eff}
=
[E^{\rm eff}_{{\rm t}, 0}, E^{\rm eff}_{{\rm t}, 1}, \cdots, E^{\rm eff}_{{\rm t}, N-1}]^{\rm T}
\in 
\mathbb{C}^{N \times 1},
\end{align}
where $\mathbf{D}_{N}$ refers to the $N \times N$ unitary discrete Fourier transform~(DFT) matrix. The following theorem shows the output vector $\mathbf{y}_{\rm If}$ after the ICI suppression.

\begin{theorem}
\label{thm:deconv_out_vec}
Let $\mathbf{y}_{\rm If} = \lbrack Y_{{\rm If}, 0}, Y_{{\rm If}, 1}, \cdots, Y_{{\rm If}, N-1} \rbrack^{\rm T} \in \mathbb{C}^{N \times 1}$ denote the output vector by deconvolving the PN-affected-channel estimate $\hat{\mathbf{f}}_{\rm p}$ from $\mathbf{y}_{\rm f}$. The signal model of $\mathbf{y}_{\rm If}$ taking into account the approximation error of the \ac{PN} spectrum and the estimation error of the PN-affected channel is given by
\begin{align}
\label{eq:y_If}
\begin{split}
\mathbf{y}_{\rm If}
&= 
\mathbf{y}_{\rm f}
\circledast^{-1}
\hat{\mathbf{f}}_{\rm p}\\
&=
\lbrace
\mathbf{I}_{N}
-
\mathbf{\Upsilon}
\rbrace
\mathbf{H}_{\rm If}\mathbf{x}_{\rm f}
+
\bar{\mathbf{z}}_{\rm f},\\
\end{split}
\end{align}
where
\begin{align}
\label{eq:Upsilon_thm}
\mathbf{\Upsilon}
\triangleq
\alpha
\mathbf{D}_{N}
{\mathbf{G}_{\rm p}}
{\mathbf{E}_{\rm t,eff}}
\mathbf{D}_{N}^{\rm H},
\end{align}
\begin{align}
\label{eq:z_f_bar}
\bar{\mathbf{z}}_{\rm f}
\triangleq
\mathbf{D}_{N}
\mathbf{G}_{\rm p}
\mathbf{D}_{N}^{\rm H}
\mathbf{z}_{\rm f},	
\end{align}
\begin{align}
\label{eq:G_p}
\mathbf{G}_{\rm p} 
\triangleq 
{\rm diag}
\big\{
1/g_{{\rm p}, n}
\big\}_{n = 0}^{N-1},
\end{align}
\begin{align}
\label{eq:E_t_eff_mat}	
\mathbf{E}_{\rm t, eff}
\triangleq
{\rm diag}
\big\{
E^{\rm eff}_{{\rm t}, n}
\big\}_{n = 0}^{N-1}.
\end{align}   
\end{theorem}
\begin{proof}
See Appendix~\ref{append_sec:Thm_deconv_out_vec}. 
\end{proof}
The following lemma provides a constructive proof of the above theorem.
\begin{lemma}
\label{lemma:circulant_mat}
Let	$\mathbf{C} \in \mathbb{C}^{N \times N}$ be a circulant matrix whose first column is $\mathbf{c} = [c_{0}, c_{1}, \cdots, c_{N-1}]^{\rm T}$ and each subsequent column is obtained by a circular shift of the previous column. The circulant matrix $\mathbf{C}$ has eigenvector $\mathbf{d}_{k} = \frac{1}{\sqrt{N}}[1, e^{j 2 \pi k/N}, \cdots, e^{j 2 \pi k(N-1)/N}]^{\rm H}$ for $k = \{0,1,\cdots, N-1\}$, and corresponding eigenvalues
\begin{align}
\label{eq:lem_circ_eig_value}
\lambda_{k} = \sum_{\ell = 0}^{N-1} {c_{\ell}e^{j 2 \pi k\ell/N}},
\end{align}
and can be decomposed as $\mathbf{C} = \mathbf{D                                                                    }_{N} \mathbf{\Lambda} \mathbf{D}_{N}^{\rm H}$, where $\mathbf{D}_{N}$ is N-point unitary DFT matrix and $\mathbf{\Lambda}$ is ${\rm diag} \lbrace \lambda_{k} \rbrace_{k=0}^{N-1}$.
\end{lemma}
\begin{proof}
See~\cite{davis12}. 
\end{proof}

By the expression from Theorem~\ref{thm:deconv_out_vec}, we obtain the signal model to design the ICI-free-channel estimator in Section~\ref{subsec:ICIfree_est}. The effective error incurs the $\mathbf{\Upsilon}$-dependent term in the deconvolved output vector. The impact of the $\mathbf{\Upsilon}$-dependent term is divided into two; one is the distortion of the ICI-free-channel on each subcarrier, and the other is the residual interference. To see this impact, let us rewrite the deconvolution output-vector (\ref{eq:y_If}) as
\begin{align}
\label{eq:y_If_decomp}
\begin{split}
\mathbf{y}_{\rm If}
&=
\lbrace
\mathbf{I}_{N}
-
(
\mathbf{\Upsilon}_{\rm diag}
+
\mathbf{\Upsilon}_{\rm off}
)
\rbrace	
\mathbf{H}_{\rm If}
\mathbf{x}_{\rm f}
+
\bar{\mathbf{z}}_{\rm f}\\
&=
\underbrace{
\lbrace
\mathbf{I}_{N}
-
\mathbf{\Upsilon}_{\rm diag}
\rbrace
\mathbf{H}_{\rm If}
}_{\triangleq \bar{\mathbf{H}}_{\rm If}}
\mathbf{x}_{\rm f}
-
\mathbf{\Upsilon}_{\rm off}
\mathbf{H}_{\rm If}
\mathbf{x}_{\rm f}
+
\bar{\mathbf{z}}_{\rm f},
\end{split}
\end{align}
where $\mathbf{\Upsilon}_{\rm diag}$ is the diagonal matrix with the main diagonal terms of $\mathbf{\Upsilon}$. The diagonal terms are the distorted coefficients by the effective error. As the $\mathbf{\Upsilon}_{\rm off} \triangleq \mathbf{\Upsilon} - \mathbf{\Upsilon}_{\rm diag}$ is its off-diagonal matrix, the $\mathbf{\Upsilon}_{\rm off}\mathbf{H}_{\rm If} \mathbf{x}_{\rm f}$ acts as a residual interference. Notice that, in practice, $\bar{\mathbf{H}}_{\rm If}$ should be estimated to decode the data symbols, which is described in the following subsection.

\subsection{ICI-Free-Channel Estimation}
\label{subsec:ICIfree_est}
The main objective of this subsection is to estimate the diagonal elements of $\bar{\mathbf{H}}_{\rm If}$ by using the CH-dedicated pilot. The following theorem shows that the diagonal terms of $\mathbf{\Upsilon}$ have an identical coefficient, which means that constant channel frequency response over $N_{\rm cb}$ successive subcarriers is still maintained despite the impact of the effective error. 
\begin{theorem}
\label{thm:IFC_distort}
The ICI-free-channel matrix distorted by the effective error is a scaled version of $\mathbf{H}_{\rm If}$ as
\begin{align}
\label{eq:IFC_distort}
\bar{\mathbf{H}}_{\rm If}
=
(1-\bar{\varepsilon}_{\rm cd})
\mathbf{H}_{\rm If},
\;\;\;\;\; 
\bar{\varepsilon}_{\rm cd} 
\in
\mathbb{C},	
\end{align}
where we call $\bar{\varepsilon}_{\rm cd}$ a common distortion coefficient of the ICI-free channel. The $\bar{\varepsilon}_{\rm cd}$ is defined as 
\begin{align}
\label{eq:cd_coeff}
\bar{\varepsilon}_{\rm cd} 
\triangleq
\alpha 
\bigg\{
\frac{1}{N}
\sum_{n=0}^{N-1}
\frac{E_{{\rm t}, n}^{\rm eff}}
{g_{{\rm p}, n}}
\bigg\}.
\end{align} 
\end{theorem}
\begin{proof}
Note that $\mathbf{\Upsilon}$ is a circulant matrix by definition in (\ref{eq:Upsilon_thm}). Thus, $\mathbf{\Upsilon}_{\rm diag}$ can be represented as $\mathbf{\Upsilon}_{\rm diag} = \bar{\varepsilon}_{\rm cd} \mathbf{I}_{N}$. The expression for $\bar{\varepsilon}_{\rm cd}$ can be simply proved by (\ref{eq:Upsilon_thm}).
\end{proof}

The matrix $\bar{\mathbf{H}}_{\rm If}$ has $N_{\rm c}$ diagonal elements defined as $\tilde{\mathbf{h}}_{\rm If} = [(1-\bar{\varepsilon}_{\rm cd})H_{{\rm If}, 0}, (1-\bar{\varepsilon}_{\rm cd})H_{{\rm If}, 1}, \cdots, (1-\bar{\varepsilon}_{\rm cd})H_{{\rm If}, N_{\rm c}-1}]^{\rm T} \in \mathbb{C}^{N_{\rm c} \times 1}$. To estimate $\tilde{\mathbf{h}}_{\rm If}$, one PN-dedicated pilot in $\mathcal{F}_{\rm p}$ can be reused. Hence $(N_{\rm c}-1)$ CH-dedicated pilots are additionally needed. Let $\mathbf{x}_{\rm f}^{\rm c} = [X_{2\gamma}^{\rm p}, X_{1}^{\rm c}, X_{2}^{\rm c}, \cdots, X_{N_{\rm c}-1}^{\rm c}]^{\rm T} \in \mathbb{C}^{N_{\rm c} \times 1}$ be the  pilot vector for the ICI-free-channel estimation. Based on Theorem~\ref{thm:deconv_out_vec}, the output vector $\tilde{\mathbf{y}}_{\rm If} = [Y_{{\rm If}, 2\gamma}, Y_{{\rm If}, N_{\rm cb}}, Y_{{\rm If}, 2N_{\rm cb}}, \cdots, Y_{{\rm If}, (N_{\rm c}-1)N_{\rm cb}} ]^{\rm T} \in \mathbb{C}^{N_{\rm c} \times 1}$, to estimate $\tilde{\mathbf{h}}_{\rm If}$, can be expressed as
\begin{align}
\label{eq:y_tilde}
\begin{split}
\tilde{\mathbf{y}}_{\rm If}
&=
\tilde{\mathbf{D}}_{N}
\lbrace
\mathbf{I}_{N}
-
\alpha
\mathbf{G}_{\rm p}
\mathbf{E}_{\rm t, eff}
\rbrace
\mathbf{D}_{N}^{\rm H}
\mathbf{H}_{\rm If}
\mathbf{x}_{\rm f}
+
\tilde{\mathbf{z}}_{\rm f}\\
&=
\tilde{\mathbf{H}}_{\rm If}
\mathbf{x}_{\rm f}^{\rm c}
+
\tilde{\mathbf{\Upsilon}}_{\rm off}	
\mathbf{H}_{\rm If}
\mathbf{x}_{\rm f}
+
\tilde{\mathbf{z}}_{\rm f},
\end{split}	
\end{align}
where $\tilde{\mathbf{H}}_{\rm If} = {\rm diag} \lbrace (1-\bar{\varepsilon}_{\rm cd}) H_{\rm If,m} \rbrace_{m = 0}^{N_{\rm c}-1}$ is the diagonal matrix with entries from $\tilde{\mathbf{h}}_{\rm If}$ on its main diagonal, $\tilde{\mathbf{\Upsilon}}_{\rm off} \triangleq \tilde{\mathbf{D}}_{N} \lbrace \bar{\varepsilon}_{\rm cd} \mathbf{I}_{N} - \alpha\mathbf{G}_{\rm p} \mathbf{E}_{\rm t, eff} \rbrace \mathbf{D}_{N}^{\rm H}$, $\tilde{\mathbf{z}}_{\rm f} \triangleq \tilde{\mathbf{D}}_{N} \mathbf{G}_{\rm p} \mathbf{D}_{N}^{\rm H} \mathbf{z}_{\rm f}$, and $\tilde{\mathbf{D}}_{N} \in \mathbb{C}^{N_{\rm c} \times N}$ is a semi-unitary matrix formed by rows $m = \lbrace N_{\gamma}, N_{\rm cb}, 2N_{\rm cb}, \cdots, (N_{\rm c}-1)N_{\rm cb} \rbrace$ of $\mathbf{D}_{N}$. The second equation on the right side in (\ref{eq:y_tilde}) represents the expression by separating residual interference, \ie, $\tilde{\mathbf{\Upsilon}}_{\rm off}\mathbf{H}_{\rm If}\mathbf{x}_{\rm f}$.

We employ the LS and LMMSE estimators for the ICI-free channel. The LS and LMMSE estimators are, respectively, (See Appendix~\ref{append_sec:LMMSE_IFC} for $\mathbf{V}_{\mathsf{lmmse}}$)
\begin{align}
\label{eq:V_ls}
\mathbf{V}_{\mathsf{ls}}
=
(
\mathbf{X}_{\rm f}^{\rm c}	
)^{-1}, \text{and}
\end{align}
\begin{align}
\label{eq:V_lmmse}
\mathbf{V}_{\mathsf{lmmse}}
=
\frac{
1 
- 
\sigma_{\varepsilon}^{2}
}
{1 + (1/\mathsf{SNR})}
(
\mathbf{X}_{\rm f}^{\rm c}
)^{\rm H},
\end{align}
where $\mathbf{X}_{\rm f}^{\rm c} \in \mathbb{C}^{N_{\rm c} \times N_{\rm c}}$ is a diagonal matrix with entries from $\mathbf{x}_{\rm f}^{\rm c}$ on its main diagonal and $\sigma_{\varepsilon}^{2}$ is the variance of the effective error. The ICI-free-channel estimate $\hat{\mathbf{h}}_{\rm If} \in \mathbb{C}^{N_{\rm c} \times 1}$, which becomes the last estimate for decoding the transmitting data, is given by
\begin{align}
\label{eq:h_If_est}
\hat{\mathbf{h}}_{{\rm {If}}, q} 
= 
\mathbf{V}_{q}
\tilde{\mathbf{y}}_{\rm {If}},
\;\;\;
q
\in
\lbrace
\mathsf{ls, lmmse}
\rbrace
\end{align}  

\section{Normalized Mean Squared-Error Analysis}
\label{sec:NMSE_analysis}
\ac{NMSE} has been widely used, as a performance metric, to evaluate channel estimators for fading environments, e.g., spatially- or temporally-correlated channels~\cite{li02, yin13, shariati14}. Furthermore, Hamila \emph{et al.}~\cite{hamila16} and Liu \emph{et al.}~\cite{liu17} have derived closed-from expressions of the \ac{NMSE} (a modified \ac{NMSE} in~\cite{liu17}). These expressions provide useful insights into channel estimation performance, according to system parameters. This section presents an \ac{NMSE} analysis of \ac{PN}-affected-\hspace{0.1em}/\hspace{0.1em}\ac{ICI}-free-channel estimation. For the \ac{NMSE} analysis, we offer a simple closed-form expression for their respective \acp{NMSE}, based on the assumption of \ac{PN} modeled by a Wiener process. It helps in understanding the NMSE behavior in low and high \ac{SNR} regimes. In the following expressions, the channel coherence matrix of $\lbrace H_{k} \rbrace_{k = 0}^{N_{\rm c}-1}$ has an identity matrix, \ie, $\mathbf{R}_{\mathbf{hh}} = \mathbf{I}_{N_{\rm c}}$, by the coherence block model given in Section~\ref{sec:sys}.

\subsection{NMSE of PN-affected channel}
\label{subsec:NMSE_PNAC}
The NMSE for PN-affected-channel estimation is defined as
\begin{align}
\label{eq:NMSE_PNAC}
{\mathsf{NMSE}}_{{\rm{p}},q} 
\triangleq 
\frac{\mathbb{E} [\|\hat{\mathbf{f}}_{\bar{\rm p}, q}-\mathbf{f}_{\bar{\rm p}}\|_{2}^2]}
{\mathbb{E} [\|\mathbf{f}_{\bar{\rm p}}\|_{2}^2]},
\;\;
q 
\in 
\lbrace 
{\mathsf{ls}}, {\mathsf{lmmse}}
\rbrace.
\end{align}
From (\ref{eq:NMSE_PNAC}), we derive the NMSEs of \ac{LS} and \ac{LMMSE} PN-affected-channel estimators, respectively, as
\begin{align}
\label{eq:NMSE_PNAC_LS}
\begin{split}
{\mathsf{NMSE}}_{\rm{p},\mathsf{ls}}
&=	
\frac{
\mathbb{E}
\lbrace
\| 
\mathbf{y}_{\rm f}^{\rm p}
-
\mathbf{f}_{\bar{\rm p}} 
\|_{2}^{2}
\rbrace
}
{
\mathbb{E}
\lbrace
\|
\mathbf{f}_{\bar{\rm p}}
\|_{2}^{2}
\rbrace
}
\\
&=
\frac{
{\rm tr}
\lbrace
\mathbf{R}_{\mathbf{ici}}^{\gamma}
+
(1/\mathsf{SNR})
\mathbf{I}_{N_{\rm p}}
\rbrace
}
{
{\rm tr}
\lbrace
\mathbf{R}_{\mathbf{pp}}^{\gamma}
\rbrace
}, \text{and}
\end{split}
\end{align}
\begin{align}
\label{eq:NMSE_PNAC_LMMSE}
\begin{split}
{\mathsf{NMSE}}_{\rm{p},\mathsf{lmmse}}
\hspace{-0.75mm}
&=
\hspace{-0.75mm}
\frac{
\mathbb{E}
\lbrace
\|
\hat{\mathbf{f}}_{\bar{\rm p}, {\mathsf{lmmse}}}
-
\mathbf{f}_{\bar{\rm p}}
\|_{2}^{2}
\rbrace
}
{
\mathbb{E}
\lbrace
\|
\mathbf{f}_{\bar{\rm p}}
\|_{2}^{2}
\rbrace
}
\\
\vspace{0.2mm}
&=
\hspace{-0.75mm}
1
\hspace{-1.05mm}
-
\hspace{-1.05mm}
\frac{
{\rm tr}
\big\{ 
\mathbf{R}_{\mathbf{pp}}^{\gamma} 
\lbrace
\mathbf{R}_{\mathbf{pp}}^{\gamma} 
\hspace{-0.55mm}
+ 
\hspace{-0.55mm}
\mathbf{R}_{\mathbf{ici}}^{\gamma} 
\hspace{-0.55mm}
+ 
\hspace{-0.55mm}
(\frac{1}{\mathsf{SNR}}) \mathbf{I}_{N_{\rm p}}
\rbrace^{-1}
\mathbf{R}_{\mathbf{pp}}^{\gamma}
\big\}
}
{
{\rm tr}
\lbrace
\mathbf{R}_{\mathbf{pp}}^{\gamma}
\rbrace
},
\end{split}
\end{align}
In (\ref{eq:NMSE_PNAC_LS}) and (\ref{eq:NMSE_PNAC_LMMSE}), the matrix $\mathbf{R}_{\mathbf{pp}}^{\gamma}$ is a submatrix of the autocorrelation matrix
\begin{align}
\label{eq:R_pp}	
\mathbf{R}_{\mathbf{pp}}
=
\mathbb{E} 
\lbrace 
\mathbf{p}_{\rm f}
\mathbf{p}_{\rm f}^{\rm H} 
\rbrace
=
\frac{1}{N} 
\mathbf{D}_{N}
\mathbf{\Psi}^{\rm T}
\mathbf{D}_{N}^{\rm H}
\in 
\mathbb{C}^{N \times N},
\end{align}
where $\mathbf{\Psi}$ has entries of $\psi_{m,n} \triangleq e^{-\pi \beta |m-n| T_{\rm s}}$ for $m,n \in \lbrace 0,1,\cdots, N-1 \rbrace$. The entries in $\mathbf{R}_{\mathbf{ici}}^{\gamma}$ can be defined as a function of autocorrelation coefficients in $\mathbf{R}_{\mathbf{pp}}$. (See Appendix~\ref{append_sec:R_pp_and_R_ici} for the autocorrelation coefficients of $\mathbf{R}_{\mathbf{pp}}^{\gamma}$ and $\mathbf{R}_{\mathbf{ici}}^{\gamma}$)
\begin{figure*}[t!]
\centering
\subfigure[$\beta = 500$]
{
\includegraphics[width=3.42in]{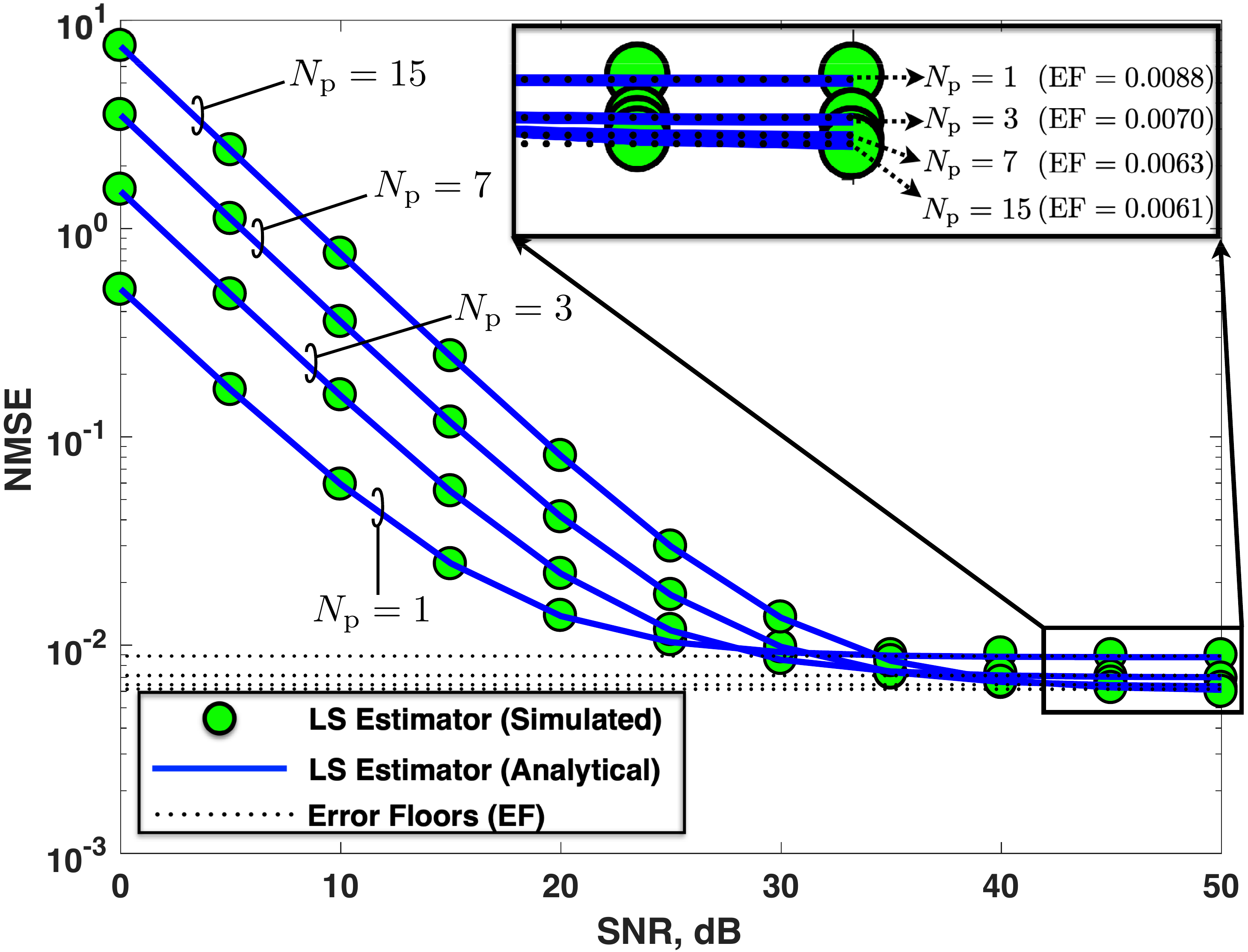}
\label{fig:NMSE_PNAC_LS_500}
}
\subfigure[$\beta = 5000$]
{
\includegraphics[width=3.42in]{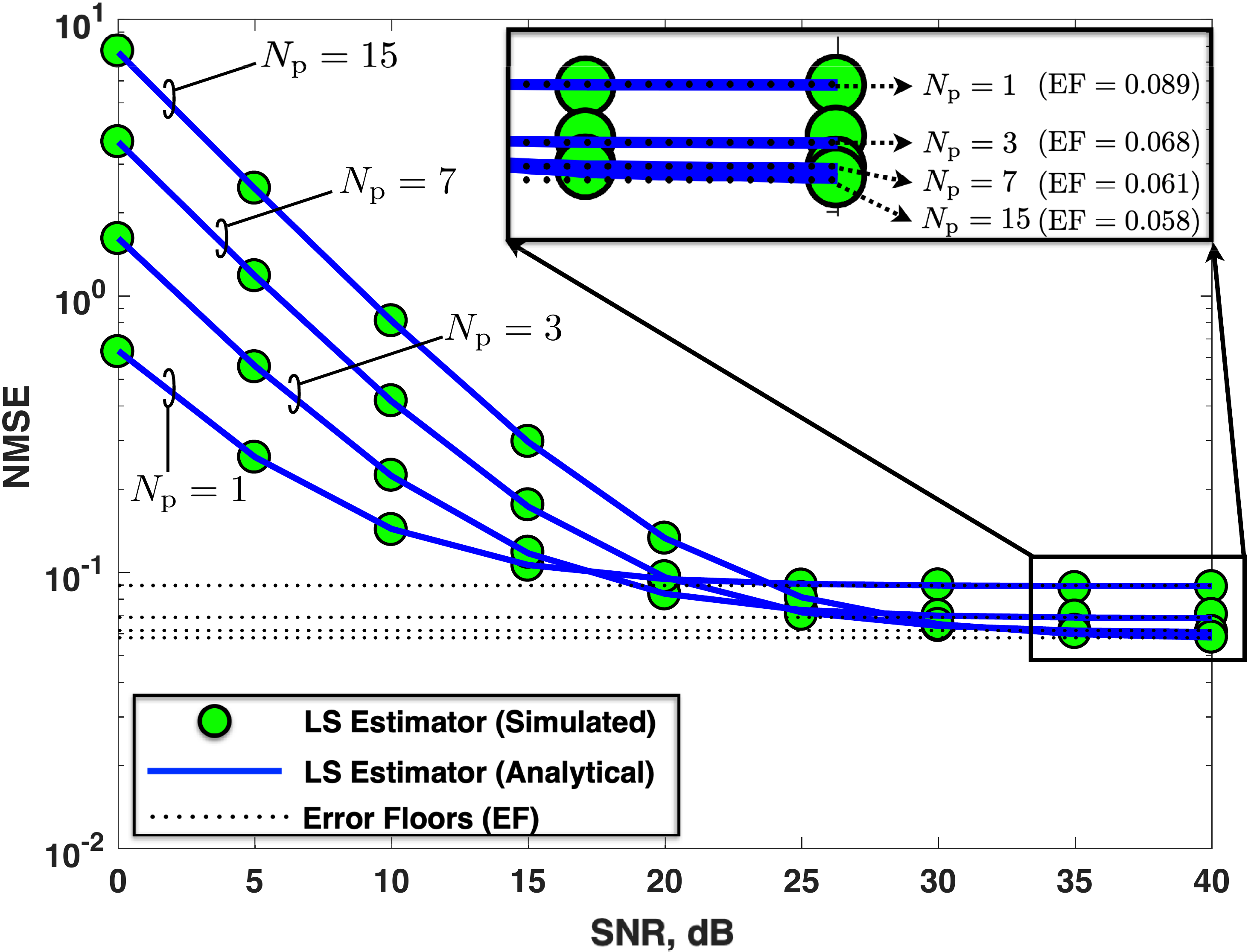}
\label{fig:NMSE_PNAC_LS_5000}
}
\caption{NMSE of LS PN-affected-channel estimator as a function of $\mathsf{SNR}$ for $\beta \in \lbrace 500, 5000\rbrace$ and $\mathcal{P}_{\rm d} = \lbrace 1, 3, 7, 15 \rbrace$. Also shown are the error floors corresponding to the elements in $\mathcal{P}_{\rm d}$. The error floors are obtained by the NMSE expression in~(\ref{eq:NMSE_PNAC_LS}) with $\mathsf{SNR}= \infty$.}
\label{fig:NMSE_PNAC_LS}
\vspace*{-0.25 cm}
\end{figure*}
\begin{figure*}[t!]
\centering
\subfigure[$\beta = 500$]
{
\includegraphics[width=3.42in]{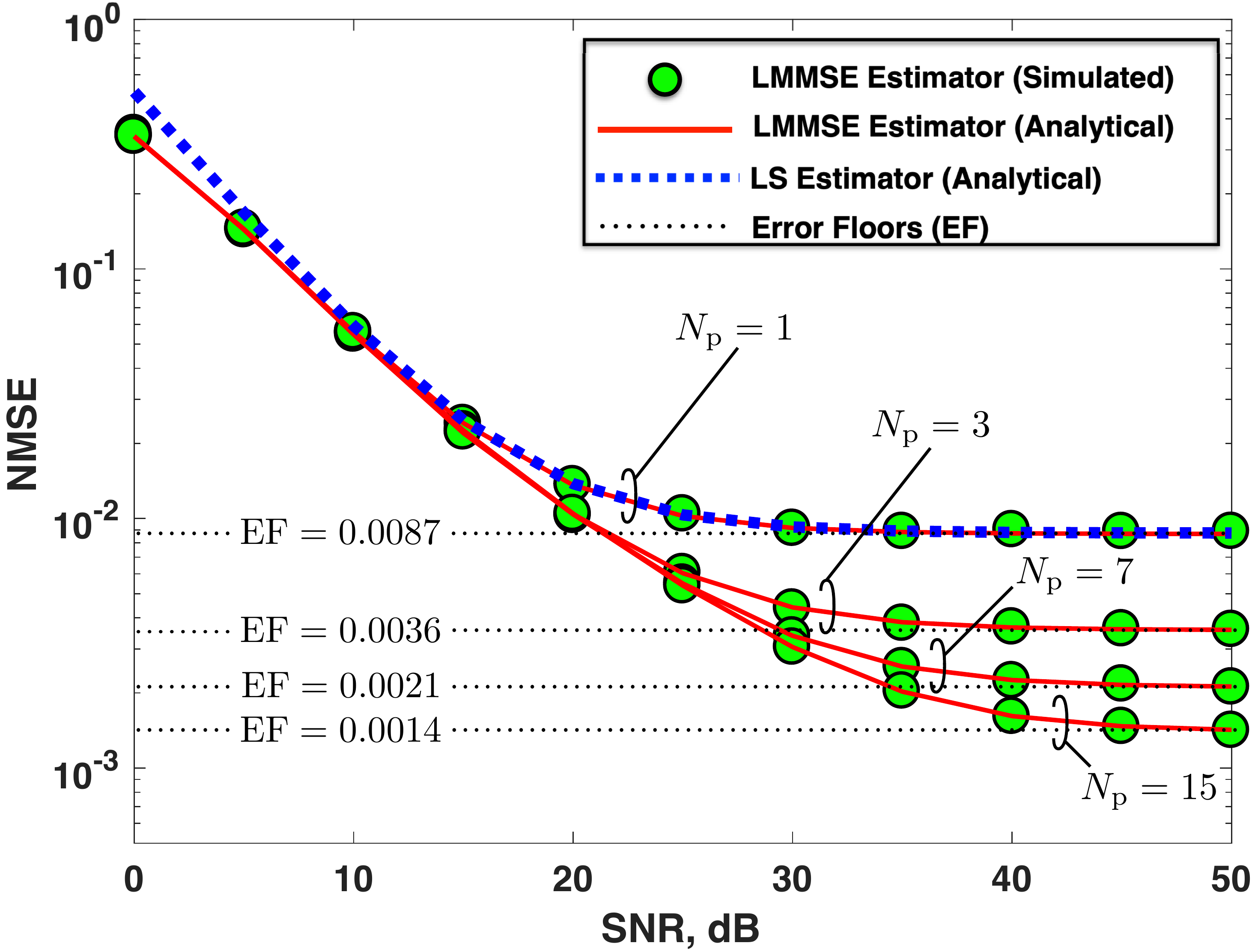}
\label{fig:NMSE_PNAC_LMMSE_500}
}
\subfigure[$\beta = 5000$]
{
\includegraphics[width=3.42in]{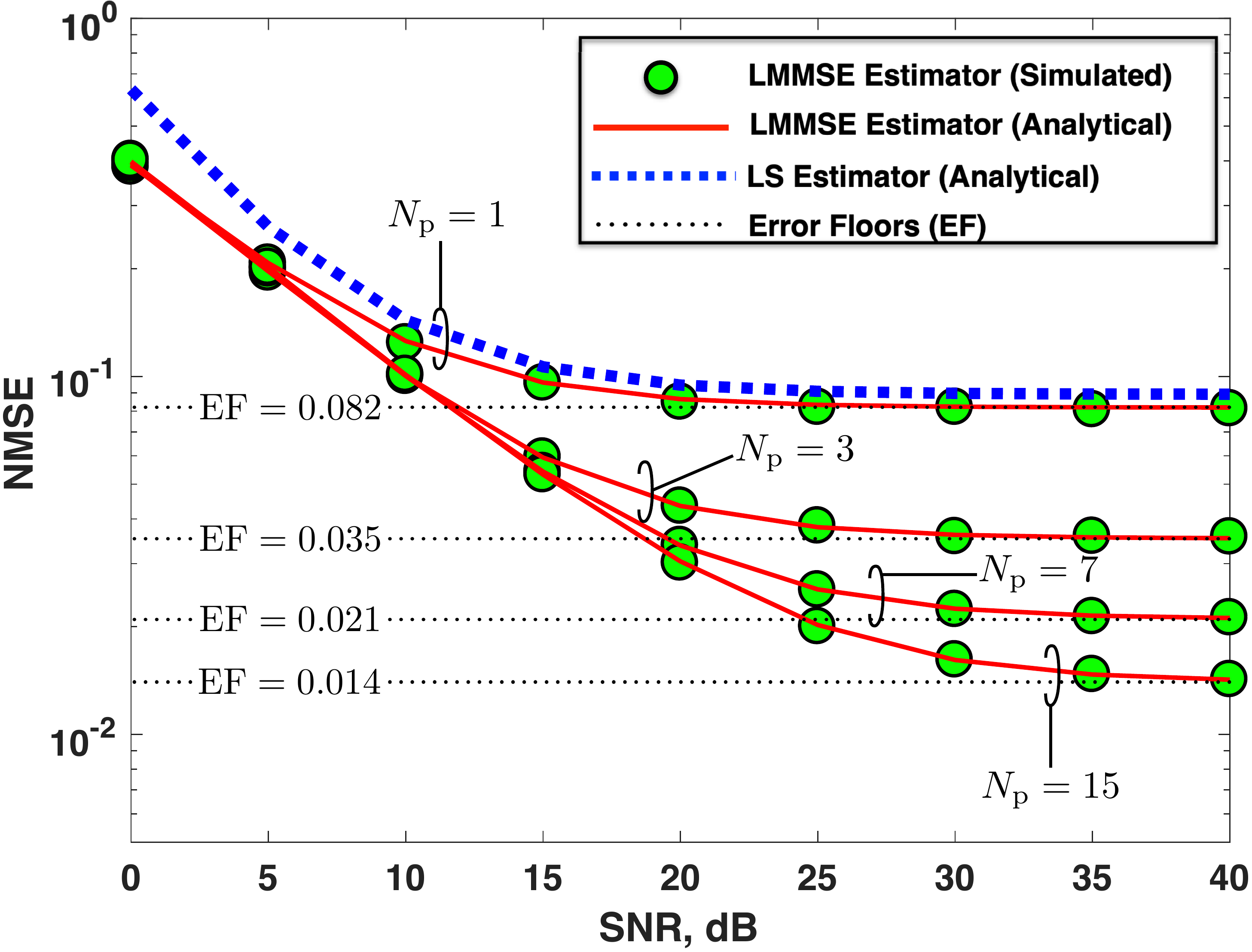}
\label{fig:NMSE_PNAC_LMMSE_5000}
}
\caption{NMSE of LMMSE PN-affected-channel estimator as a function of $\mathsf{SNR}$ for $\beta \in \lbrace 500, 5000\rbrace$ and $\mathcal{P}_{\rm d} = \lbrace 1, 3, 7, 15 \rbrace$. For comparison, the NMSE curve of LS estimator corresponding to $N_{\rm p} = 1$~(dotted blue line) is included. Also shown are the error floors corresponding to the elements in $\mathcal{P}_{\rm d}$. The error floors are obtained by the NMSE expression in~(\ref{eq:NMSE_PNAC_LMMSE}) with $\mathsf{SNR}= \infty$.}
\label{fig:NMSE_PNAC_LMMSE}
\vspace*{-0.25 cm}
\end{figure*}

\begin{remark}\normalfont
{\textbf{(NMSE behavior for PN-affected-channel estimation):}}
\label{rem:PNAC_NMSE_Beh}
The LMMSE estimator with the second-order statistics of \ac{PN} spectrum achieves better NMSE performance as increasing $N_{\rm p}$. One remarkable observation is that the LS estimator has different NMSE behavior depending on the SNR range. At low SNRs, the NMSE increases with $N_{\rm p}$ while it is the opposite at high SNRs. To look at the NMSE in the low and high SNR regimes, we approximate the NMSE of LS estimator (\ref{eq:NMSE_PNAC_LS}) as follows.
\begin{align}
\label{RMK_eq:NMSE_PNAC_Approx}
{\mathsf{NMSE}}_{{\rm p}, {\mathsf{ls}}}
\approx
{\mathsf{NMSE}}_{{\rm p}, {\mathsf{ls}}}^{\mathsf{app}}
=
\frac{
1
-
\mathsf{P}_{\rm dom}
+
N_{\rm p}
/
\mathsf{SNR}
}
{\mathsf{P}_{\rm dom}},	
\end{align}
where ${\mathsf{P}}_{\rm dom} \triangleq \mathbb{E} \big\{ {\sum_{i \in \mathcal{P}}}{\| P_{i} \|_{2}^{2}} \big\}$ as the power sum of the $N_{\rm p}$ dominant \ac{PN} components. The NMSE in the low and high SNR regimes, respectively, are
\begin{align}
\label{RMK_eq:NMSE_PNAC_HighSNR}
\lim_{{\mathsf{SNR}}\to\infty} 
{\mathsf{NMSE}}_{{\rm p}, \mathsf{ls}}^{\mathsf{app}}
= 
\frac
{1-{\mathsf{P}}_{\rm dom}}
{{\mathsf{P}}_{\rm dom}}, \text{and}
\end{align}
\begin{align}
\label{RMK_eq:NMSE_PNAC_LowSNR}
\lim_{{\mathsf{SNR}}\to 0} 
{\mathsf{NMSE}}_{{\rm p}, \mathsf{ls}}^{\mathsf{app}}
= 
\frac
{N_{\rm p}/{\mathsf{P}}_{\rm dom}}
{\mathsf{SNR}}.
\end{align}
The NMSE at high SNRs (\ref{RMK_eq:NMSE_PNAC_HighSNR}) obviously decreases with $N_{\rm p}$. For the low SNR regime, let us define the numerator in (\ref{RMK_eq:NMSE_PNAC_LowSNR}) as $f(\gamma) \triangleq N_{\rm p}/{\mathsf{P}_ {\rm dom}}$. This is an increasing function of the approximation order $\gamma$, \ie, $f(\gamma)' > 0$ for all $\gamma \geq 1$, translating into an NMSE degradation as $N_{\rm p}$ increases.
\end{remark}
 
To validate our analysis, we compare the NMSE expressions for LS/LMMSE PN-affected-channel estimation (\ref{eq:NMSE_PNAC_LS}) and (\ref{eq:NMSE_PNAC_LMMSE}) with the simulation result in Figs.~\ref{fig:NMSE_PNAC_LS}-\ref{fig:NMSE_PNAC_LMMSE}. For the numerical evaluation, the following parameters\footnote{In the 3GPP standard, the \SI{245.76}{\MHz} is defined as a sampling frequency, and the actual transmission bandwidth is less than the sampling frequency because the transmit data symbol is not fully allocated on the available subcarriers. We assumed that the sampling frequency and the bandwidth are equal in this paper.} are assumed: $N=4096$, $B = \SI{245.76}{\MHz}$, $\Delta f = \SI{60}{\kHz}$, which corresponds to one \ac{3GPP} \ac{NR} signaling resource block to support communication at mmWave frequencies~\cite{38_211}. Also, we consider the set of dominant \ac{PN} components $\mathcal{P}_{\rm d} = \lbrace 1, 3, 7,  15 \rbrace$ and two kinds of 3-\si{\decibel} linewidth  $\beta \in \lbrace 500, 5000 \rbrace$ (\si{Hz}) as LO parameters. The \ac{PN} model that we adopt for the numerical evaluation is illustrated in \fig\ref{fig:pn_quant_anal}. Both have severe \ac{PN} spectrum compared to the one in conventional transceivers~\cite{zaidi16}. Unless otherwise stated, the same settings are assumed for numerical evaluation in this paper. As shown in Figs.~\ref{fig:NMSE_PNAC_LS}-\ref{fig:NMSE_PNAC_LMMSE}, the agreement is excellent for all SNR and $N_{\rm p}$ values. Furthermore, it shows that the NMSE behavior follows the analysis in Remark~\ref{rem:PNAC_NMSE_Beh}.
\begin{figure}[t!]
\centering
\subfigure[Phase-noise trajectories (time domain)]
{
\includegraphics[width=3.42in]{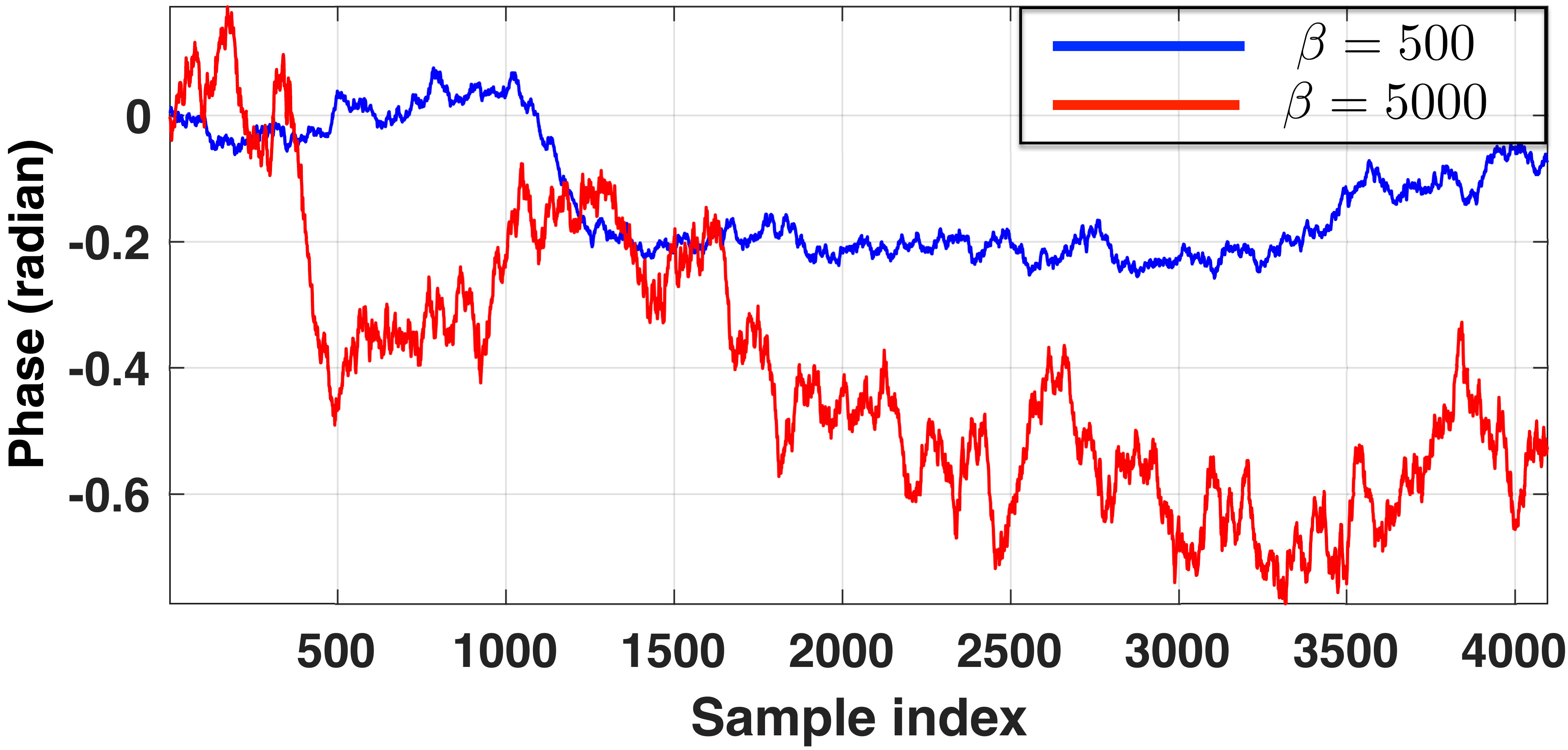}
\label{fig:pn_time}
}
\subfigure[Power spectrum density (frequency domain)]
{
\includegraphics[width=3.42in]{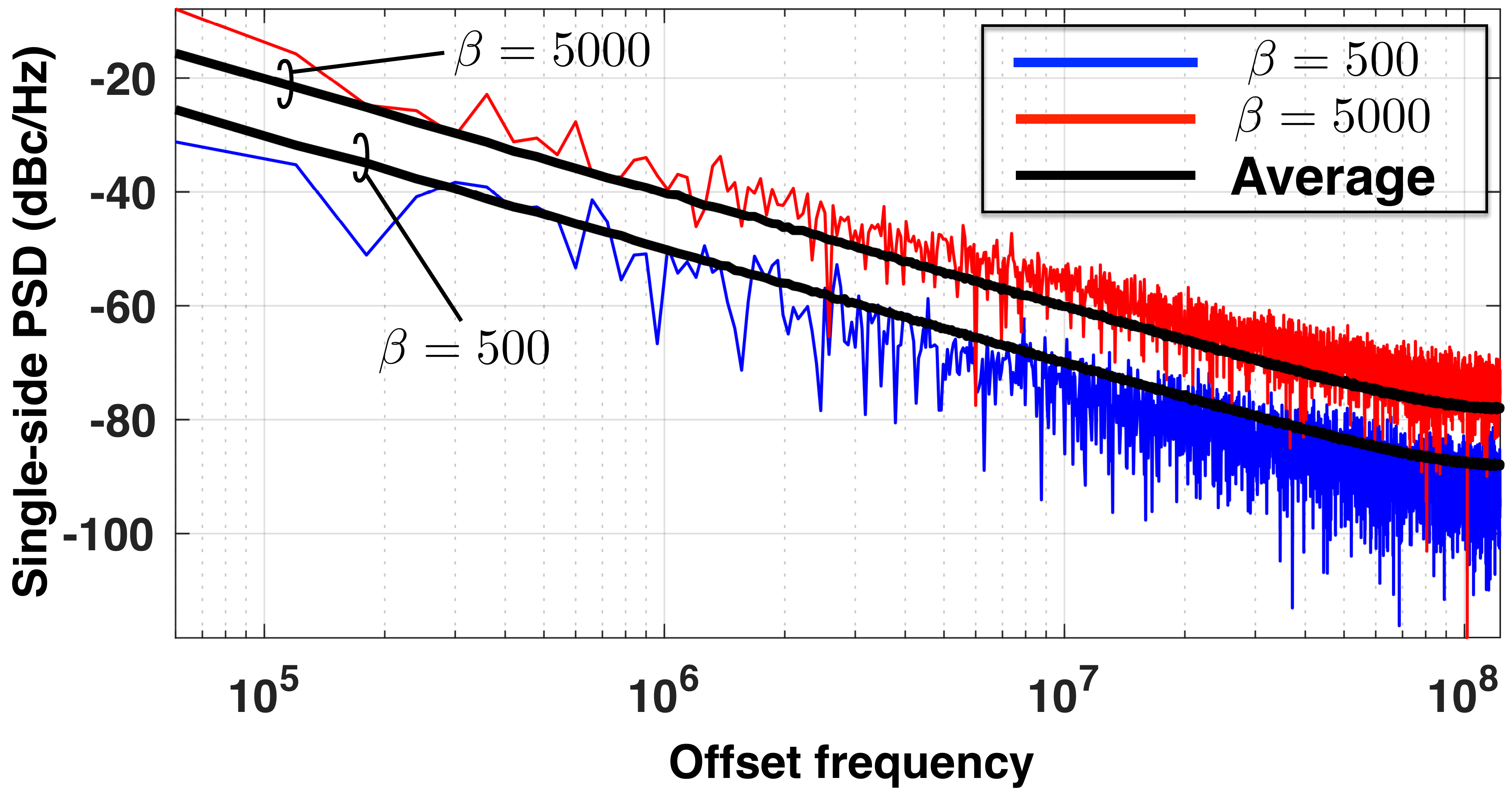}
\label{fig:pn_freq}
}
\caption{Phase-noise trajectories (one OFDM symbol with $N = 4096$) and power spectrum density of free running oscillators, with $\beta = \lbrace 500, 5000 \rbrace$, respectively. (simulated)}
\label{fig:pn_quant_anal}
\end{figure}

\subsection{NMSE of ICI-free channel}
\label{subsec:NMSE_IFC}
The NMSE for ICI-free-channel estimation is
\begin{align}
\label{eq:NMSE_IFC}
{\mathsf{NMSE}}_{{\rm c}, q}
\triangleq 
\frac{\mathbb{E} [\|\hat{\mathbf{h}}_{{\rm If}, q}
-
\tilde{\mathbf{h}}_{\rm If}\|_{2}^2]}
{\mathbb{E} [\| \tilde{\mathbf{h}}_{\rm If}\|_{2}^2]}, 
\;\;
q 
\in 
\lbrace 
{\mathsf{ls}}, {\mathsf{lmmse}}
\rbrace.
\end{align}
From (\ref{eq:NMSE_IFC}), the NMSEs of the LS and LMMSE ICI-free-channel estimators can be derived, respectively, as
\begin{align}
\label{eq:NMSE_IFC_LS}
\begin{split}
{\mathsf{NMSE}}_{{\rm c}, \mathsf{ls}}
\hspace{-0.8mm}
&=
\hspace{-0.8mm}
\frac{
{\rm tr}
\lbrace
\mathbf{R}_
{\tilde{\mathbf{y}}\tilde{\mathbf{y}}}
-
\mathbf{R}_
{\tilde{\mathbf{h}}\tilde{\mathbf{h}}}
\rbrace
}
{
{\rm tr}
\lbrace
\mathbf{R}_
{\tilde{\mathbf{h}}\tilde{\mathbf{h}}}
\rbrace
}\\
&=
\hspace{-0.8mm}
\frac{
{\rm tr}
\big\{
\bar{G}
\lbrace
1
+
(1/\mathsf{SNR})
\rbrace
\mathbf{I}_{N_{\rm c}}
-
\bar{G}
(
1
-
\sigma_{\varepsilon}^{2}
)
\mathbf{I}_{N_{\rm c}}
\big\}
}
{
{\rm tr}
\lbrace
\bar{G}
(
1
-
\sigma_{\varepsilon}^{2}
)
\mathbf{I}_{N_{\rm c}}
\rbrace
}
\\
&=
\hspace{-0.8mm}
\frac
{
N_{\rm c}
\bar{G}
\lbrace
1+
(1/\mathsf{SNR})
-
1+
\sigma_{\varepsilon}^{2}
\rbrace
}
{
N_{\rm c}
\bar{G}
(
1-
\sigma_{\varepsilon}^{2}
)
}
\hspace{-0.8mm}
=
\hspace{-0.8mm}
\frac{
\frac{1}{\mathsf{SNR}}
+
\sigma_{\varepsilon}^{2}
}
{
1-
\sigma_{\varepsilon}^{2}
}, \text{and}
\end{split}	
\end{align} 
\begin{align}
\label{eq:NMSE_IFC_LMMSE}
\begin{split}
{\mathsf{NMSE}}_{{\rm c}, \mathsf{lmmse}}
\hspace{-0.8mm}
&=
\hspace{-0.8mm}
\frac{
{\rm tr}
\big\{
\mathbf{R}_{\tilde{\mathbf{h}}\tilde{\mathbf{h}}}
-
\mathbf{R}_{\tilde{\mathbf{h}}\tilde{\mathbf{y}}}
\mathbf{R}_{\tilde{\mathbf{y}}\tilde{\mathbf{y}}}^{-1}
\mathbf{R}_{\tilde{\mathbf{y}}\tilde{\mathbf{h}}}
\big\}
}
{
{\rm tr}
\lbrace
\mathbf{R}_{\tilde{\mathbf{h}}\tilde{\mathbf{h}}}
\rbrace
}\\
&=
\hspace{-0.8mm}
\frac{
{\rm tr}
\Big\{
\bar{G}
(
1
-
\sigma_{\varepsilon}^2
)
\mathbf{I}_{N_{\rm c}}
-
\frac{
\lbrace
\bar{G}
(
1
-
\sigma_{\varepsilon}^2
)
\rbrace^{2}
}
{\bar{G}
\lbrace
1 + (1/\mathsf{SNR})
\rbrace
}
(\mathbf{X}_{\rm f}^{\rm c})^{\rm H}
\mathbf{X}_{\rm f}^{\rm c}
\Big\}
}
{
{\rm tr}
\lbrace
\bar{G}
(1-\sigma_{\varepsilon}^{2})
\mathbf{I}_{N_{\rm c}}
\rbrace
}
\\
&=
\hspace{-0.8mm}
\frac{
{\rm tr}
\Big\{
\bar{G}
(
1
-
\sigma_{\varepsilon}^2
)
\Big\{
1
-
\frac{
\bar{G}
(
1
-
\sigma_{\varepsilon}^2
)
}
{\bar{G}
\lbrace
1 + (1/\mathsf{SNR})
\rbrace
}
\Big\}
\mathbf{I}_{N_{\rm c}}
\Big\}
}
{
{\rm tr}
\lbrace
\bar{G}
(1-\sigma_{\varepsilon}^{2})
\mathbf{I}_{N_{\rm c}}
\rbrace
}\\
&=
\hspace{-0.8mm}
\frac{
N_{\rm c}
\bar{G}
(
1
-
\sigma_{\varepsilon}^2
)
\Big\{
1
-
\frac{
1
-
\sigma_{\varepsilon}^2
}
{
\lbrace
1 + (1/\mathsf{SNR})
\rbrace
}
\Big\}
}
{
N_{\rm c}
\bar{G}
(1-\sigma_{\varepsilon}^{2})
}\\
&=
\hspace{-0.8mm}
\frac{
\sigma_{\varepsilon}^{2}
+
(1/\mathsf{SNR})
}
{1+
(1/\mathsf{SNR})
}
=
\frac
{
1
+
\sigma_{\varepsilon}^{2}
\mathsf{SNR}
}
{
1
+
\mathsf{SNR}
},
\end{split}
\end{align}
where ${\mathbf{R}}_{\tilde{\mathbf{y}}\tilde{\mathbf{y}}} \triangleq \mathbb{E}\lbrace \tilde{\mathbf{y}}_{\rm If}\tilde{\mathbf{y}}_{\rm If}^{\rm H}\rbrace$, ${\mathbf{R}}_{\tilde{\mathbf{h}}\tilde{\mathbf{h}}} \triangleq \mathbb{E}\lbrace \tilde{\mathbf{h}}_{\rm If}\tilde{\mathbf{h}}_{\rm If}^{\rm H}\rbrace$, and ${\mathbf{R}}_{\tilde{\mathbf{h}}\tilde{\mathbf{y}}} \triangleq \mathbb{E}\lbrace \tilde{\mathbf{h}}_{\rm If}\tilde{\mathbf{y}}_{\rm If}^{\rm H}\rbrace$. Both NMSE expressions (\ref{eq:NMSE_IFC_LS}) and (\ref{eq:NMSE_IFC_LMMSE}) can be formulated by only the average SNR and the effective-error variance.
\begin{figure*}[t!]
\centering
\subfigure[$\beta = 500$]
{
\includegraphics[width=3.4in]{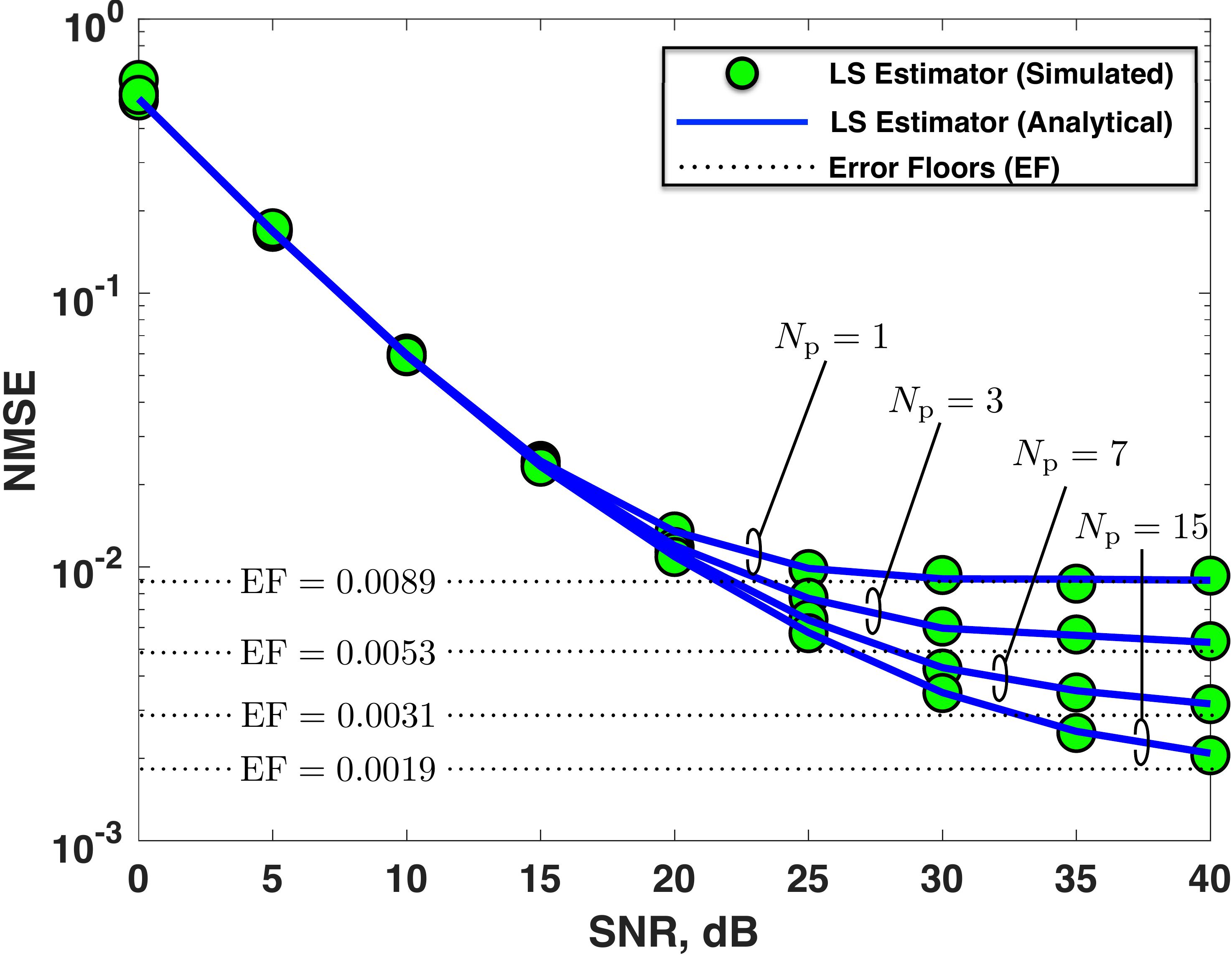}
\label{fig:NMSE_IFC_LS_500}
}
\subfigure[$\beta = 5000$]
{
\includegraphics[width=3.4in]{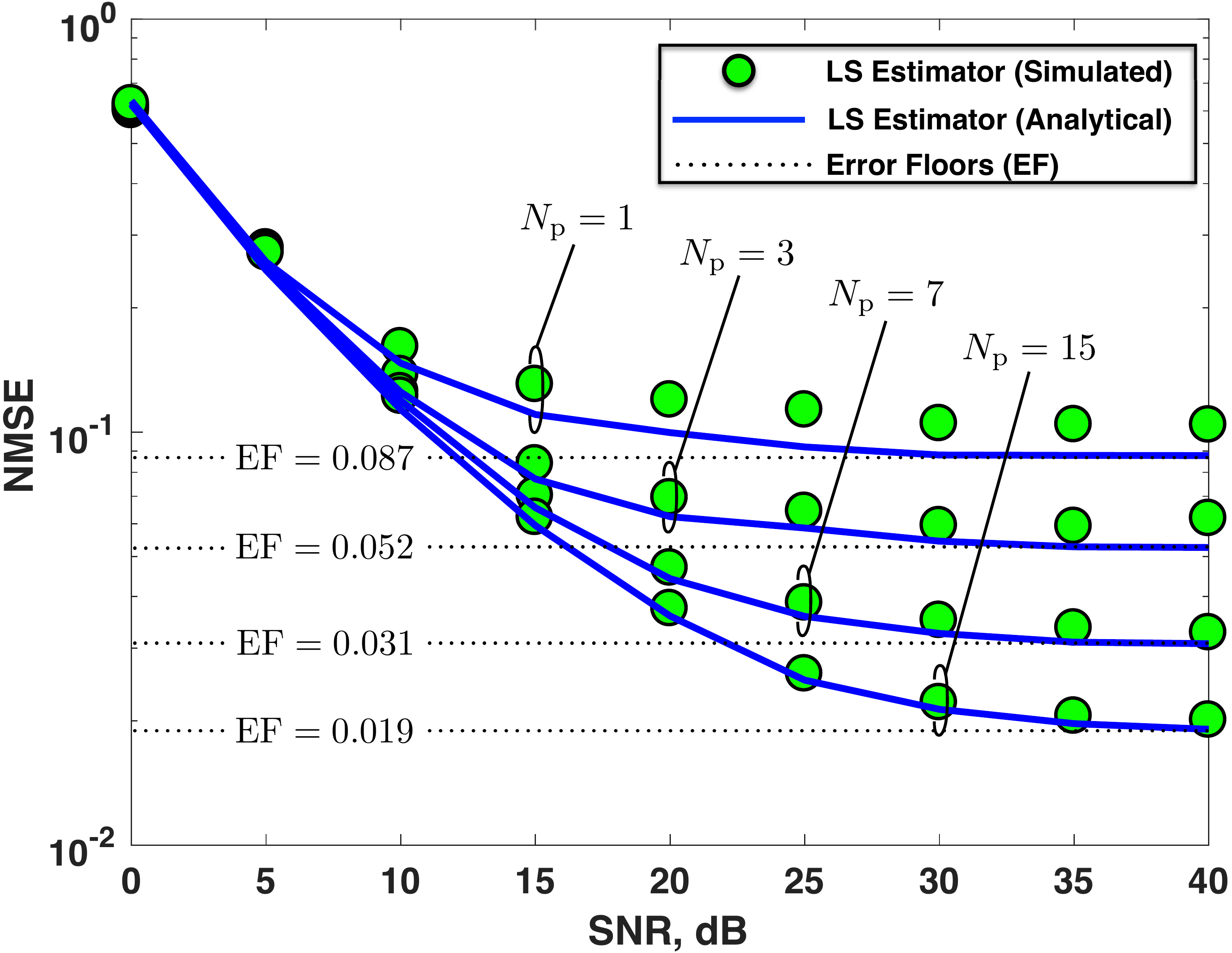}
\label{fig:NMSE_IFC_LS_5000}
}
\caption{NMSE of LS ICI-free-channel estimator as a function of $\mathsf{SNR}$ for $\beta \in \lbrace 500, 5000\rbrace$ and $\mathcal{P}_{\rm d} = \lbrace 1, 3, 7, 15 \rbrace$. Also shown are the error floors corresponding to the elements in $\mathcal{P}_{\rm d}$. The error floors are obtained by the NMSE expression in~(\ref{eq:NMSE_IFC_LS}) with $\mathsf{SNR}= \infty$.}
\label{fig:NMSE_IFC_LS}
\vspace*{-0.25 cm}
\end{figure*}
\begin{figure*}[t!]
\centering
\subfigure[$\beta = 500$]
{
\includegraphics[width=3.4in]{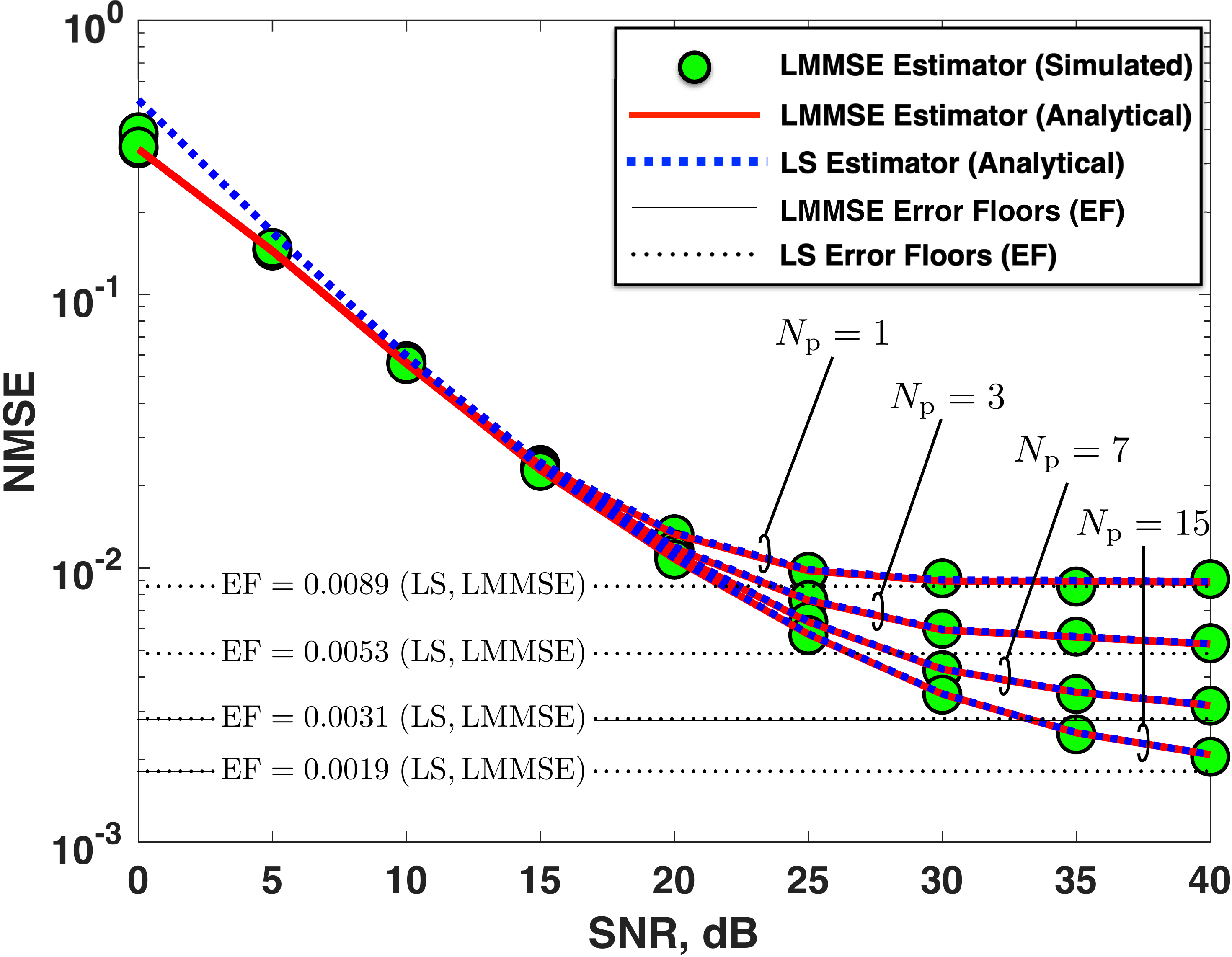}
\label{fig:NMSE_IFC_LMMSE_500}
}
\subfigure[$\beta = 5000$]
{
\includegraphics[width=3.4in]{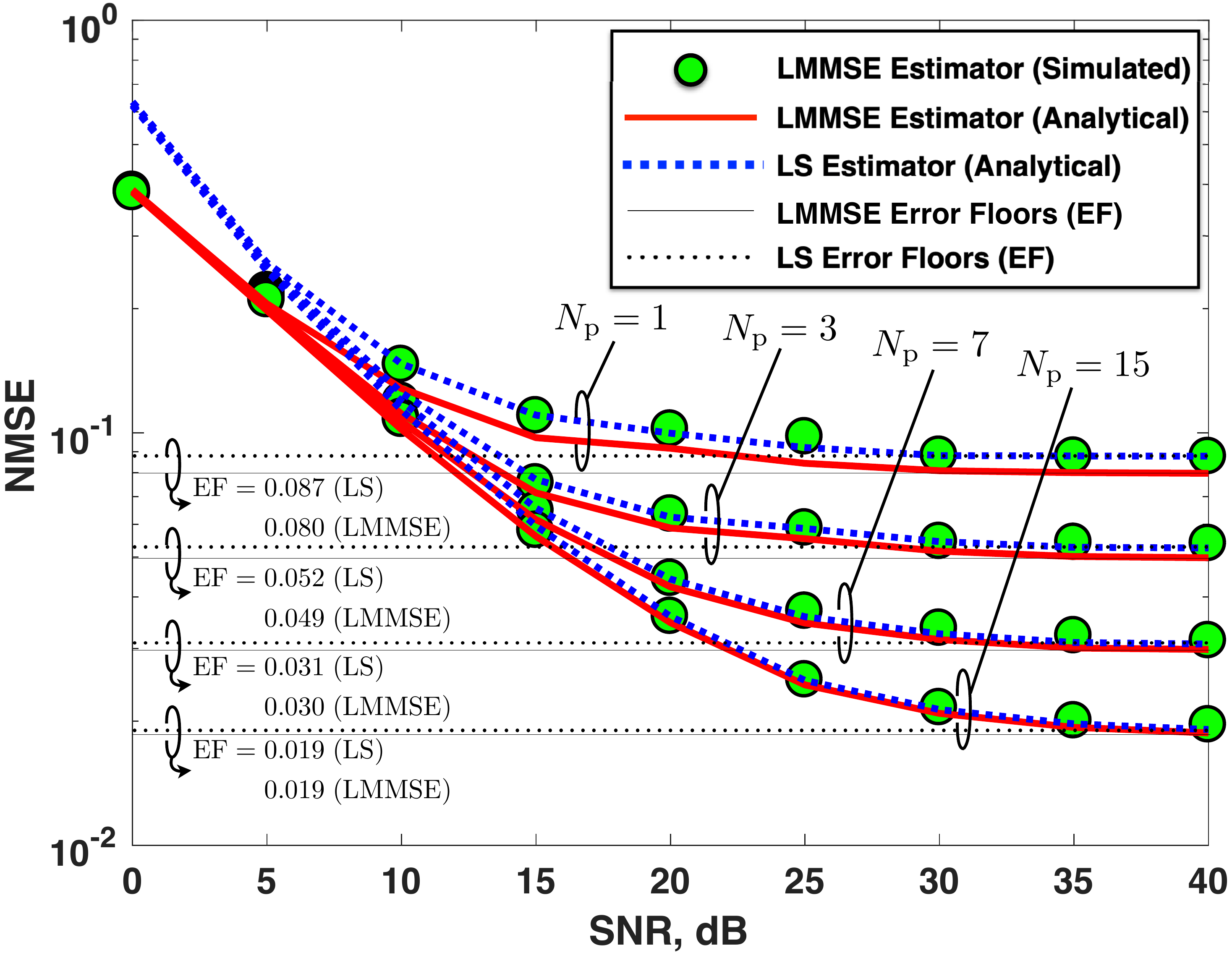}
\label{fig:NMSE_IFC_LMMSE_5000}
}
\caption{NMSE of LMMSE ICI-free-channel estimator as a function of $\mathsf{SNR}$ for $\beta \in \lbrace 500, 5000\rbrace$ and $\mathcal{P}_{\rm d} = \lbrace 1, 3, 7, 15 \rbrace$. For comparison, the NMSE curves of LS estimator corresponding to the elements in $\mathcal{P}_{\rm d}$~(dotted blue line) are included. Also shown are the error floors corresponding to the elements in $\mathcal{P}_{\rm d}$. The error floors are obtained by (\ref{eq:NMSE_IFC_LS_HighSNR}) and (\ref{eq:NMSE_IFC_LMMSE_HighSNR}), respectively.}
\label{fig:NMSE_IFC_LMMSE}
\vspace*{-0.5 cm}
\end{figure*}
\begin{remark}\normalfont
\label{remark:NMSE_IFC}
{\textbf{(NMSE floor of ICI-free-channel estimation):}}
To present the NMSE floor of ICI-free-channel estimation, which bounds the achievable NMSE for linear estimators, let us look at the NMSEs in the high SNR regime. The NMSEs of LS/LMMSE ICI-free-channel estimators are lower-bounded by, respectively, \ie, ${\mathsf{NMSE}}_{{\rm c}, q} \geq {\mathsf{NMSE}}_{{\rm c}, q}^{\mathsf{lb}}, \; q \in \lbrace \mathsf{ls}, \mathsf{lmmse} \rbrace$,
\begin{align}
\label{eq:NMSE_IFC_LS_HighSNR}
{\mathsf{NMSE}}_{{\rm c}, \mathsf{ls}}^{\mathsf{lb}}
=
\lim_{{\mathsf{SNR}}\to\infty} 
{\mathsf{NMSE}}_{{\rm c}, \mathsf{ls}}
= 
\frac
{\sigma_{\varepsilon}^{2}}
{1 - \sigma_{\varepsilon}^{2}}, \text{and}
\end{align}
\begin{align}
\label{eq:NMSE_IFC_LMMSE_HighSNR}
{\mathsf{NMSE}}_{{\rm c}, \mathsf{lmmse}}^{\mathsf{lb}}
=
\lim_{{\mathsf{SNR}}\to\infty} 
{\mathsf{NMSE}}_{{\rm c}, \mathsf{lmmse}}
= 
\sigma_{\varepsilon}^{2}.
\end{align}
In the case where the effective-error variance is small enough ($\sigma_{\varepsilon}^{2} \ll 1$), the lower bound of LS ICI-free-channel estimation (\ref{eq:NMSE_IFC_LS_HighSNR}) can be approximated as ${\mathsf{NMSE}}_{{\rm c}, \mathsf{ls}}^{\mathsf{lb}} \approx \sigma_{\varepsilon}^{2}$, resulting in the same NMSE floor as the LMMSE estimator. 
\end{remark}
Comparisons of the NMSE expressions for LS/LMMSE ICI-free-channel estimation (\ref{eq:NMSE_IFC_LS}) and (\ref{eq:NMSE_IFC_LMMSE}) with their simulation results are shown in Fig.~\ref{fig:NMSE_IFC_LS}-\ref{fig:NMSE_IFC_LMMSE}. In the numerical evaluation, we used the LMMSE PN-affected-channel estimator.  All figures have good agreements. The NMSE gap between LS and LMMSE estimators decreases as the SNR increases. In Fig.~\ref{fig:NMSE_IFC_LMMSE_500}, it is observed that he  LS and LMMSE NMSE floors are equal (rounded to fourth decimal place), as analyzed in Remark~\ref{remark:NMSE_IFC}. However, in the more severe \ac{PN} case ($\beta = 5000$), higher effective-error variance arises, translating into a gap between LS and LMMSE NMSE floors shown in Fig.~\ref{fig:NMSE_IFC_LMMSE_5000}.  
\begin{table*}[t]
\caption{Computational complexity comparison.}
\renewcommand{\arraystretch}{1.4} 
\centering 
\setlength{\abovecaptionskip}{0pt}
\scalebox{0.9}{
\begin{threeparttable}
\begin{tabular} { >{\centering\arraybackslash}p{1.7cm} |   >{\centering\arraybackslash}p{1.1cm} ||   >{\centering\arraybackslash}p{3.7cm} | >{\centering\arraybackslash}p{3.7cm} ||
    >{\centering\arraybackslash}p{2.6cm}}
    \multicolumn{2}{c||}{\multirow{2}{*} {}} & \multicolumn{2}{c||} {\bf{Estimation}} & \bf{Compensation}\\
    \cline{3-5} 
\multicolumn{2}{c||}{{}} & \bf{Phase Noise} & \bf{Wireless Channel} & \bf{Phase Noise}\\
\hline\hline
\multirow{2}{*}{\centering{\bf{Proposed}}\tnote{\P}} & \bf{LS} & $\mathcal{O}(0)$ & $\mathcal{O}(N_{\rm c})$ &  \multirow{2}{*} {$\mathcal{O}(NN_{\rm p})$}\\
\cline{2-4}
{} & \bf{LMMSE} & $\mathcal{O} \big( N_{\rm p}^{2} (N_{\rm p} + 1) \big)$ & $\mathcal{O}(N_{\rm c}^{2})$ & {}\\ 
\hline
\multicolumn{2}{c||}{\centering{\cite{rab10}\tnote{\dag}}}
& $\mathcal{O}(N_{\rm p}^{3}+NL)$ & $\mathcal{O}\big(N(N_{\rm p}+L) \big)$ & $\mathcal{O}(NN_{\rm p})$\\  
\hline
\multicolumn{2}{c||}{\centering{\cite{zou07}}}
 & \multicolumn{2}{c||}{$\mathcal{O}(N_{\rm it}N^{3})$\tnote{\ddag}} & $\mathcal{O}(N)$\\
\hline
\multicolumn{2}{c||}{\centering{\cite{wang17}}}
 & $\mathcal{O}(N^{2}{\rm{log}_{2}}{N} + N_{\rm it}N)$\tnote{\ddag} & $\mathcal{O}(N)$ & $\mathcal{O}(N)$\\
\end{tabular}
\begin{tablenotes}
\footnotesize
\item[\P] In the proposed method, the PN-affected and ICI-free channels are applied instead of \ac{PN} and wireless channel, respectively.
\item[\dag] $L$ denotes the number of effective channel taps in time domain. 
\item[\ddag] $N_{\rm it}$ denotes the number of iterations required.
\end{tablenotes}
\end{threeparttable}
}
\label{table:complex}
\end{table*}

\section{Pilot Overhead and Complexity Analysis}
\label{sec:OH_Complex_Analysis}
Our proposed algorithm translates into a practical \ac{PN} estimation/compensation for \ac{mmWave} \ac{OFDM} systems. To derive this, we address the pilot-overhead and the computational complexity of our proposed method.
\subsection{Pilot Overhead Analysis}
\label{subsec:OH}
Recall that the resource allocation in each $\mathcal{S}$~($| \mathcal{S} | = NN_{\rm ct}$) is identical where $\mathcal{S}$ is a set of coherence blocks across $N$ subcarriers as illustrated in~\fig\ref{fig:tx_structure}. The pilot overhead is defined as $\rho_{\rm oh} \triangleq N_{\rm tp} / NN_{\rm ct}$ where $N_{\rm tp}$ is the total number of pilots. The following theorem provides the minimum pilot-overhead of the proposed algorithm.
\begin{theorem}
\label{thm:overhead}
Supposing a set of system parameters $\lbrace N, N_{\rm ct}, N_{\rm c}, N_{\rm p} \rbrace$, the minimum pilot-overhead for the PN-affected- and ICI-free-channel estimation  is 
\begin{align}
\label{eq:min_OH_ratio}
\rho_{\rm oh}
=
\frac
{
N_{\rm ct}(2N_{\rm p}-1)
+
(N_{\rm c}-1)
}
{NN_{\rm ct}}.	
\end{align}
\end{theorem}
\begin{proof}
Consider the allocation of PN- and CH-dedicated pilots in the $\mathcal{S}$. It is shown in Theorem~\ref{thm:optimal_pilot_design} that $(2N_{\rm p}-1)$ PN-dedicated pilots are required to estimate $N_{\rm p}$ PN-affected-channel coefficients. The PN-affected-channel estimation for each OFDM symbol leads the allocation of $N_{\rm ct}(2N_{\rm p}-1)$ PN-dedicated pilots in the $\mathcal{S}$. Recall that $(N_{\rm c}-1)$ CH-dedicated pilots are additionally needed for ICI-free-channel estimation over $N_{\rm ct}$ OFDM symbols. Hence (\ref{eq:min_OH_ratio}) can be clearly derived.  
\end{proof}
We provide an example below to help the understanding of how much the pilot overhead for our proposed algorithm is, as compared to the conventional cellular systems.

{\it{Example 4 (Comparison with the Cell-Specific Reference Symbol Overhead of Conventional Cellular Systems):}} In this example, let us consider a set of parameters{\footnote{A resource block in LTE systems consists of 12 consecutive subcarriers and 7 OFDM symbols. 100 resource blocks are used to support \SI{20}{\MHz} bandwidth. Thus, the number of occupied subcarriers is 1200~\cite{sesia2009lte}. In this example, we use the number of occupied subcarriers for $N$.}} in Long-Term Evolution~(LTE) systems supporting \SI{20}{\MHz} channel bandwidth: $N = 1200$, $N_{\rm ct} = 7$, $N_{\rm c} = 100$. We assume that one Cell-Specific Reference Symbol~(CRS) is allocated for a resource block, \ie, $N_{\rm tp} = N_{\rm c}$. Based on this parameter set, therefore, the CRS overhead $\rho_{\rm oh, crs} \triangleq N_{\rm c}/(NN_{\rm ct})$ is $\SI{1.19}{\%}$, which does not include the overhead for \ac{PN} estimation. Consider the set of the number of dominant \ac{PN} components $\mathcal{P}_{\rm d} = \lbrace 1, 3, 7,  15 \rbrace$. The corresponding minimum pilot-overhead ratios from (\ref{eq:min_OH_ratio}) are \SI{1.26}{\%}, \SI{1.60}{\%}, \SI{2.26}{\%}, and \SI{3.60}{\%}, respectively. These are quite reasonable values for the practical use of our algorithm.

\subsection{Computational Complexity Analysis}
\label{subsec:complex}
In this subsection, we investigate the computational complexity of the PN-affected-/ICI-free-channel estimation and the ICI suppression (\ac{PN} compensation). Since the LS PN-affected-channel estimator (\ref{eq:Q_LS}) is an identity matrix, no computation is required for obtaining $\hat{\mathbf{f}}_{{\bar{\rm p}},\mathsf{ls}}$. The LMMSE PN-affected-channel estimator (\ref{eq:Q_LMMSE}) and the matrix-vector multiplication (\ref{eq:f_p_hat_LS}) have a complexity of respectively $\mathcal{O} \big( N_{\rm p}^{3} \big)$ and $\mathcal{O} \big( N_{\rm p}^{2} \big)$, leading to a total complexity in the order of $\mathcal{O} \big( N_{\rm p}^{2} (N_{\rm p} + 1) \big)$. According to (\ref{eq:V_ls})\hspace{0.2mm}--\hspace{0.2mm}(\ref{eq:h_If_est}), the complexity order of the LS/LMMSE ICI-free-channel estimation is $\mathcal{O}(N_{\rm c})$ and $\mathcal{O}(N_{\rm c}^{2})$, respectively. As described in Section~\ref{subsec:deconv}, the \ac{PN} compensation in the proposed method is performed in the frequency domain. Recall that the \ac{PN} effect is a circular convolution process in the frequency domain. Hence the \ac{PN} compensation process is the deconvolution{\footnote{The deconvolution of two length-$N$ sequences is equivalent to their polynomial division where the polynomial coefficients correspond the coefficients in each sequence, and its operation has a complexity $\mathcal{O}(N^{2})$.}} of the received signal and the \ac{PN} estimate in frequency. It results in a complexity of $\mathcal{O}(N^{2})$. Since the length-$N$ PN-affected-channel estimate $\hat{\mathbf{f}}_{\rm p}$ includes only $N_{\rm p}$ nonzero values, the deconvolution~(\ref{eq:y_If}) has a complexity $\mathcal{O}({NN_{\rm p}})$. 

The complexity comparison with existing work on low-complexity \ac{PN} estimation and compensation is shown in \tbl~\ref{table:complex}. From the relation $N \gg N_{\rm p}, N_{\rm c}$, the proposed method has lower complexity for both \ac{PN} and channel estimation than the existing solutions. Let us consider a total complexity, including joint \ac{PN}/channel estimation and \ac{PN} compensation, with mmWave system parameters{\footnote{$N$ is a \ac{3GPP} \ac{NR} parameter for mmWave communications~\cite{38_211} and $L$ is selected based on the measurement campaign result that the mean number of effective multipath components at \SI{28}{\GHz} and \SI{73}{\GHz} was 3.3 -- 7.2~\cite{samimi16}.}. For example, if $N=4096$, $N_{\rm p} = 7$, $N_{\rm c} = 100$, $L = 5$, and $N_{\rm it} = 1$, the proposed method with the LMMSE estimation obtains a reduction of $2.53 \times$, $(1.76\hspace{-0.7mm}\times\hspace{-0.7mm}10^{6}) \times$, and $(5.15\hspace{-0.7mm}\times\hspace{-0.7mm}10^{3}) \times$, respectively, in the total complexity, as compared to~\cite{rab10, zou07, wang17}. Furthermore, all of these existing solutions require a full-pilot \ac{OFDM} symbol to perform joint \ac{PN} and channel estimation, which leads to significant pilot overhead to tackle the problem of fast-varying \ac{PN} estimation.

\section{Trade-Off Analysis}
\label{sec:tradeoff}
This section uses \ac{BER} and throughput to study the trade-off between performance and pilot-overhead. For the numerical evaluation, the following parameters are used: $N = 4096$, $\Delta f = \SI{60}{\kHz}$, and $N_{\rm c} = 275$, which corresponds to one \ac{3GPP} \ac{NR} signaling to support communication at \ac{mmWave} frequency~\cite{38_211}. Also, we consider the set of dominant phase-noise components $\mathcal{P}_{\rm d} \in \lbrace 1,3,7,15\rbrace$ and two kinds of 3-$\rm{dB}$ linewidth $\beta \in \lbrace 500, 5000 \rbrace$.

\subsection{Bit-Error Rate Performance}
\fig\ref{fig:ber} shows the \ac{BER} performance for an \ac{OFDM} system transmitting uncoded 16-\ac{QAM}. The \ac{BER} curves of $N_{\rm p}$-perfect \ac{PN} compensation serve as a benchmark, where it is assumed that $N_{\rm p}  \in \lbrace 3, 7\rbrace$ dominant \ac{PN} components are perfectly known, and able to be used for the compensation. The performance curve without \ac{PN} is used as another benchmark for comparison. As illustrated in \fig\ref{fig:ber_500}, the proposed method has quite good \ac{BER} performance by using the estimation of even only three significant \ac{PN} components when $\beta = 500$. In case of $N_{\rm p} \in \lbrace 3, 7 \rbrace$, it is shown that there is around $3\hspace{0.1em}{\rm dB}$ difference between perfect $N_{\rm p}$-perfect phase-noise compensation and proposed method at a \ac{BER} of $10^{-3}$. In comparison with no \ac{PN} case, there is around $3.5\hspace{0.1em}{\rm dB}$ and $5\hspace{0.1em}{\rm dB}$, respectively, for $N_{\rm p} \in \lbrace 3, 7 \rbrace$, at the same \ac{BER} level. Also, it is observed that, when $N_{\rm p} > 3$, the performance improvements by the proposed method is relatively small. It means that most of \ac{PN} energy is focused in three dominant \ac{PN} components in the $\beta = 500$ case. Whereas the \ac{BER} performance shown in \fig\ref{fig:ber_5000}~($\beta = 5000$ ) is largely improved, as more number of dominant \ac{PN} components is considered. For example, as $N_{\rm p} \in \lbrace 1, 3, 7, 15 \rbrace$ increases, their BERs at the $30\hspace{0.1em}{\rm dB}$ \ac{SNR} level, are 0.089, 0.048, 0.02, and 0.007, respectively. We can tell that there is more room to improve the \ac{BER} performance by the use of pilot-overhead, as compared to the $\beta = 500$.

\begin{figure*}[t!]
\centering
\subfigure[$\beta = 500$]
{
\includegraphics[width=3.42in]{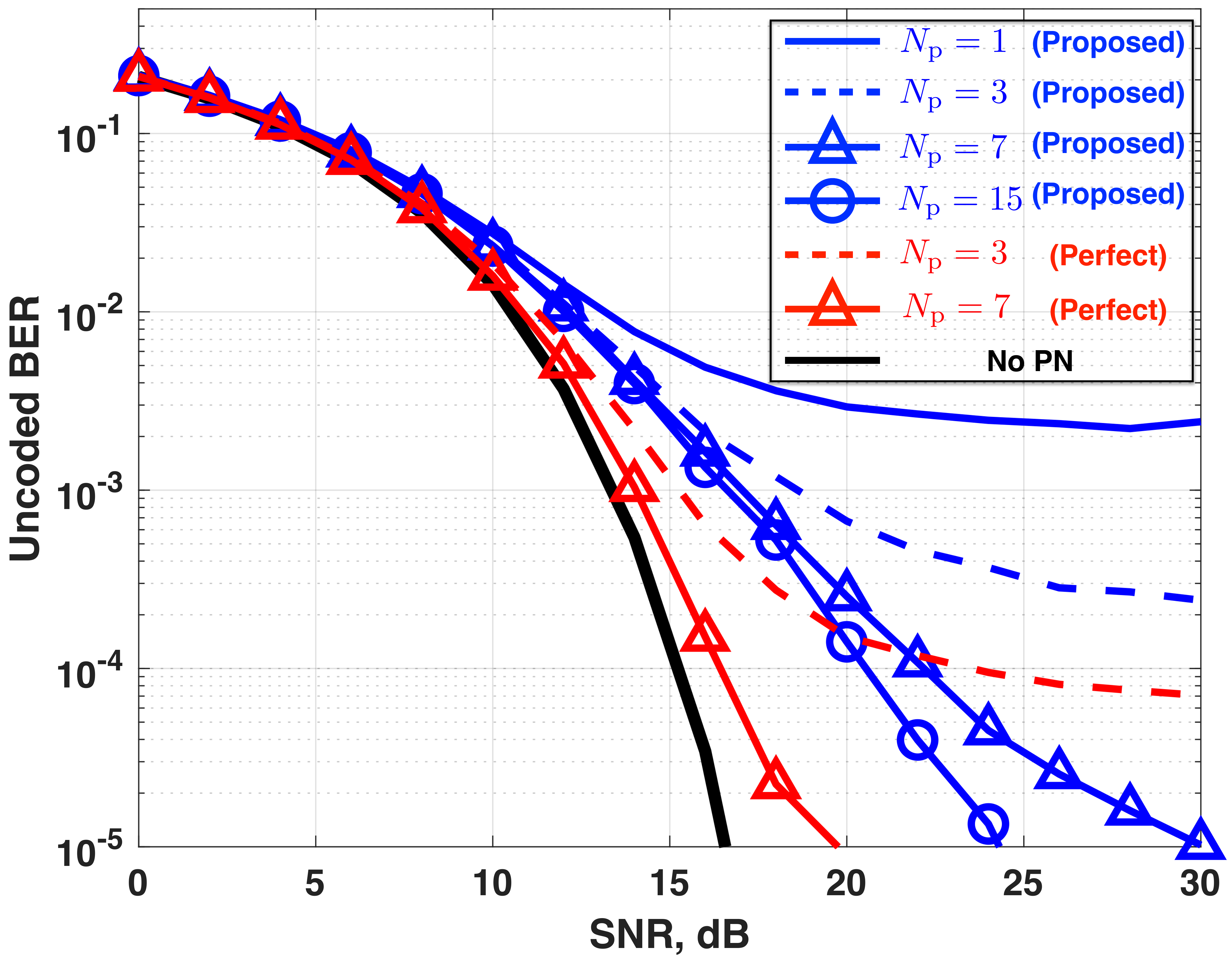}
\label{fig:ber_500}
}
\subfigure[$\beta = 5000$]
{
\includegraphics[width=3.42in]{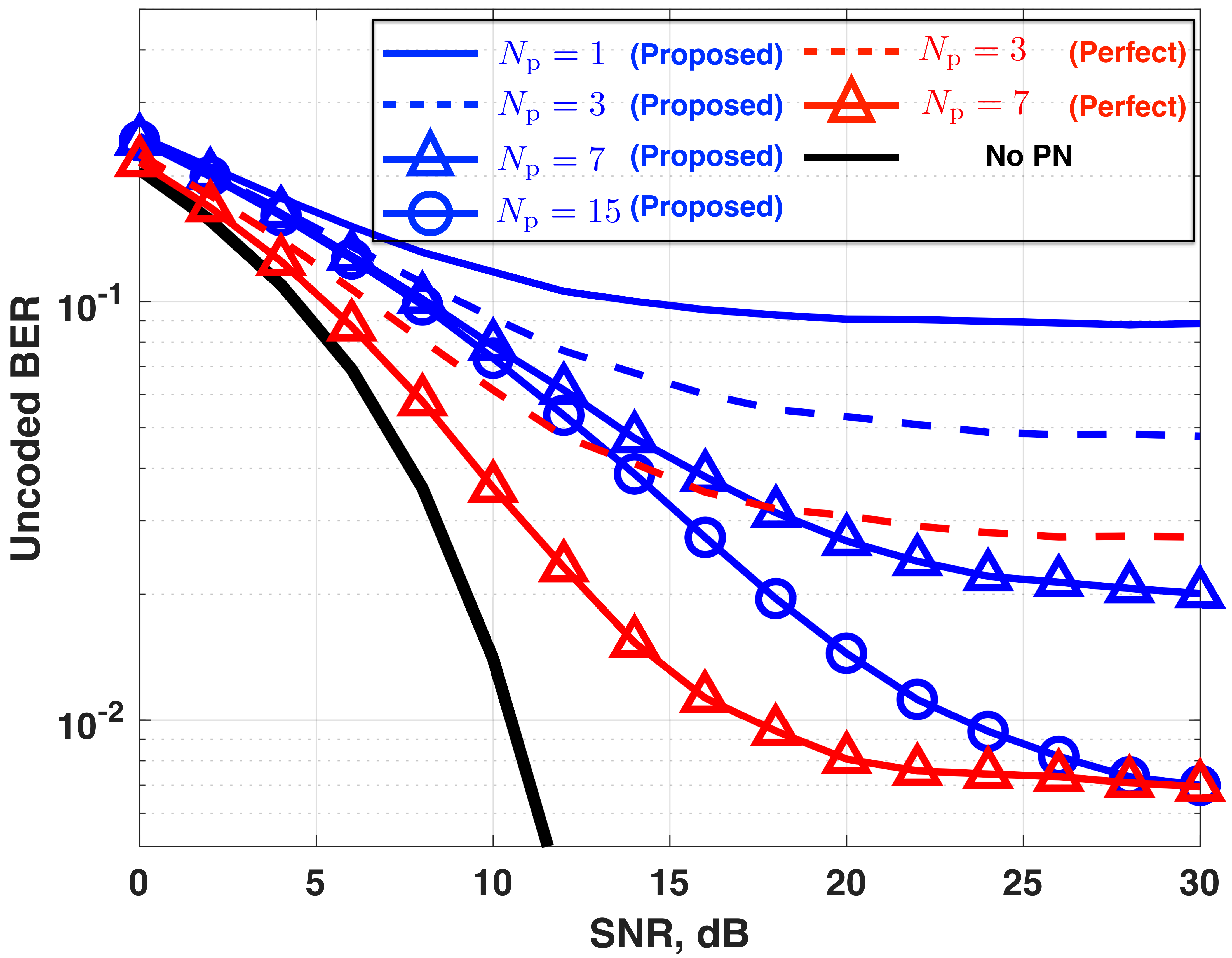}
\label{fig:ber_5000}
}
\caption{BER performance of the proposed solutions~(blue), $N_{\rm p}$-perfect \ac{PN} compensation~(red), and no \ac{PN} case~(black) with $\beta \in \lbrace 500, 5000 \rbrace$.}
\label{fig:ber}
\vspace*{-0.25 cm}
\end{figure*}

\subsection{Throughput versus Pilot-Overhead Trade-Off}
To study the trade-off between performance and pilot-overhead, we define the throughput based on \ac{3GPP} terminologies, as follows:
\begin{align}
\label{eq:thp}
{\mathsf{THP}}
=
(
1 - \rho_{\rm oh}
)
\times
\frac{
N_{\rm c} 
\times
N_{\rm re}
\times
N_{\mathsf{ofdm}}
}
{T_{\mathsf{slot}}}
\times
M_{\mathsf{qam}}
\times
(1 - \mathsf{BER})
,
\end{align}
where $N_{\rm c}$ and $N_{\rm re}$ are the number of resource blocks, resource elements in a resource block, respectively; $N_{\mathsf{ofdm}}$ is the number of \ac{OFDM} symbols per slot, $T_{\mathsf{slot}}$ the slot duration, $M_{\mathsf{qam}}$ a modulation order per resource element, $\mathsf{BER}$ the average \ac{BER}. For the numerical evaluation with (\ref{eq:thp}), the following parameters are assumed: $N_{\rm re} = 12$, $N_{\mathsf{ofdm}} = 14$, $T_{\mathsf{slot}} = \SI{0.25}{\ms}$, $M_{\mathsf{qam}} = 4$, which also corresponds to one \ac{3GPP} \ac{NR} signaling resource block to support communication at \ac{mmWave} frequency~\cite{38_211}. \fig\ref{fig:thp} shows the throughput performance as a function of SNR. From (\ref{eq:min_OH_ratio}), the pilot-overhead\footnote{As with Example 4, the number of occupied subcarriers is 3300 in the corresponding \ac{3GPP} \ac{NR} signaling~\cite{38_211}. Thus, 3300 is applied for $N$ in (\ref{eq:min_OH_ratio}), instead of 4096.} $\rho_{\rm oh}$, according to $N_{\rm p} \in \mathcal{P}_{\rm d}$, is 1.22\hspace{0.1em}\%, 1.34\hspace{0.1em}\%, 1.58\hspace{0.1em}\%, and 2.06\hspace{0.1em}\%, respectively. When $\beta = 5000$, the higher-order \ac{PN} approximation and its estimation lead to the better throughput performance although the more pilot-overhead is required. On the other hand, when $\beta = 500$, the estimation of three dominant \ac{PN} components, \ie, $N_{\rm p} = 3$, results in better throughput performance than the others, in the \ac{SNR} range more than $6\hspace{0.1em}{\rm dB}$. At the high \acp{SNR}, the throughput with $N_{\rm p} = 3 $ is around $5.2\hspace{0.1em}{\rm{Mbits/s}}$ higher than the one with $N_{\rm p} = 15$, while the $N_{\rm p} = 15$ case has around $24.4\hspace{0.1em}{\rm{Mbits/s}}$ higher throughput, as compared to the $N_{\rm p} = 3$. From these results, it is found that higher-order \ac{PN} estimation does not guarantee better throughput performance due to the increase of the pilot-overhead.
\begin{figure}[t!]
\centering
\includegraphics[width=3.42in]{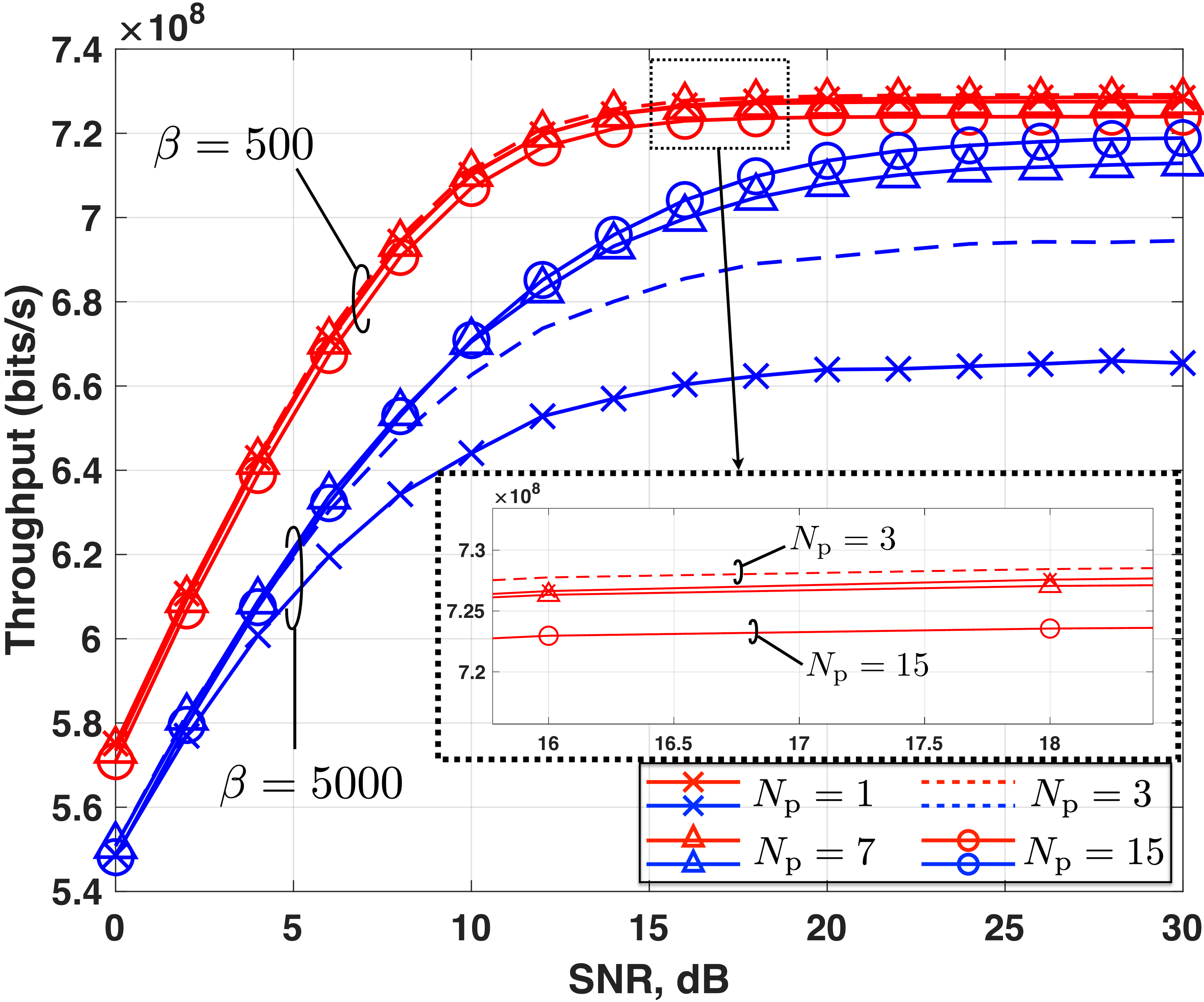}
\caption{Trade-off between throughput and pilot-overhead with $\mathcal{P}_{\rm d} = \lbrace 1,3,7,15\rbrace$ and $\beta \in \lbrace 500, 5000 \rbrace$.}
\label{fig:thp}
\vspace*{-0.25 cm}
\end{figure}

Although the pilot-overhead effort leads to the throughput improvement when $\beta = 5000$, the high-order approximation of \ac{PN} spectrum, \eg, $N_{\rm p} = 7$ or $15$, may not always be required. The computational complexity of \ac{PN}-affected-channel estimation is a function of $N_{\rm p}$, \ie, $\mathcal{O} \big( N_{\rm p}^{2} (N_{\rm p} + 1) \big)$ with \ac{LMMSE} estimator. In case where the throughput difference is marginal according to $N_{\rm p}$, the lowest $N_{\rm p}$ could be selected to reduce the complexity, if not for throughput-sensitive applications. For example, when $\beta = 5000$, the \ac{PN} spectrum approximation with $N_{\rm p} = 1$ could be considered for a \ac{SNR} of less than about $2\hspace{0.1em}{\rm dB}$, and $N_{\rm p} = 3$ for a \ac{SNR} of $2$\hspace{0.1em}-\hspace{0.1em}$5\hspace{0.1em}{\rm dB}$.

\section{Conclusion}
\label{sec:conclusion}
Practically suppressing the effect of \ac{PN} is a critical aspect of mmWave communication systems to realize its potential benefits. This paper has outlined a novel framework for \ac{PN} compensation on OFDM systems, which uses LS/LMMSE estimators and pilot-assisted transmission. Our main conclusion is that the large coherence bandwidth in mmWave bands and an approximation of the \ac{PN} spectrum enable low-complexity \ac{PN} compensation with a reasonable pilot-overhead, which leads to a very efficient solution for the severe \ac{PN} problem. Further, we have derived analytically tractable expressions for the NMSE performance of our proposed framework, and studied the trade-off between performance and pilot-overhead. These expressions and trade-off analysis offer an insight into an appropriate approximation of the \ac{PN} spectrum, according to the \ac{SNR} and \ac{PN} environments.  

%
\addtocounter{equation}{4}
\begin{figure*}[b!]
\hrulefill\\
\begin{align}
\label{append:PHI_ICI_gamma}
{\mathbf{\Phi}_{\rm ici,{\gamma}}^{\rm p}} 
=
\begin{bmatrix}
P_{N-N_{\rm b}} & P_{N-N_{\rm b}-1} & P_{N-N_{\rm b}-2} & \cdots & P_{N_{\rm a}+1} & P_{N_{\rm a}}\\
P_{N-N_{\rm b}+1} & P_{N-N_{\rm b}}  & P_{N-N_{\rm b}-1} & \cdots & P_{N_{\rm a}+2} & P_{N_{\rm a}+1}\\
\vdots & \vdots & \vdots & \cdots & \vdots & \vdots\\
P_{N-N_{\rm b}+(N_{\rm p}-1)} & P_{N-N_{\rm b}+(N_{\rm p}-2)} & P_{N-N_{\rm b}+(N_{\rm p}-3)} & \cdots & P_{N_{\rm a}+N_{\rm p}} & P_{N_{\rm a}+(N_{\rm p}-1)}
\end{bmatrix}.
\end{align}
\end{figure*}
\addtocounter{equation}{-5}

\begin{appendices}
\section{Proof of Theorem~\ref{thm:optimal_pilot_design}}
\label{append_sec:optimal_PNDP}
From~(\ref{eq:y_pnac_np_3_2}), we have the generalized form of $\mathbf{y}_{\rm f}^{\rm p}$ with respect to $\gamma$-order approximation as follows:
\begin{align}
\label{eq:rx_signal_gamma_approx}
	\mathbf{y}_{\rm f}^{\rm p}
	= 
	{\mathbf{X}_{\rm f, \gamma}^{\rm p}}{\mathbf{f}_{\bar{\rm {p}}}}
	+ {\mathbf{w}}_{\gamma}^{\rm p},
\end{align}
where ${\mathbf{w}}^{\rm p}_{\gamma} = [W_{\gamma, 0}^{\rm p}, W_{\gamma, 1}^{\rm p}, \cdots, W_{\gamma, N_{\rm p}-1}^{\rm p}] \in \mathbb{C}^{N_{\rm p} \times 1}$ is the ICI by the $\gamma$-order-approximation error plus AWGN in $\mathbf{y}_{\rm f}^{\rm p}$. The element set in the $\mathbf{X}_{{\rm f}, \gamma}^{\rm p}$ is $\lbrace X_{k}^{\rm p} \rbrace_{k = 0}^{2N_{\rm p}-2}$, which means that $(2N_{\rm p}-1)$ PN-dedicated pilots are required to estimate $N_{\rm p}$ PN-affected-channel components. Regarding each observation in $\mathbf{y}_{\rm f}^{\rm p}$, the $N_{\rm p}$ PN-dedicated pilots are multiplied with the $\mathbf{f}_{\bar {\rm p}}$. The remaining $(N_{\rm p}-1)$-pilot, however, combines with higher-order \ac{PN} components than $\gamma$, resulting in being involved in $\mathbf{w}_{\gamma}^{\rm p}$. To meet two conditions for PN-dedicated-pilot pattern, which are the ICI minimization and ${\rm rank}(\mathbf{X}_{\rm f, \gamma}^{\rm p}) = N_{\rm p}$, we employ the fact that the diagonal term $X_{\rm 2\gamma}^{\rm p}$ in $\mathbf{X}_{\rm f, \gamma}^{\rm p}$ does not belong to $\mathbf{w}_{\gamma}^{\rm p}$ and can be used for making the full rank of $\mathbf{X}_{\rm f, \gamma}^{\rm p}$. Hence, a non-zero pilot symbol is allocated for $X_{\rm 2\gamma}^{\rm p}$ and zero-pilot for the remainder to minimize the ICI, leading to the optimal PN-dedicated pilot matrix~(\ref{eq:optimal_pilot_1}).

\section{LMMSE Estimator for PN-Affected Channel}
\label{append_sec:LMMSE_PNAC}
The LMMSE PN-affected-channel estimator is defined as
\begin{align}
\label{append:Q_LMMSE}
{\mathbf{Q}}	_{\mathsf{lmmse}} 
=
{\mathbf{R}}	_{\mathbf{f}\mathbf{y}}
{{\mathbf{R}}_{\mathbf{y}\mathbf{y}}}^{-1},	
\end{align}
where ${\mathbf{R}}_{\mathbf{f}\mathbf{y}}=\mathbb{E}\lbrace{\mathbf{f}_{\bar{\rm p}}}({\mathbf{y}_{\rm f}^{\rm p}})^{\rm H}\rbrace$ is the cross-covariance matrix between ${\mathbf{f}_{\bar{\rm p}}}$ and ${\mathbf{y}_{\rm f}^{\rm p}}$, ${\mathbf{R}}_{\mathbf{y}\mathbf{y}}=\mathbb{E}\lbrace{\mathbf{y}_{\rm f}^{\rm p}}({\mathbf{y}_{\rm f}^{\rm p}})^{\rm H}\rbrace$ the autocorrelation matrix of ${\mathbf{y}_{\rm f}^{\rm p}}$. Substituting~(\ref{eq:optimal_pilot_1}) in Theorem~\ref{thm:optimal_pilot_design} into~(\ref{eq:rx_signal_gamma_approx}) , we have 
\begin{align}
\begin{split}
\label{append:R_fy_PNAC}
{\mathbf{R}}_{\mathbf{f}\mathbf{y}}
&=
\mathbb{E}
\lbrace
{\mathbf{f}_{\bar{\rm p}}}
(
{\mathbf{X}_{\rm f, \gamma}^{\rm p}}
{\mathbf{f}_{\bar{\rm p}}}
+
{\mathbf{w}_{\gamma}^{\rm p}}
)^{\rm H}
\rbrace\\
&=
\mathbb{E}
\lbrace
{\mathbf{f}_{\bar{\rm p}}}
{\mathbf{f}_{\bar{\rm p}}}^{\rm H} 
\rbrace
+
\underbrace{
\mathbb{E}
\lbrace
{\mathbf{f}_{\bar{\rm p}}}
(
{\mathbf{E}_{\gamma}^{\rm p}}
{\mathbf{x}_{\rm f}} 
+ 
{\mathbf{z}_{\rm f}}
)^{\rm H}
\rbrace
}_{\mathbf{0}_{N_{\rm p} \times N_{\rm p}}} 
\\
&\stackrel{\text{(a)}}{=}
\underbrace{
\mathbb{E}
\lbrace 
\| \alpha \|_{2}^{2} 
\rbrace}_{1} 
\underbrace{
\mathbb{E} 
\lbrace
\bar{\mathbf{p}}_{{\rm f},\gamma}
(
\bar{\mathbf{p}}_{{\rm f},\gamma}
)^{\rm H}
\rbrace
}_{\mathbf{R}_{\mathbf{p}\mathbf{p}}^{\gamma}}
=
\mathbf{R}_{\mathbf{p}\mathbf{p}}^{\gamma}, 
\end{split}
\end{align}
\begin{align}
\begin{split}
\label{append:R_yy_PNAC}
{\mathbf{R}}_{\mathbf{y}\mathbf{y}}
&=
\mathbb{E}
\lbrace
(
{\mathbf{f}_{\bar{\rm p}}}
+
{\mathbf{w}_{\gamma}^{\rm p}}
)
(
{\mathbf{f}_{\bar{\rm p}}}
+
{\mathbf{w}_{\gamma}^{\rm p}})^{\rm H}
\rbrace\\
&=
\mathbf{R}_{\mathbf{p}\mathbf{p}}^{\gamma}
+
\mathbb{E}
\lbrace
{\mathbf{E}_{\gamma}^{\rm p}}
{\mathbf{x}_{\rm f}}
{{\mathbf{x}_{\rm f}}^{\rm H}}
(
{\mathbf{E}_{\gamma}^{\rm p}}
)^{\rm H}
\rbrace
+ 
\underbrace{
\mathbb{E}
\lbrace
{\mathbf{z}_{\rm f}}
{\mathbf{z}_{\rm f}^{\rm H}}
\rbrace
}_{\sigma_{z}^{2}{\mathbf{I}_{N_{\rm p}}}}\\
&=
\mathbf{R}_{\mathbf{p}\mathbf{p}}^{\gamma}
+
\underbrace{
\mathbb{E}
\lbrace
\tilde{\mathbf{\Phi}}_{{\rm f}, \gamma}^{\rm p}
(
{\tilde{\mathbf{\Phi}}_{{\rm f}, \gamma}^{\rm p}
)^{\rm H}}
\rbrace 
}_{\mathbf{R}_{\mathbf{ici}}}
+
\sigma_{z}^{2}{\mathbf{I}_{N_{\rm p}}}\\
&=
\mathbf{R}_{\mathbf{p}\mathbf{p}}^{\gamma} 
+ 
\mathbf{R}_{\mathbf{ici}} 
+ 
(1/{\mathsf{SNR}})
{\mathbf{I}_{N_{\rm p}}}.
\end{split}
\end{align}
where ${\mathbf{E}}_{\gamma}^{\rm p} \in \mathbb{C}^{N_{\rm p} \times N}$ is the $\gamma$-order-approximation-error matrix in $\mathbf{y}_{\rm f}^{\rm p}$, $\tilde{\mathbf{\Phi}}_{{\rm f}, \gamma}^{\rm p} = \lbrack \mathbf{\Phi}_{{\rm ici},\gamma}^{\rm p} \;\; \mathbf{0}_{N_{\rm p} \times (2N_{\rm p}-1)} \rbrack \in \mathbb{C}^{N_{\rm p} \times N}$. The $\mathbf{\Phi}_{{\rm ici}, \gamma}^{\rm p} \in \mathbb{C}^{N_{\rm p} \times (N-2N{\rm p}+1)}$ is given in (\ref{append:PHI_ICI_gamma}) at the bottom of this page, where $N_{\rm a} \triangleq \frac{N_{\rm p}+1}{2}$ and $N_{\rm b} \triangleq \frac{3N_{\rm p}-1}{2}$.
\addtocounter{equation}{8}
\begin{figure*}[b!]
\hrulefill\\
\begin{align}
\begin{split}
\label{append:R_hy_IFC}
\mathbf{R}_{\tilde{\mathbf{h}}\tilde{\mathbf{y}}}
&=
\mathbb{E}
\lbrace
\tilde{\mathbf{h}}_{\rm If}
(
\tilde{\mathbf{H}}_{\rm If}
\mathbf{x}_{\rm f}^{\rm c}
+
\tilde{\mathbf{\Upsilon}}_{\rm off}	
\mathbf{H}_{\rm If}
\mathbf{x}_{\rm f}
+
\tilde{\mathbf{z}}_{\rm f}
)^{\rm H}
\rbrace
\\
&=
\mathbb{E}
\lbrace
\tilde{\mathbf{h}}_{\rm If}
\tilde{\mathbf{h}}_{\rm If}^{\rm H}
\rbrace
(
\mathbf{X}_{\rm f}^{\rm c}
)^{\rm H}
+
\underbrace
{
\mathbb{E}
\lbrace
\tilde{\mathbf{h}}_{\rm If}
(
\tilde{\mathbf{\Upsilon}}_{\rm off}
\mathbf{H}_{\rm If}
\mathbf{x}_{\rm f}
)^{\rm H}
\rbrace
}_{\mathbf{0}_{N_{\rm c} \times N_{\rm c}}}
+
\underbrace
{
\mathbb{E}
\lbrace
\tilde{\mathbf{h}}_{\rm If}
\tilde{\mathbf{z}}_{\rm f}^{\rm H}
\rbrace
}_{\mathbf{0}_{N_{\rm c} \times N_{\rm c}}}
\\
&=
\bar{G}
(
1
-
\sigma_{\varepsilon}^2
)
(
\mathbf{X}_{\rm f}^{\rm c}
)^{\rm H}
\end{split}
\end{align}
\addtocounter{equation}{-9}
\addtocounter{equation}{9}
\begin{align}
\begin{split}
\label{append:R_yy_IFC}
\mathbf{R}_{\tilde{\mathbf{y}}\tilde{\mathbf{y}}}
&=
\mathbb{E}
\Big\{
\big(
\tilde{\mathbf{D}}_{N}
\lbrace
\mathbf{I}_{N}
-
\alpha
\mathbf{G}_{\rm p}
\mathbf{E}_{\rm t, eff}
\rbrace
\mathbf{D}_{N}^{\rm H}
\mathbf{H}_{\rm If}
\mathbf{x}_{\rm f}
+
\tilde{\mathbf{z}}_{\rm f}
\big)
\big(
\tilde{\mathbf{D}}_{N}
\lbrace
\mathbf{I}_{N}
-
\alpha
\mathbf{G}_{\rm p}
\mathbf{E}_{\rm t, eff}
\rbrace
\mathbf{D}_{N}^{\rm H}
\mathbf{H}_{\rm If}
\mathbf{x}_{\rm f}
+
\tilde{\mathbf{z}}_{\rm f}
\big)^{\rm H}
\Big\}
\\
&=
\frac{1}{\| \alpha \|_{2}^{2}}
\mathbb{E}
\big\{
\tilde{\mathbf{D}}_{N}
\lbrace
\mathbf{I}_{N}
-
\alpha
\mathbf{G}_{\rm p}
\mathbf{E}_{\rm t, eff}
\rbrace
\lbrace
\mathbf{I}_{N}
-
\alpha
\mathbf{G}_{\rm p}
\mathbf{E}_{\rm t, eff}
\rbrace^{\rm H}
\tilde{\mathbf{D}}_{N}^{\rm H}
\big\}
+
\mathbb{E}
\lbrace
\tilde{\mathbf{z}}_{\rm f}
\tilde{\mathbf{z}}_{\rm f}^{\rm H}
\rbrace\\
&=
\frac{1}{\| \alpha \|_{2}^{2}}
\mathbb{E}
\big\{
\tilde{\mathbf{D}}_{N}
\lbrace
\alpha
\mathbf{G}_{\rm p}
\mathbf{\Phi}_{\rm t}
\rbrace
\lbrace
\alpha
\mathbf{G}_{\rm p}
\mathbf{\Phi}_{\rm t}
\rbrace^{\rm H}
\tilde{\mathbf{D}}_{N}^{\rm H}
\big\}
+
\tilde{\mathbf{D}}_{N}
\mathbf{G}_{\rm p}
\mathbf{D}_{N}^{\rm H}
\underbrace{
\mathbb{E}
\lbrace
\mathbf{z}_{\rm f}
\mathbf{z}_{\rm f}^{\rm H}
\rbrace
}_{\sigma_{z}^{2} \mathbf{I}_{N}}
\mathbf{D}_{N}
\mathbf{G}_{\rm p}^{\rm H}
\tilde{\mathbf{D}}_{N}^{\rm H}
\\
&=
\tilde{\mathbf{D}}_{N}
\mathbf{G}_{\rm p}
\mathbf{G}_{\rm p}^{\rm H}
\tilde{\mathbf{D}}_{N}^{\rm H}
+
\sigma_{z}^{2}
\tilde{\mathbf{D}}_{N}
\mathbf{G}_{\rm p}
\mathbf{G}_{\rm p}^{\rm H}
\tilde{\mathbf{D}}_{N}^{\rm H}
=
(1+\sigma_{z}^{2})
\tilde{\mathbf{D}}_{N}
\mathbf{G}_{\rm p}
\mathbf{G}_{\rm p}
\tilde{\mathbf{D}}_{N}^{\rm H}
\\
&=
\bar{G}
\big\{
1
+
(1/\mathsf{SNR})
\big\}
\mathbf{I}_{N_{\rm c}}
\end{split}
\end{align}
\end{figure*}
\addtocounter{equation}{-10}

\section{Proof of Theorem~\ref{thm:deconv_out_vec}}
\label{append_sec:Thm_deconv_out_vec}
The equivalent time-domain representation of $\hat{\mathbf{f}}_{\rm p}$ and $\mathbf{e}_{\rm f, app}$ can be described as follows:
\addtocounter{equation}{1}
\begin{align}
\label{eq:PNAC_est_time_freq}
\hat{\mathbf{f}}_{\rm p}
=
\alpha
(
\mathbf{p}_{\rm f, \gamma}
+
\mathbf{e}_{\rm f, est}
)
\longleftrightarrow
\mathbf{g}_{\rm p}
=
\alpha
(
\mathbf{p}_{\rm t, \gamma}
+
\mathbf{e}_{\rm t, est}
),
\end{align}
\begin{align}
\label{eq:e_app_time_freq}
\mathbf{e}_{\rm f, app} 
= 
\mathbf{p}_{\rm f} 
-
\mathbf{p}_{\rm f, \gamma}
\longleftrightarrow
\mathbf{e}_{\rm t, app}
=
\mathbf{p}_{\rm t}
-
\mathbf{p}_{\rm t, \gamma}. 	
\end{align}
where $\mathbf{p}_{\rm t} \triangleq \sqrt{N} \mathbf{D}_{N}^{\rm H} \mathbf{p}_{\rm f} = \lbrack e^{j\phi_{0}}, e^{j\phi_{1}}, \cdots, e^{j\phi_{N-1}} \rbrack^{\rm T}$, $\mathbf{p}_{\rm t, \gamma} \triangleq \sqrt{N}\mathbf{D}_{N}^{\rm H}\mathbf{p}_{\rm f, \gamma} = \lbrack p_{0}, p_{1}, \cdots, p_{N-1} \rbrack^{\rm T}$, $\mathbf{e}_{\rm t, est} \triangleq \sqrt{N}{\mathbf{D}}_{N}^{\rm H}\mathbf{e}_{\rm f, est} =  \lbrack E_{\rm t, 0}^{\rm est}, E_{\rm t, 1}^{\rm est}, \cdots, E_{{\rm t}, N-1}^{\rm est} \rbrack^{\rm T} \in \mathbb{C}^{N \times 1}$. The deconvolution output-vector of $\mathbf{y}_{\rm If}$ and $\hat{\mathbf{f}}_{\rm p}$ is
\begin{align}
\label{eq:deconv_out_2}
\begin{split}
\mathbf{y}_{\rm If}
&=
\mathbf{D}_{N}
\underbrace
{
\Big\{
\big\{ 
{\rm diag} 
\big\{
g_{{\rm p}, n}
\big\}_{n = 0}^{N-1}
\big\}^{-1}
\Big\}
}_{\mathbf{G}_{\rm p}} 
\mathbf{D}_{N}^{\rm H}\mathbf{y}_{\rm f}\\
&=
\mathbf{D}_{N} 
\Big\{
\frac{1}{\alpha} 
\big\{ 
\underbrace{
{\rm diag} 
\lbrace 
p_{n} 
\rbrace_{n = 0}^{N-1}
}_{\triangleq \mathbf{\Phi}_{{\rm t}, \gamma}}
+ 
\underbrace{
{\rm diag} 
\lbrace 
E_{{\rm t}, n}^{\rm est}
\rbrace_{n = 0}^{N-1}
}_{\triangleq \mathbf{E}_{\rm t, est}} \big\}^{-1}\Big\} 
\mathbf{D}_{N}^{\rm H}
\mathbf{y}_{\rm f}\\
&\stackrel{\text{(a)}}{=}
\mathbf{D}_{N} 
\Big\{ 
\frac{1}{\alpha}
\big\{
\mathbf{\Phi}_{{\rm t}, \gamma}^{-1} 
- 
\mathbf{\Phi}_{{\rm t}, \gamma}^{-1}  
\big\{ 
\mathbf{\Phi}_{{\rm t}, \gamma}^{-1} 
+ 
\mathbf{E}_{\rm t, est}^{-1}
\big\}^{-1} 
\mathbf{\Phi}_{{\rm t}, \gamma}^{-1}
\big\}
\Big\} 
\mathbf{D}_{N}^{\rm H}
\mathbf{y}_{\rm f}\\
&\stackrel{\text{(b)}}{=} 
\frac{1}
{\alpha} 
\mathbf{D}_{N} 
\Big\{ 
\mathbf{I}_{N} - 
\big\{ 
\mathbf{\Phi}_{{\rm t}, \gamma} 
+ 
\mathbf{E}_{\rm t, est}
\big\}^{-1} 
\mathbf{E}_{\rm t, est}
\Big\}
\mathbf{\Phi}_{{\rm t}, \gamma}^{-1} 
\mathbf{D}_{N}^{\rm H}
\mathbf{y}_{\rm f}\\
&=
\frac{1}{\alpha}
\mathbf{D}_{N}
\big\{
\mathbf{I}_{N}
-
\alpha 
\mathbf{G}_{\rm p}
\mathbf{E}_{\rm t, est}
\big\}
\mathbf{y}_{\rm d}\\
&=
\underbrace{
\frac{1}{\alpha}
\mathbf{H}_{\rm f}
}_{\triangleq \mathbf{H}_{\rm If}}
\mathbf{x}_{\rm f}
+
\bar{\mathbf{\Upsilon}}
\mathbf{H}_{\rm If}
\mathbf{x}_{\rm f}
+
\bar{\mathbf{z}}_{\rm f}
=
\lbrace
\mathbf{I}
+
\bar{\mathbf{\Upsilon}}
\rbrace
\mathbf{H}_{\rm If}\mathbf{x}_{\rm f}
+
\bar{\mathbf{z}}_{\rm f},
\end{split}
\end{align}
where (a) and (b) follow from the matrix identity $(\mathbf{A} + \mathbf{B})^{-1} = \mathbf{A}^{-1} - \mathbf{A}^{-1}(\mathbf{A}^{-1} + \mathbf{B}^{-1})^{-1} \mathbf{A}^{-1}$ and $(\mathbf{A}^{-1} + \mathbf{B}^{-1})^{-1} = \mathbf{A}(\mathbf{A} + \mathbf{B})^{-1} \mathbf{B}$, respectively;  
\begin{align}
\label{eq:y_d}
\begin{split}
\mathbf{y}_{\rm d}
&=
\mathbf{\Phi}_{{\rm t}, \gamma}^{-1}
\mathbf{D}_{N}^{\rm H}\mathbf{y}_{\rm f}
=
\mathbf{\Phi}_{{\rm t}, \gamma}^{-1}
\mathbf{D}_{N}^{\rm H} 
\Big\{ 
\mathbf{F}_{\gamma}\mathbf{x}_{\rm f} 
+ 
\mathbf{E}_{\gamma}\mathbf{x}_{\rm f} 
+ 
\mathbf{z}_{\rm f}
\Big\}\\
&=
\mathbf{\Phi}_{{\rm t}, \gamma}^{-1}
\mathbf{D}_{N}^{\rm H} 
\Big\{ 
\mathbf{\Phi}_{{\rm f}, \gamma}
\mathbf{H}_{\rm f}\mathbf{x}_{\rm f} 
+ 
\underbrace{
\lbrace
\mathbf{\Phi}_{\rm f}
-
\mathbf{\Phi}_{\rm f, \gamma}
\rbrace
}_{\triangleq \tilde{\mathbf{\Phi}}_{\rm f, \gamma}}
\mathbf{H}_{\rm f}\mathbf{x}_{\rm f} 
+ 
\mathbf{z}_{\rm f}
\Big\}\\
&\stackrel{\text{(c)}}{=}
\mathbf{\Phi}_{{\rm t}, \gamma}^{-1}
\mathbf{D}_{N}^{\rm H}
\Big\{ 
\mathbf{D}_{N}
\mathbf{\Lambda}_{\Phi}
\mathbf{D}_{N}^{\rm H}
\mathbf{H}_{\rm f}\mathbf{x}_{\rm f} 
+ 
\mathbf{D}_{N}
\tilde{\mathbf{\Lambda}}_{\Phi}
\mathbf{D}_{N}^{\rm H}
\mathbf{H}_{\rm f}\mathbf{x}_{\rm f} 
+ 
\mathbf{z}_{\rm f}
\Big\}\\
&=
\underbrace{
\mathbf{\Phi}_{{\rm t}, \gamma}^{-1} \mathbf{\Lambda}_{\Phi}
}_{\mathbf{I}_{N}}
\mathbf{D}_{N}^{\rm H}
\mathbf{H}_{\rm f}\mathbf{x}_{\rm f} 
+
\mathbf{\Phi}_{{\rm t}, \gamma}^{-1} \tilde{\mathbf{\Lambda}}_{\Phi}
\mathbf{D}_{N}^{\rm H}
\mathbf{H}_{\rm f}\mathbf{x}_{\rm f}
+
\mathbf{\Phi}_{{\rm t}, \gamma}^{-1}
\mathbf{D}_{N}^{\rm H}\mathbf{z}_{\rm f}\\
&=
\mathbf{D}_{N}^{\rm H}
\mathbf{H}_{\rm f}\mathbf{x}_{\rm f} 
+
\mathbf{\Phi}_{{\rm t}, \gamma}^{-1} \tilde{\mathbf{\Lambda}}_{\Phi}
\mathbf{D}_{N}^{\rm H}
\mathbf{H}_{\rm f}\mathbf{x}_{\rm f}
+
\mathbf{\Phi}_{{\rm t}, \gamma}^{-1}
\mathbf{D}_{N}^{\rm H}\mathbf{z}_{\rm f},
\end{split}	
\end{align}
\begin{align}
\label{eq:Upsilon_bar}
\begin{split}
\bar{\mathbf{\Upsilon}}
&=
\mathbf{D}_{N}
\Big\{
\big\{
\mathbf{I}_{N}
-
\alpha 
\mathbf{G}_{\rm p}
\mathbf{E}_{\rm t, est}
\big\}
\Big\{
\mathbf{\Phi}_{{\rm t}, \gamma}^{-1} \tilde{\mathbf{\Lambda}}_{\Phi}
\Big\}
-
\alpha 
\mathbf{G}_{\rm p}
\mathbf{E}_{\rm t, est}
\Big\}
\mathbf{D}_{N}^{\rm H}
\\
&\stackrel{\text{(d)}}{=}
\mathbf{D}_{N}
\Big\{
\big\{
\mathbf{I}_{N}
-
\alpha 
\mathbf{G}_{\rm p}
\mathbf{E}_{\rm t, est}
\big\}
\Big\{ 
\mathbf{\Phi}_{{\rm t}, \gamma}^{-1} 
\mathbf{E}_{\rm t, app}
\Big\}
-
\alpha 
\mathbf{G}_{\rm p}
\mathbf{E}_{\rm t, est}
\Big\}
\mathbf{D}_{N}^{\rm H}
\\
&=
\mathbf{D}_{N}
\Big\{
\underbrace{
\big\{
\mathbf{I}_{N}
-
\alpha 
\mathbf{G}_{\rm p}
\mathbf{E}_{\rm t, est}
\big\}
\mathbf{\Phi}_{{\rm t}, \gamma}^{-1} 
}_{\alpha\mathbf{G}_{\rm p}}
\mathbf{E}_{\rm t, app}
-
\alpha 
\mathbf{G}_{\rm p}
\mathbf{E}_{\rm t, est}
\Big\}
\mathbf{D}_{N}^{\rm H}
\\
&=
\alpha	
\mathbf{D}_{N}
\mathbf{G}_{\rm p} 
\underbrace{
\Big\{
\mathbf{E}_{\rm t, app}
- 
\mathbf{E}_{\rm t, est}
\Big\}
}_{- \mathbf{E}_{\rm t, eff}}
\mathbf{D}_{N}^{\rm H}
\\
&=
-
\alpha
\mathbf{D}_{N}
\mathbf{G}_{\rm p}
\mathbf{E}_{\rm t, eff}
\mathbf{D}_{N}^{\rm H}
=
-
\mathbf{\Upsilon}.
\end{split}
\end{align}
In~(\ref{eq:y_d}) and (\ref{eq:Upsilon_bar}), (c) and (d) follow from Lemma 2 ($\mathbf{\Lambda}_{\Phi}= \mathbf{\Phi}_{{\rm t}, \gamma}$, $\tilde{\mathbf{\Lambda}}_{\Phi} = \mathbf{E}_{\rm t, app}$); $\mathbf{E}_{\rm t, app}$ is the diagonal matrix with entries from $\mathbf{e}_{\rm t, app}$ on its main diagonal.
\begin{align}
\label{eq:z_f_bar}
\begin{split}
\bar{\mathbf{z}}_{\rm f}
&=
\frac{1}{\alpha}
\mathbf{D}_{N}
\Big\{
\lbrace
\mathbf{I}_{N}
-
\alpha
\mathbf{G}_{\rm p}
\mathbf{E}_{\rm t, est}
\rbrace
\mathbf{\Phi}_{{\rm t}, \gamma}^{-1}
\Big\}
\mathbf{D}_{N}^{\rm H}\mathbf{z}_{\rm f}\\
&=
\frac{1}{\alpha}
\mathbf{D}_{N}
\Big\{
\alpha
\mathbf{G}_{\rm p}
\Big\}
\mathbf{D}_{N}^{\rm H}
\mathbf{z}_{\rm f}\\
&=
\mathbf{D}_{N}
\mathbf{G}_{\rm p}
\mathbf{D}_{N}^{\rm H}
\mathbf{z}_{\rm f}.\\	
\end{split}
\end{align}

\section{LMMSE Estimator for ICI-Free Channel}
\label{append_sec:LMMSE_IFC}
The LMMSE estimator for ICI-free channel vector $\tilde{\mathbf{h}}_{\rm If}$ is defined as
\begin{align}
\label{append:V_LMMSE}
{\mathbf{V}}	_{\mathsf{lmmse}} 
=
{\mathbf{R}}	_{\tilde{\mathbf{h}}\tilde{\mathbf{y}}}
{\mathbf{R}}_{\tilde{\mathbf{y}}\tilde{\mathbf{y}}}^{-1},	
\end{align}
where ${\mathbf{R}}_{\tilde{\mathbf{h}}\tilde{\mathbf{y}}} = \mathbb{E}\lbrace \tilde{\mathbf{h}}_{\rm If}\tilde{\mathbf{y}}_{\rm If}^{\rm H}\rbrace$ is the cross-covariance matrix between $\tilde{\mathbf{h}}_{\rm If}$ and $\tilde{\mathbf{y}}_{\rm If}$; ${\mathbf{R}}_{\tilde{\mathbf{y}}\tilde{\mathbf{y}}} = \mathbb{E}\lbrace \tilde{\mathbf{y}}_{\rm If}\tilde{\mathbf{y}}_{\rm If}^{\rm H}\rbrace$ is the autocorrelation matrix of $\tilde{\mathbf{y}}_{\rm If}$. Based on (\ref{eq:y_tilde}), the ${\mathbf{R}}_{\tilde{\mathbf{h}}\tilde{\mathbf{y}}}$ and ${\mathbf{R}}_{\tilde{\mathbf{y}}\tilde{\mathbf{y}}}$ are represented as (\ref{append:R_hy_IFC}) and (\ref{append:R_yy_IFC}), respectively, in the bottom of this page, where $\bar{G} \triangleq \frac{1}{N} \sum_{n=0}^{N-1} \lbrace \| 1/g_{{\rm p}, n} \|_{2}^{2} \rbrace$ is the mean of absolute-squared diagonal coefficients in $\mathbf{G}_{\rm p}$.

\section{Autocorrelation Coefficients of $\mathbf{R}_{\mathbf{pp}}^{\gamma}$ and $\mathbf{R}_{\mathbf{ici}}^{\gamma}$}
\label{append_sec:R_pp_and_R_ici}
The autocorrelation coefficient $R_{k,\ell}$ in $\mathbf{R}_{\mathbf{pp}}$ is~\cite{den07}
\addtocounter{equation}{2}
\begin{align}
\label{append:R_pp_entry}
\begin{split}
	R_{k,\ell}
	&=
	\mathbb{E} \lbrace P_{k}P_{\ell}^{*} \rbrace\\ 
	&=
	{\frac{1}{N^2}} \mathbb{E} \Bigg\{{\sum_{m=0}^{N-1} \sum_{n=0}^{N-1}} {e^{j (\phi_{m}-\phi_{n})}}{e^{-j{\frac{2 \pi}{N}}(mk-n\ell)}} \Bigg\}\\
	&=
	{\frac{1}{N^2}} {\sum_{m=0}^{N-1} \sum_{n=0}^{N-1}} \mathbb{E} \lbrace{e^{j \Delta \phi_{m,n}}} \rbrace {e^{-j{\frac{2 \pi}{N}}(mk-n\ell)}}\\ 
	&\stackrel{\text{(a)}}{=}
	{\frac{1}{N^2}} {\sum_{m=0}^{N-1} \sum_{n=0}^{N-1}} \underbrace{\lbrace{e^{-\pi \beta |m-n| T_{\rm s}}} \rbrace}_{\triangleq  \psi_{m,n}} {e^{-j{\frac{2 \pi}{N}}(mk-n\ell)}}  
\end{split}
\end{align}
where (a) is determined using the moment generating function of $\Delta \phi_{m,n}$.  The autocorrelation coefficient $R_{k,\ell}^{\rm{ici}}$ of $\mathbf{R}_{\mathbf{ici}}^{\gamma}$ is
\begin{align}
\label{append:R_ici_entry}
\begin{split}
R_{k,\ell}^{\rm ici} = {\sum_{i = N_{\rm a} + k}^{N-N_{\rm b} + k}} {R_{i, i+(\ell-k)}}
\end{split}
\end{align} 
\end{appendices}


\ifCLASSOPTIONcaptionsoff
  \newpage
\fi

\renewcommand{\baselinestretch}{1.0}
\bibliographystyle{IEEEtran}
\bibliography{reference_21.bib} 
\input{acronym_21.acro}

\end{document}